\documentclass[11pt,reqno]{article}

\usepackage{amsmath, amsthm, amsfonts, amssymb, amsxtra, fullpage, color, hyperref, url}
\usepackage{graphicx}
\usepackage[inline]{enumitem}

\numberwithin{equation}{section}

\newcommand{\R}{{\mathbb R}}
\newcommand{\C}{{\mathbb C}}
\newcommand{\compC}{{\mathbb C}}
\newcommand{\realR}{{\mathbb R}}

\renewcommand{\d}{\partial}
\renewcommand{\Re}{{\operatorname{Re\,}}}
\renewcommand{\Im}{{\operatorname{Im\,}}}

\newcommand{\F}{{\mathcal F}}
\newcommand{\M}{{\mathcal M}}

\newcommand{\diag}{{\operatorname{diag}}}
\newcommand{\dist}{{\operatorname{dist}}}
\newcommand{\Ai}{{\operatorname{Ai}}}

\newcommand{\crit}{{\operatorname{cr}}}

\newcommand{\be}{\beta}
\newcommand{\ga}{\gamma}

\newcommand{\la}{\lambda}

\newcommand{\De}{\Delta}

\newcommand{\sg}{\sigma}
\newcommand{\Sg}{\Sigma}

\newcommand{\Om}{\Omega}

\newcommand{\z}{\zeta}

\newcommand{\acal}{{\mathcal A}}

\newcommand{\bigO}{{\mathcal O}}

\newcommand{\qcal}{{\mathcal Q}}
\newcommand{\lcal}{{\mathcal L}}

\newcommand{\Painleve}{Painlev\'{e}}
\newcommand{\Veto}{Vet\H{o}}
\newcommand{\resp}{resp.}

\DeclareMathOperator{\Tr}{Tr}
\DeclareMathOperator{\sgn}{sgn}
\DeclareMathOperator{\scaled}{scaled}
\DeclareMathOperator{\tac}{tac}
\DeclareMathOperator{\HM}{HM}
\DeclareMathOperator{\tMM}{2MM}
\newtheorem{theo}{{\sc \bf Theorem}}[section]
\newtheorem{cor}[theo]{{\sc \bf Corollary}}
\newtheorem{lem}[theo]{{\sc \bf Lemma}}
\newtheorem{prop}[theo]{{\sc \bf Proposition}}
\theoremstyle{definition}
\newtheorem{RHP}[theo]{{\sc \bf Riemann--Hilbert Problem}}

\theoremstyle{remark}

\newtheorem{rem}{Remark}[section]

\theoremstyle{definition}

\title{Two Lax systems for the \Painleve\ II equation, and two related kernels in random matrix theory}

\author{Karl Liechty \thanks{Department of Mathematical Sciences, DePaul University, Chicago, IL, 
\href{mailto:kliechty@depaul.edu}{\nolinkurl{kliechty@depaul.edu}}
 \newline
    Supported by DePaul University College of Science and Health Summer Research Grant, an AMS--Simons travel grant, and a grant from the Simons Foundation (\#357872, Karl Liechty)
} \and Dong Wang\thanks{Department of Mathematics, National University of Singapore, Singapore, 119076, \href{mailto:matwd@nus.edu.sg}{\nolinkurl{matwd@nus.edu.sg}} \newline
    Supported partially by the startup grant R-146-000-164-133}}
    
\begin{document}

\maketitle

\begin{abstract}
We consider two Lax systems for the homogeneous \Painleve\ II equation: one of size $2\times 2$  studied by Flaschka and Newell in the early 1980's, and one of size $4\times 4$ introduced by Delvaux--Kuijlaars--Zhang and Duits--Geudens in the early 2010's. We prove that solutions to the $4\times 4$ system can be derived from those to the $2\times 2$ system via an integral transform, and consequently relate the Stokes multipliers for the two systems. As corollaries we are able to express two kernels for determinantal processes as contour integrals involving the Flaschka--Newell Lax system: the tacnode kernel arising in models of nonintersecting paths, and a critical kernel arising in a two-matrix model.
\end{abstract}

\section{Introduction and statement of results}

The homogeneous \Painleve\ II equation (PII) is the second order nonlinear ODE
\begin{equation}\label{int:1}
  y''=xy+2y^3.
\end{equation} 
Despite its unassuming form, its solutions, known as the \emph{\Painleve\ transcendents}, appear in exact solutions of many models in mathematical physics. For example, one particular solution to \eqref{int:1} is the one satisfying the boundary condition 
\begin{equation}\label{par:0}
  q(\sg) \sim \Ai(\sg) \textrm{ as } \sg\to+\infty,
\end{equation}
where $\Ai$ is the Airy function. This solution is known as the \emph{Hastings--McLeod solution} \cite{Hastings-McLeod80}. It is particularly important in random matrix theory, for it defines the celebrated Tracy--Widom distributions which describe the generic soft edge behavior of random matrices from  orthogonal-, unitary-, or symplectic-invariant ensembles \cite{Tracy-Widom94}, \cite{Tracy-Widom96}.

The PII equation \eqref{int:1} is an \emph{integrable equation}, and its integrability is characterized by the existence of {\it Lax pairs}. 
A Lax pair, or more generally a Lax system, is a system of overdetermined linear differential equations whose compatibility implies a nonlinear equation. Let $\Psi = \Psi(z_1, \dotsc, z_r)$ be an $n \times n$ matrix-valued function with variables $z_1, \dotsc, z_r$. Let 
\begin{equation}\label{eq:def_Lax_system}
  \frac{\partial \Psi}{\partial z_1} = A_1 \Psi, \quad  \dotsc, \quad \frac{\partial \Psi}{\partial z_r} = A_r \Psi,
\end{equation}
be an (overdetermined) system of differential equations satisfied by $\Psi$ with $n\times n$ coefficient matrices $A_1, \dotsc, A_r$.
For the overdetermined differential equations to have nontrivial solutions, we need the compatibility among $A_1, \dotsc, A_r$, the \emph{Frobenius compatibility conditions}, sometimes called \emph{zero-curvature relations}:
\begin{equation} \label{eq:zero_curvature}
  \frac{\partial A_i}{\partial z_j} - \frac{\partial A_j}{\partial z_i} + [A_i, A_j] = 0, \quad \text{for all $i, j = 1, \dotsc, r$}.
\end{equation}
The Frobenius compatibility conditions are in general nonlinear differential equations for the entries of $A_j$, and we call the system \eqref{eq:def_Lax_system} the Lax system for the nonlinear equation(s)  \eqref{eq:zero_curvature}. In the most common cases $r = 2$ and we call the system \eqref{eq:def_Lax_system} a Lax pair, but we may also consider the general case $r \geq 2$.

\begin{rem}
The term Lax pair originates with the work of Peter Lax in the late 1960's \cite{Lax68}, in which he used the compatibility of a pair of linear differential equations to study a nonlinear {\it partial} differential equation. In the problem considered by Lax the evolution of the time variable gives an {\it isospectral} deformation of the linear operator. On the other hand, \Painleve\ equations represent {\it isomonodromic} deformations of the analogous linear equations with respect to the singularities, i.e., the monodromy data is invariant as the argument of the (fixed) \Painleve\ function changes, and the isomonodromic relations are expressed in the same form of Lax pairs \cite[Chapter 4]{Fokas-Its-Kapaev-Novokshenov06}. The idea of representing the \Painleve\ equations as isomonodromy deformations of a system of linear equations is nearly as old as the \Painleve\ equations themselves, dating back to the work of Fuchs \cite{Fuchs07} and later Garnier \cite{Garnier12}. Therefore it may be more appropriate to call the overdetermined systems \eqref{eq:def_Lax_system} and \eqref{int:7} {\it Garnier--Fuchs pairs/systems} rather than Lax pairs/systems. Such terminology can be found in the literature, see \cite{Joshi-Kitaev-Treharne07} and \cite{Joshi-Kitaev-Treharne09}. However, the phrase Lax pair is much more abundant in the literature and this is the nomenclature we use, following the terminology of \cite{Delvaux-Kuijlaars-Zhang11}, \cite{Duits-Geudens13}, \cite{Delvaux13}, and \cite{Fokas-Its-Kapaev-Novokshenov06}.
\end{rem}

Nonlinear differential equations which possess a Lax system representation are in some sense integrable, although they can be rather complicated. All of the \Painleve\ equations, including \eqref{int:1}, can be represented by Lax pairs/systems \cite{Fokas-Its-Kapaev-Novokshenov06}. However, the construction of Lax pairs/systems for a given \Painleve\ equation is far from trivial, and the relations between different Lax pair/systems for a \Painleve\ equation deserve investigation for their own sake. In this paper we demonstrate the relation between one classical Lax pair and a recently discovered Lax system for the PII equation \eqref{int:1}.
However, the main motivation of our paper is not purely theoretical, but is driven by the appearance of these Lax systems in random matrix theory and related problems. The classical Lax pair and the new Lax system are both related to random matrix theory, but in quite different aspects.

\subsection{The Flaschka--Newell Lax pair for PII} \label{subsec:Flaschka-Newell_Lax}
First we present a classical Lax pair for \eqref{int:1}, found by Flaschka and Newell \cite{Flaschka-Newell80}.

\begin{rem}
  The Flaschka--Newell Lax pair was originally presented for the general PII equation which has a free parameter (see Section \ref{subsec:outlook}), and we only present it for the homogeneous case \eqref{int:1}. A different Lax pair for PII was found by Jimbo and Miwa around the same time \cite{Jimbo-Miwa81} (with a precursor in \cite{Garnier12}), but in the homogeneous case the Jimbo--Miwa Lax pair can be reduced to the Flaschka--Newell one \cite[Section 4.2]{Fokas-Its-Kapaev-Novokshenov06}. Other Lax pairs associated to PII have been found by Harnad, Tracy, and Widom in \cite{Harnad-Tracy-Widom93} (of size $2\times 2$) and by Joshi, Kitaev, and Treharne in \cite{Joshi-Kitaev-Treharne09} (of size $3\times 3$). The equivalence among these Lax pairs is discussed in \cite{Joshi-Kitaev-Treharne09}.
\end{rem}

Let $\Phi = \Phi(\zeta; \sigma)$ be a $2 \times 2$ matrix-valued function with variables $\zeta$ and $\sigma$ which satisfies the overdetermined equations
\begin{subequations}\label{int:2}
  \begin{align}
    \frac{\d}{\d \z}\Phi(\z; \sg) = {}& A\Phi(\z; \sg), \label{int:2a} \\
    \frac{\d}{\d \sg}\Phi(\z; \sg) = {}& B\Phi(\z; \sg), \label{int:2b}
  \end{align}
\end{subequations} 
where
\begin{equation}\label{int:3}
  A=\begin{pmatrix} -4i\z^2-i(\sg+2q^2) & 4\z q+2ir \\ 4\z q -2ir  & 4i\z^2+i(\sg+2q^2) \end{pmatrix}, \quad
  B= \begin{pmatrix} -i\z & q \\ q & i\z \end{pmatrix},
\end{equation} 
and $q$ and $r$ are parameters which may depend on $\sigma$.
It is an amiable exercise to show that the compatibility of the two equations in \eqref{int:2} is reduced to the fact that $q \equiv q(\sg)$ solves the \Painleve\ equation \eqref{int:1}, and the parameter $r$ in \eqref{int:3} is $r\equiv r(\sg)=q'(\sg).$

It is known that all solutions to the \eqref{int:1} are meromorphic, so if we choose $q\equiv q(\sg)$ to be any particular solution to \eqref{int:1} and take $r\equiv q'(\sg)$, then the system \eqref{int:2} is solvable provided $\sg$ is not a pole of the chosen PII transcendent. Notice then that, given a particular solution $q(\sg)$ and fixing $\sg$ that is not a pole of this solution, we can find a solution to the overdetermined equation \eqref{int:2} using only \eqref{int:2a}, given proper initial conditions. Thus below we concentrate on \eqref{int:2a} when we talk about the solutions to \eqref{int:2}, where $q(\sg)$ is a fixed solution to \eqref{int:1}, $r(\sg) = q'(\sg)$, and $\sg$ is a constant that is not a pole of $q$. In some formulas in this paper, we suppress the dependence on $\sigma$ if it is treated as a constant.

Since $\infty$ is the only singular point of $A$, and
\begin{equation}
  A = (I + \bigO(\zeta^{-1}))
  \begin{pmatrix}
    -4i \zeta^2 & 0 \\
    0 & 4i \zeta^2
  \end{pmatrix},
  \quad \text{as $\zeta \to \infty$},
\end{equation}
it is natural to construct the fundamental solution $\Phi$ such that 
\begin{equation} \label{eq:fundamental_eq}
  \Phi(\zeta) = (I + \bigO(\zeta^{-1}))
  \begin{pmatrix}
    e^{-\frac{4}{3}i \zeta^3 - i \sigma \zeta} & 0 \\
    0 & e^{\frac{4}{3}i \zeta^3 + i \sigma \zeta}
  \end{pmatrix},
  \quad \text{as $\zeta \to \infty$}. 
\end{equation}
 But $\infty$ is an irregular singularity of $A$, so the Stokes phenomenon allows us only to consider the solution $\Phi$ that satisfies \eqref{eq:fundamental_eq} sectorally. For a rigorous version of the heuristic argument above see \cite[Section 5.0]{Fokas-Its-Kapaev-Novokshenov06}.

 For $j=0, 1,\dots, 5$,  define the sectors (see Figure \ref{fig:2x2jump}),
\begin{equation}\label{eq:Sj_sectors}
S_j = \left\{ z\in \C : -\frac{\pi}{6}+\frac{j\pi}{3} < \arg z < \frac{\pi}{6}+\frac{j\pi}{3}\right\}.
\end{equation}
Their boundaries are the rays with outward orientation
 \begin{equation} \label{eq:defn_Sigma_k}
   \Sigma_k = \left\{ t e^{(k - 1/2) \frac{i\pi}{3}} \mid t \in [0, \infty) \right\}, \quad k = 0, \dotsc, 5.
 \end{equation}
Then there are fundamental solutions $\Psi^{(0)}, \dotsc, \Psi^{(5)}$ to \eqref{int:2a} such that $\Psi^{(j)}$ satisfies the boundary condition \eqref{eq:fundamental_eq} in sector $S_j$. 
Of course the solution space to \eqref{int:2a} is two dimensional and so there are linear relations between the solutions $\Psi^{(0)}, \dotsc, \Psi^{(5)}$. These relations depend on the particular Painlev\'{e} transcendent appearing in the coefficient matrices $A$ and $B$ in \eqref{int:3}, and can be described in the following way \cite[Section 5.0]{Fokas-Its-Kapaev-Novokshenov06}.

For each PII solution $q(\sg)$ to \eqref{int:1}, there is a triple of complex numbers $(t_1, t_2, t_3)$ satisfying the relation
\begin{equation} \label{eq:triple_number}
  t_1 + t_2 + t_3 + t_1 t_2 t_3 = 0,
\end{equation}
such that the fundamental solutions $\Psi^{(k)}$ associated with $q(\sigma)$ satisfy
\begin{equation} \label{eq:jump_condition_2x2}
    \Psi^{(k)} = \Psi^{(k - 1)} J_k, \quad k =  0, \dotsc, 5, \quad \text{with $J_k$ shown in Figure \ref{fig:2x2jump} and $\Psi^{(-1)} := \Psi^{(5)}$.}
\end{equation}
The jump matrices $J_k$ are called the {\it Stokes matrices}, and the numbers $t_1, t_2, t_3$ are called the {\it Stokes multipliers} corresponding to the given PII solution $q(\sg)$. Remarkably, each triple $(t_1, t_2, t_3)$ of Stokes multipliers satisfying \eqref{eq:triple_number} corresponds uniquely to a PII solution, and so the solutions to PII are parametrized by the surface \eqref{eq:triple_number}. Thus in order to specify a solution to PII, it is enough to specify the Stokes multipliers $(t_1, t_2, t_3)$, see \cite[Proposition 5.1]{Fokas-Its-Kapaev-Novokshenov06}. In Figure \ref{fig:2x2jump} we show the rays, sectors, and the jump matrices $J_k$.
\begin{figure}[htb]
  \begin{minipage}[t]{0.36\linewidth}
    \centering
    \includegraphics{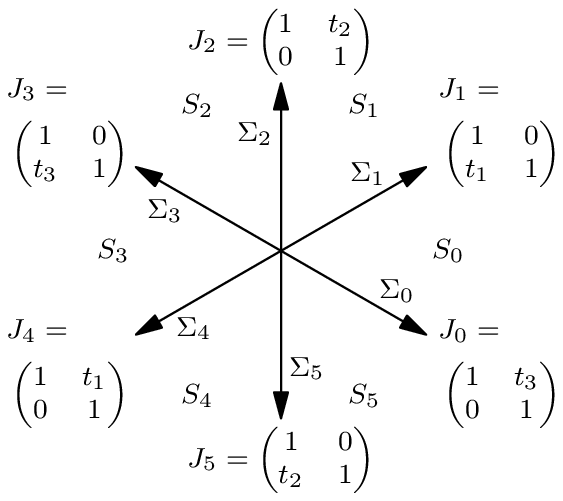}
    \caption{Rays $\Sigma_k$, sectors $S_k$, and jump matrices $J_k$ placed on $\Sigma_k$ for $k = 0, \dotsc, 5$.}
    \label{fig:2x2jump}
  \end{minipage}
  \hspace{\stretch{1}}
  \begin{minipage}[t]{0.62\linewidth}
    \centering
    \includegraphics{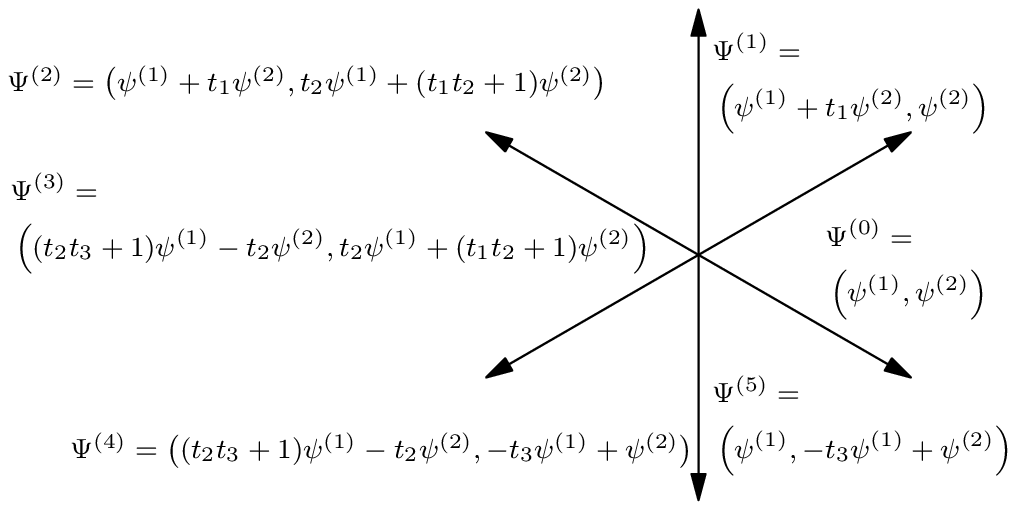}
    \caption{The formulas of $\Psi^{(0)}, \dotsc, \Psi^{(5)}$ expressed in $\psi^{(1)}$ and $\psi^{(2)}$.}
    \label{fig:2x2Phi}
  \end{minipage}
\end{figure}

For a given set of Stokes multipliers, the jump properties \eqref{eq:jump_condition_2x2} determine any of the fundamental solutions in terms of the solution $\Psi^{(0)}$. Indeed if we denote
\begin{equation} \label{eq:defn_psi^1_psi^2}
  \Psi^{(0)}(\zeta; \sigma) = \left( \psi^{(1)}(\zeta; \sigma), \psi^{(2)}(\zeta; \sigma) \right),
\end{equation}
where $\psi^{(1)}$ and $\psi^{(2)}$ are two $2$-dimensional vector-valued functions defined on the whole complex plane, then the other $\Psi^{(k)}$ are expressed in $\psi^{(1)}$ and $\psi^{(2)}$ as in Figure \ref{fig:2x2Phi}.
The asymptotics of the columns of $\Psi^{(k)}$ are summarized below (with $\delta$ being any small positive constant):
\begin{align}
  \left.
  \begin{aligned}
    \psi^{(1)}(\zeta) & \\
    \psi^{(1)}(\zeta) + t_1 \psi^{(2)}(\zeta) & \\
    (t_2 t_3 + 1)\psi^{(1)}(\zeta) - t_2 \psi^{(2)}(\zeta) & 
  \end{aligned}
                                                             \right\} = {}& (I + \bigO(\zeta^{-1}))
                                                                            \begin{pmatrix}
                                                                              e^{-\frac{4}{3}i \zeta^3 - i\sigma \zeta} \\
                                                                              0
                                                                            \end{pmatrix} &&
                                                                                             \begin{cases}
                                                                                               \text{if $\arg(\zeta) \in (-\frac{2\pi}{3} + \delta, \frac{\pi}{3} - \delta)$}, \\
                                                                                               \text{if $\arg(\zeta) \in (\delta, \pi - \delta)$}, \label{eq:asy_psi^1+t_1psi^2} \\
                                                                                               \text{if $\arg(\zeta) \in (\frac{2\pi}{3} + \delta, \frac{5\pi}{3} - \delta)$},
                                                                                             \end{cases} \\
  \left.
  \begin{aligned}
    \psi^{(2)}(\zeta) & \\
    -t_3 \psi^{(1)} + \psi^{(2)}(\zeta) & \\
    t_2 \psi^{(1)} + (t_1 t_2 + 1)\psi^{(2)}(\zeta) &
  \end{aligned}
                                                      \right\} = {}& (I + \bigO(\zeta^{-1}))
                                                                     \begin{pmatrix}
                                                                       0 \\
                                                                       e^{\frac{4}{3}i \zeta^3 + i\sigma \zeta}
                                                                     \end{pmatrix} &&
                                                                                      \begin{cases}
                                                                                        \text{if $\arg(\zeta) \in (-\frac{\pi}{3} + \delta, \frac{2\pi}{3} - \delta)$}, \\
                                                                                        \text{if $\arg(\zeta) \in (\pi + \delta, 2\pi - \delta)$}, \\
                                                                                        \text{if $\arg(\zeta) \in (\frac{\pi}{3} + \delta, \frac{4\pi}{3} - \delta)$}. 
                                                                                      \end{cases}
  \label{eq:asy_psi^2}
\end{align}

\subsubsection{Critical kernel in one-matrix model} \label{subsubsec:1MM}

As mentioned earlier, the Hastings--McLeod solution to \eqref{int:1}, the one satisfying \eqref{par:0}, is of special importance in random matrix theory. It is the solution to PII that corresponds to the Stokes multipliers $(t_1, t_2, t_3) = (1, 0, -1)$, and it is well established that it has no poles on the real line. Thus the solution $\Psi^{(0)}\equiv \Psi^{(0)}(\z;\sg)$ exists for any real $\sg$ \cite[Section 11.7]{Fokas-Its-Kapaev-Novokshenov06}.

Consider the one-matrix model given by the probability measure on the space of $n \times n$ Hermitian matrices $M$,
\begin{equation}
  \frac{1}{C_n} \exp(-n t \Tr V(M)) dM,
\end{equation}
where $V$ is the potential and $t > 0$ is a scaling factor. The eigenvalues of $M$ are a determinantal process that is characterized by a correlation kernel. In the case that $V(x) = x^4/4 - x^2$ and $n \to \infty$, the model is in a critical phase if $t = 1$. As $n \to \infty$, under the double scaling limit $t = 1 - (2n)^{-2/3} \sigma$, the correlation kernel at $u(n/4)^{-1/3}$ and $v(n/4)^{-1/3}$ converges to
\begin{equation}
  K_1^{\crit}(u,v; \sigma) = \frac{-\psi^{(1)}_1(u; \sg) \psi^{(1)}_2(v; \sg)+\psi^{(1)}_2(u; \sg) \psi^{(1)}_1(v; \sg)}{2\pi i (u-v)},
\end{equation}
where $\psi^{(1)}_1$ and $\psi^{(1)}_2$ are the two components of the 2-vector $\psi^{(1)}$ defined in \eqref{eq:defn_psi^1_psi^2}, see \cite{Bleher-Its03}. We use the notation $K_1^{\crit}$ to emphasize that this kernel arises in a 1-matrix model and to differentiate it from the kernel \eqref{critkernel} which arises in a 2-matrix model, which we denote $K_2^{\crit}$.   Note that although we only state the limiting correlation kernel for a very special potential function, the convergence to $K^{\crit}_1$ holds for a large class of potentials that have a quadratic interior critical point. See \cite{Claeys-Kuijlaars06} for the universality of the limiting kernel $K^{\crit}_1$.

Finally we remark that if we give the potential $V$ a logarithmic perturbation at $0$, i.e., let $V(x) = x^4/4 - x^2 - (2\alpha/n) \log \lvert x \rvert$, then the limiting kernel at $0$ is changed, and it is expressed by the Flaschka--Newell Lax pair for the Hastings--McLeod solution of the  \emph{inhomogeneous} PII equation. See \cite{Claeys-Kuijlaars-Vanlessen08} for detail, and also see Section \ref{subsec:outlook}.

\subsection{A $4\times 4$ Lax system for PII}

Now we introduce the other Lax system for the PII equation \eqref{int:1}, which was discovered recently by Delvaux, Kuijlaars, and Zhang in their study of non-intersecting Brownian motions \cite{Delvaux-Kuijlaars-Zhang11}, by Delvaux in the study of non-intersecting squared Bessel processes \cite{Delvaux13a}, and by Duits and Geudens in their study of the 2-matrix model \cite{Duits-Geudens13}, see also \cite{Delvaux13}, \cite{Kuijlaars14}. In its most general form this Lax system is a $4$-dimensional overdetermined differential system consisting of 16 equations. Here we consider a $4 \times 4$ matrix valued function $M = M(z, s_1, s_2, \tau)$, and the Lax system is
 \begin{subequations} \label{int:7}
   \begin{gather}
    \frac{\d}{\d z} M =U M, \label{int:7_a} \\ 
    \frac{\d}{\d s_1} M =V_1 M, \quad \frac{\d}{\d s_2} M =V_2 M, \quad \frac{\d}{\d \tau} M =W M. \label{int:7_b}
   \end{gather}
 \end{subequations}
The coefficient matrix $U$ is given by
\begin{equation}\label{int:4}
  U =
  \begin{pmatrix}
    U^{11} & U^{12} \\
    U^{21} & U^{22}
  \end{pmatrix},
\end{equation}
where each $U^{ij}$ is a $2 \times 2$ block, such that
\begin{equation}\label{int:5}
  \begin{gathered}
    U^{11} =
    \begin{pmatrix}
      \tau-s_1^2+\frac{u}{C} & \frac{\sqrt{r_2} q}{\ga C \sqrt{r_1}} \\
      -\ga \frac{\sqrt{r_1} q}{C\sqrt{r_2}} & -\tau +s_2^2-\frac{u}{C}
    \end{pmatrix}, \quad
    U^{12} =
    \begin{pmatrix}
      ir_1 & 0 \\
      0 & ir_2
    \end{pmatrix}, \quad
    U^{22} =
    \begin{pmatrix}
      \tau+s_1^2-\frac{u}{C} & \frac{\sqrt{r_1} q}{\ga\sqrt{r_2} C} \\
      -\ga \frac{\sqrt{r_2} q}{\sqrt{r_1} C} & -\tau-s_2^2+\frac{u}{C}
    \end{pmatrix}, \\
    U^{21} = i
    \begin{pmatrix}
      r_1 z -2s_1+\frac{s_1^4}{r_1}-\frac{2s_1^2 u}{r_1 C}+\frac{u^2-q^2}{r_1C^2} & \frac{\sqrt{r_1r_2}C(q'+uq)}{\ga}-\frac{(r_1^2s_2^2+r_2^2s_1^2)q}{\ga C (r_1r_2)^{3/2}} \\ \ga \sqrt{r_1r_2}C(q'+uq)-\frac{\ga(r_1^2s_2^2+r_2^2s_1^2)q}{C (r_1r_2)^{3/2}} & -r_2 z -2s_2+\frac{s_2^4}{r_2}-\frac{2s_2^2 u}{r_2 C}+\frac{u^2-q^2}{r_2C^2}
    \end{pmatrix}.
  \end{gathered}
\end{equation}
Here the numbers $r_1$ and $r_2$ are positive constants, and $C$, $\ga$, $q$, $q'$, and $u$ depend on $r_1, r_2, s_1, s_2, \tau$. We relegate the formulas for $V_1, V_2, W$ to Appendix \ref{sec:formulas_VVW}, since we do not use them in the rest of this paper. In the symmetric case $r_1 = r_2$ and $s_1 = s_2$, see also \cite[Section 5.3]{Delvaux13a}, \cite{Duits-Geudens13}, and in the $\tau = 0$ case, see also \cite[Section 5.2]{Delvaux-Kuijlaars-Zhang11}. By the compatibility of the overdetermined system, which is routine but laborious, see \cite[Section 6.5]{Delvaux13}, we derive
\begin{equation}\label{int:6}
  C=(r_1^{-2} +r_2^{-2})^{1/3}, \qquad \ga=\exp\left(\frac{8}{3}\frac{r_1^2-r_2^2}{(r_1^2+r_2^2)^2}\tau^3 - 4\frac{r_1s_1-r_2s_2}{r_1^2+r_2^2}\tau\right),
\end{equation}
and $q$ and $u$ are functions of
\begin{equation}\label{int:9}
  \sg:=\frac{2}{C} \left( \frac{s_1}{r_1} + \frac{s_2}{r_2} - \frac{2 \tau^2}{r_1^2+r_2^2} \right).
\end{equation}
Furthermore, $q = q(\sg)$ satisfies the PII equation \eqref{int:1}, $q' = q'(\sg)$ is the derivative with respect to $\sg$, and $u$ is the PII Hamiltonian
\begin{equation}\label{int:10}
  u(\sg):=q'(\sg)^2-q(\sg)^2-q(\sg)^4,
\end{equation}
which satisfies
\begin{equation}\label{int:11}
  u'(\sg)=-q(\sg)^2.
\end{equation}
Now as with the Lax pair \eqref{int:2}, we fix a particular solution $q(\sg)$ to PII and assume $\sg$ is not a pole of this solution. We can then solve the Lax system by \eqref{int:7_a} alone, with proper initial conditions.

\begin{rem}
  The authors of \cite{Delvaux-Kuijlaars-Zhang11}, \cite{Duits-Geudens13}, \cite{Delvaux13}, and \cite{Delvaux13a} introduced the Lax system \eqref{int:7} as a technical tool to study the tacnode Riemann--Hilbert problem (RHP), a $4 \times 4$ Riemann--Hilbert problem associated with the PII equation \eqref{int:1}. The tacnode RHP is only defined for the Hastings--McLeod solution to PII, but the Lax system is algebraic and the Frobenius compatibility conditions \eqref{eq:zero_curvature} are independent of boundary condition, so the Lax system exists for all solutions to the PII equation. From the Lax system we can construct an RHP that is associated with all solutions to the PII equation and thus generalize the tacnode RHP. See Riemann--Hilbert problem \ref{RHP:4x4} in Section \ref{subsec:tacnode_RHP} below.
\end{rem}

Since $\infty$ is the unique singular point of $U$, it is natural to put the boundary condition to the solution $M$ at $\infty$. The situation is a bit more complicated than for the $2\times 2$ Lax system, since infinity is, in the language of \cite{Fokas-Its-Kapaev-Novokshenov06}, a \emph{general irregular singular point} of the coefficient matrix $U$. Nonetheless, it is possible to transform the equation \eqref{int:7_a} into one with a regular singular point by means of an explicit change of variable, and then to derive the asymptotic structure of its solutions using the methods of \cite{Wasow87}. This asymptotic structure was derived by Duits and Geudens in \cite{Duits-Geudens13}. In order to describe it, we define the functions
\begin{equation}\label{eq:theta}
  \begin{aligned}
    \theta_1(z) = {}& \frac{2}{3} r_1 (-z)^{\frac{3}{2}} + 2 s_1 (-z)^{\frac{1}{2}}, & z \in {}& \compC \setminus [0, \infty),\\
    \theta_2(z) = {}& \frac{2}{3} r_2 z^{\frac{3}{2}} + 2 s_2 z^{\frac{1}{2}}, & z \in {}& \compC \setminus (-\infty, 0],
  \end{aligned}
\end{equation}
and then the $4$-dimensional vector-valued functions
\begin{equation} \label{eq:defn_v_1-v_4}
  \begin{aligned}
    v_1(z) = {}& \frac{1}{\sqrt{2}} e^{-\theta_1(z) + \tau z}
    \left( (-z)^{-\frac{1}{4}},
      0,
      -i(-z)^{\frac{1}{4}},
      0 \right)^T,
    &
    v_2(z) = {}& \frac{1}{\sqrt{2}} e^{-\theta_2(z) - \tau z}
    \left( 0,
      z^{-\frac{1}{4}},
      0,
      iz^{\frac{1}{4}} \right)^T,
    \\ v_3(z) = {}& \frac{1}{\sqrt{2}} e^{\theta_1(z) + \tau z}
    \left( -i(-z)^{-\frac{1}{4}},
      0,
      (-z)^{\frac{1}{4}},
      0 \right)^T,
    &
    v_4(z) = {}& \frac{1}{\sqrt{2}} e^{\theta_2(z) - \tau z}
    \left( 0,
      i z^{-\frac{1}{4}},
      0,
      z^{\frac{1}{4}} \right)^T,
  \end{aligned}
\end{equation}
and the matrix-valued function
\begin{equation} \label{eq:defn_A}
  \acal(z) := \Big( v_1(z), v_2(z), v_3(z), v_4(z) \Big).
\end{equation}
For the fractional powers in \eqref{eq:defn_v_1-v_4} we take the principal branches, so $\acal(z)$ has cuts on $\R_+$ and $\R_-$. More precisely, the functions $v_1(z)$ and $v_3(z)$ each have cuts on the positive real axis, and the functions $v_2(z)$ and $v_4(z)$ each have cuts on the negative real axis. We also define the function $\acal^+(z)$ to be the continuation of $\acal(z)$ from the upper half plane with a cut on the negative imaginary axis, and and $\acal^-(z)$ to be the continuation of $\acal(z)$ from the lower half plane with a cut on the positive imaginary axis. To be concrete, we denote
\begin{equation}
  \acal^{\pm}(z) = \Big( v^{\pm}_1(z), v^{\pm}_2(z), v^{\pm}_3(z), v^{\pm}_4(z) \Big),
\end{equation}
such that for all $j = 1, \dotsc, 4$, $v^{\pm}_j(z) = v_j(z)$ in $\compC_{\pm}$, and the branch cut for $v^{\pm}_j(z)$ is $\{ \mp i t \mid t \geq 0 \}$. 
If we denote by $v_{j}^+(z)$ (resp. $v_{j}^-(z)$) the limiting value of $v_j(z)$ from the upper (resp. lower) half-plane for $j=1,2,3,4$, then we have the following relations:
\begin{equation} \label{eq:v_jumps}
\begin{aligned}
  v^+_1(z) = -v^-_3(z) &\quad \text{and} & & v^+_3(z) =  v^-_1(z), & z\in {}& \R_+, \\
  v^+_2(z) = -v^-_4(z) &\quad \text{and} & & v^+_4(z) =  v^-_2(z), & z\in {}& \R_-.
\end{aligned}
\end{equation}

Again due to the Stokes phenomenon, we cannot find solutions that satisfy the boundary conditions at $\infty$ from all directions, but only sectorally. Here we follow the notation in \cite{Duits-Geudens13} and define six overlapping sectors in the complex plane
\begin{equation} \label{eq:defn_Omega_sectors}
  \Omega_j := \left\{ z\in \C : -\frac{\pi}{12}+\frac{j\pi}{3} < \arg z <\frac{7\pi}{12}+\frac{j\pi}{3}\right\}, \quad j = 0, \dotsc, 5,
\end{equation}
as shown in Figure \ref{fig:Omega_k}.
The following result was proved in \cite[Lemma 5.2]{Duits-Geudens13}.
\begin{prop}\label{prop:Duits-Geudens}
For fixed $r_1, r_2>0$ and $\Om_j$ one of the sectors defined in \eqref{eq:defn_Omega_sectors}, the equation \eqref{int:7_a} has a unique fundamental solution $M^{(j)}$ such that as $z\to \infty$ within $\Om_j$, 
\begin{equation}\label{eq:M_asymptotics}
M^{(j)}(z) = 
\begin{cases}
\left(I+\bigO(z^{-1})\right)\acal^+(z), & \text{for $j = 0,1,2$}, \\
\left(I+\bigO(z^{-1})\right)\acal^-(z), & \text{for $j = 3,4,5$}.
\end{cases}
\end{equation}
\end{prop}
\begin{rem}
  In \cite[Lemma 5.2]{Duits-Geudens13}, the above result is stated for $s_1=s_2 \in \realR$, and $r_1=r_2=1$, but it is trivial to extend to the more general parameters $s_1, s_2, \tau$ and $r_1, r_2 >0$.
\end{rem}

\begin{rem}
  The general theory outlined in \cite[Theorem 19.1]{Wasow87} would indicate a weaker result, namely an asymptotic expansion in powers of $z^{-1/2}$ rather than in powers of $z^{-1}$. The stronger asymptotics above are the result of some symmetry in the equation \eqref{int:7_a}, see the proof in \cite{Duits-Geudens13}.
\end{rem}


Below we construct six 4-vector-valued functions solutions to
\begin{equation} \label{eq:diff_eq_4d}
  \frac{\partial}{\partial z} m = Um,
\end{equation}
which we denote by $n^{(0)}, \dotsc, n^{(5)}$, explicitly from the solutions to the Flaschka--Newell Lax pair \eqref{int:2a}. It is then shown that the solution $n^{(j)}$ is recessive in the sector $S_j$ which was defined in \eqref{eq:Sj_sectors}. Thus these solutions comprise the essential components of the fundamental solutions $M^{(j)}$ satisfying \eqref{eq:M_asymptotics}.

\subsection{Main results}

In order to state the construction and properties of $n^{(0)}, \dotsc, n^{(5)}$, we first introduce some notations. Suppose $\Gamma=\Gamma_1 \cup \Gamma_2 \cup \Gamma_3$ is a trivalent contour, where $\Gamma_1, \Gamma_2,$ and $\Gamma_3$ are three rays in the complex plane which meet at the origin such that $\Gamma_1$ and $\Gamma_2$ are oriented away from the origin, and $\Gamma_3$ is oriented towards the origin. Denote $a, b$, $c$, $\gamma_1$, and $\gamma_2$ as
\begin{equation}\label{diffeq:2}
  \begin{gathered}
  a = \frac{4}{3}\left(\frac{r_1^2-r_2^2}{r_1^2+r_2^2}\right), \quad b = \frac{8\tau}{C^2(r_1^2+r_2^2)}, \quad c = \frac{1}{C}\left[\frac{4\tau^2(r_1^2-r_2^2)}{(r_1^2+r_2^2)^2}-2\left(\frac{s_1}{r_1}-\frac{s_2}{r_2}\right)\right], \\
\ga_1=\exp\left(-\frac{8r_1^4\tau^3}{3(r_1^2+r_2^2)^3}+\frac{4r_1s_1\tau}{r_1^2+r_2^2}\right), \quad \ga_2=\exp\left(-\frac{8r_2^4\tau^3}{3(r_1^2+r_2^2)^3}+\frac{4r_2s_2\tau}{r_1^2+r_2^2}\right),
  \end{gathered}
\end{equation}
and then the function
\begin{equation}
  G(\zeta) = \exp \left( ia \zeta^3 +  b \zeta^2 + ic \zeta \right),
\end{equation}
and the related functions
\begin{equation} \label{eq:defn_of_G_1-G_4}
  G_1(\zeta) = \sqrt{\frac{2}{\pi}} \frac{\gamma_1}{C \sqrt{r_1}} G(\zeta), \quad G_2(\zeta) = \sqrt{\frac{2}{\pi}} \frac{\gamma_2}{C \sqrt{r_2}} G(\zeta), \quad G_3(\zeta) = \frac{2i}{C} \zeta G_1(\zeta), \quad G_4(\zeta) = \frac{2i}{C} \zeta G_2(\zeta).
\end{equation}
Define now an integral transform $\qcal_{\Gamma}$ that transforms two $2$-dimensional vector-valued functions $f(\zeta) = (f_1(\zeta), f_2(\zeta))^T$ and $g(\zeta) = (g_1(\zeta), g_2(\zeta))^T$ into a $4$-dimensional vector-valued function given by,
\begin{multline} \label{Q_def}
  \qcal_\Gamma (f,g)(z) := \\
  \M 
  \begin{pmatrix}
    \int_{\Gamma_1}e^{\frac{2iz\z}{C}} f_1(\zeta) G_1(\zeta) d\zeta + \int_{\Gamma_2}e^{\frac{2iz\z}{C}}g_1(\zeta) G_1(\zeta) d\zeta + \int_{\Gamma_3} e^{\frac{2iz\z}{C}} (f_1(\zeta) + g_1(\zeta)) G_1(\zeta) d\zeta \\
    \int_{\Gamma_1}e^{\frac{2iz\z}{C}} f_2(\zeta) G_2(\zeta) d\zeta + \int_{\Gamma_2}e^{\frac{2iz\z}{C}} g_2(\zeta) G_2(\zeta) d\zeta + \int_{\Gamma_3}e^{\frac{2iz\z}{C}}(f_2(\zeta) + g_2(\zeta)) G_2(\zeta) d\zeta \\
    \int_{\Gamma_1}e^{\frac{2iz\z}{C}} f_1(\zeta) G_3(\zeta) d\zeta + \int_{\Gamma_2}e^{\frac{2iz\z}{C}} g_1(\zeta) G_3(\zeta) d\zeta + \int_{\Gamma_3} e^{\frac{2iz\z}{C}} (f_1(\zeta) + g_1(\zeta)) G_3(\zeta) d\zeta \\
    \int_{\Gamma_1}e^{\frac{2iz\z}{C}} f_2(\zeta) G_4(\zeta) d\zeta + \int_{\Gamma_2}e^{\frac{2iz\z}{C}} g_2(\zeta) G_4(\zeta) d\zeta + \int_{\Gamma_3}e^{\frac{2iz\z}{C}} (f_2(\zeta) + g_2(\zeta)) G_4(\zeta) d\zeta
  \end{pmatrix},
\end{multline}
where $r_1$, $r_2$, $s_1$, $s_2$, $\tau$, and $C$ are the parameters in \eqref{int:5} and \eqref{int:6}, and
\begin{equation} \label{eq:defn_M4x4}
  \M = e^{-\tau z \left( \frac{r^2_1 - r^2_2}{r^2_1 + r^2_2} \right)}
  \begin{pmatrix}
    1 & 0 & 0 & 0 \\
    0 & 1 & 0 & 0 \\
    \frac{i}{r_1} \left( \tau \frac{r^2_1 - r^2_2}{r^2_1 + r^2_2} + \tau - s^2_1 + \frac{u}{C} \right) & \frac{i}{r_1} \frac{\sqrt{r_2} q}{\gamma \sqrt{r_1} C} & \frac{-i}{r_1} & 0 \\
    \frac{-i}{r_2} \frac{\gamma \sqrt{r_1} q}{\sqrt{r_2} C} & \frac{i}{r_2} \left( \tau \frac{r^2_1 - r^2_2}{r^2_1 + r^2_2} - \tau + s^2_2 - \frac{u}{C} \right) & 0 & \frac{-i}{r_2}
  \end{pmatrix},
\end{equation}
where $q=q(\sg)$ is any fixed PII solution evaluated at $\sg$ defined in \eqref{int:9}, and $u=u(\sg)$ is defined in \eqref{int:10}.

We have the following proposition.
\begin{prop} \label{diffeq}
Fix some solution $q$ to \eqref{int:1} and let $\sg$ be as in \eqref{int:9} such that it is not a pole of $q$. Let $\phi(\z)$ and $\varphi(\z)$ be any two 2-vector solutions to \eqref{int:2a}. Assume that for a particular choice of $\Gamma$ the integral transform $\qcal_{\Gamma} (\phi,\varphi)(z)$ exists and is finite for every $z\in \C$. Then $\qcal_{\Gamma} (\phi, \varphi)(z)$ solves the differential equation \eqref{eq:diff_eq_4d}.
\end{prop}
The proof of this proposition is given in Section \ref{diff_eq_proof}.

Now we make a special choice of $\phi(\zeta)$ and $\varphi(\zeta)$ in Proposition \ref{diffeq} and define the particular solutions $n^{(0)}, \dotsc, n^{(5)}$ of \eqref{eq:diff_eq_4d}. Recall the rays $\Sigma_0, \dotsc, \Sigma_5$ defined in \eqref{eq:defn_Sigma_k} (see also Figure \ref{fig:2x2jump}). We define the trivalent contours $\Gamma^{(0)}, \dotsc, \Gamma^{(5)}$ as the $\Gamma$ in Proposition \ref{diffeq} as follows:
\begin{equation} \label{eq:defn_Gamma^(k)}
  \Gamma^{(k)} = \Gamma^{(k)}_1 \cup \Gamma^{(k)}_2 \cup \Gamma^{(k)}_3, \quad \text{where} \quad \Gamma^{(k)}_1 = \Sigma_{1 - k}, \quad \Gamma^{(k)}_2 = \Sigma_{2 - k}, \quad  \Gamma^{(k)}_3 = (-\Sigma_{3 - k}),
\end{equation}
where the contours $\Sg_j$ are oriented towards infinity,  $-\Sigma_{j}$ means the contour $\Sigma_{j}$ oriented in the opposite direction, and $\Sigma_{i - 6} = \Sigma_i$. For an illustration of the contours, see Figure \ref{fig:six_contours}.

Then we define
\begin{equation}\label{eq:def_of_nj1}
  n^{(k)}(z) = n^{(k)}(z; r_1, r_2, s_1, s_2, \tau) = \qcal_{\Gamma^{(k)}}(f^{(k)}, g^{(k)}), \quad k = 0, \dotsc, 5,
\end{equation}
where $r_1, r_2, s_1, s_2, \tau$ are parameters in the formula of $\qcal_{\Gamma}$, and $f^{(k)}$ and $g^{(k)}$ are the columns of fundamental solutions to \eqref{int:2a}, given as
\begin{equation}\label{eq:def_of_nj2}
  \begin{aligned}
    f^{(2j)}(\zeta) = {}&
      \begin{pmatrix}
        \Psi^{(1 - 2j)}_{1, 2}(\zeta; \sigma) \\  
        \Psi^{(1 - 2j)}_{2, 2}(\zeta; \sigma) \\  
      \end{pmatrix},
      & g^{(2j)}(\zeta) = {}& t_{2 + j}
      \begin{pmatrix}
        \Psi^{(1 - 2j)}_{1, 1}(\zeta; \sigma)(z) \\
        \Psi^{(1 - 2j)}_{2, 1}(\zeta; \sigma)(z)
      \end{pmatrix}, \\
    f^{(2j + 1)}(\zeta) = {}&
      \begin{pmatrix}
        \Psi^{(-2j)}_{1, 1}(\zeta; \sigma) \\  
        \Psi^{(-2j)}_{2, 1}(\zeta; \sigma) \\  
      \end{pmatrix},
      & g^{(2j + 1)}(\zeta) = {}& t_{1 + j}
      \begin{pmatrix}
        \Psi^{(-2j)}_{1, 1}(\zeta; \sigma)(z) \\
        \Psi^{(-2j)}_{2, 1}(\zeta; \sigma)(z)
      \end{pmatrix},
  \end{aligned}
  \qquad \text{for $j = 0, 1, 2$},
\end{equation}
where the parameter $\sigma$ is determined by the relation \eqref{int:9}, and $(t_1, t_2, t_3)$ are the Stokes multipliers corresponding to the chosen PII solution. We use the notational conventions $t_{3 + i} = t_i$ and $\Psi^{(6 + i)} = \Psi^{(i)}$, and the subscripts refer to the matrix entries. Here we note that all the $f^{(k)}$ and $g^{(k)}$ are linear combinations of $\psi^{(1)}$ and $\psi^{(2)}$, as shown in Figure \ref{fig:six_contours}, and by the jump condition \eqref{eq:jump_condition_2x2} we see that,
\begin{multline} \label{eq:form_f_and_g}
  f^{(2j)}(\zeta) + g^{(2j)}(\zeta; \sigma) =
  \begin{pmatrix}
    \Psi^{(2- 2k)}_{1, 2}(\zeta; \sigma) \\  
    \Psi^{(2- 2k)}_{2, 2}(\zeta; \sigma)
  \end{pmatrix}, \quad
  f^{(2j + 1)}(\zeta; \sigma) + g^{(2j + 1)}(\zeta; \sigma) =
  \begin{pmatrix}
    \Psi^{(1 - 2k)}_{1, 1}(\zeta; \sigma) \\  
    \Psi^{(1 - 2k)}_{2, 1}(\zeta; \sigma)
  \end{pmatrix}, \\
  \text{for $j = 0, 1, 2$}.
\end{multline}

From the definitions of the functions $n^{(j)}(z)$ and the relation \eqref{eq:triple_number}, the linear relations between them are easy to see, especially in Figure \ref{fig:six_contours}. We have, for example, the pair of independent relations
 \begin{subequations} \label{nj_relations}
 \begin{equation} \label{nj_relations_a}
 n^{(5)}(z) = -t_3 n^{(0)}(z)-(1+t_2t_3) n^{(1)}(z) + t_2n^{(2)}(z) -n^{(3)}(z),
 \end{equation}
  \begin{equation} \label{nj_relations_b}
 n^{(0)}(z) = -t_2 n^{(1)}(z)-(1+t_1t_2) n^{(2)}(z) + t_1n^{(3)}(z) -n^{(4)}(z).
 \end{equation}
 \end{subequations}

\begin{figure}[htb]
  \centering
  \includegraphics{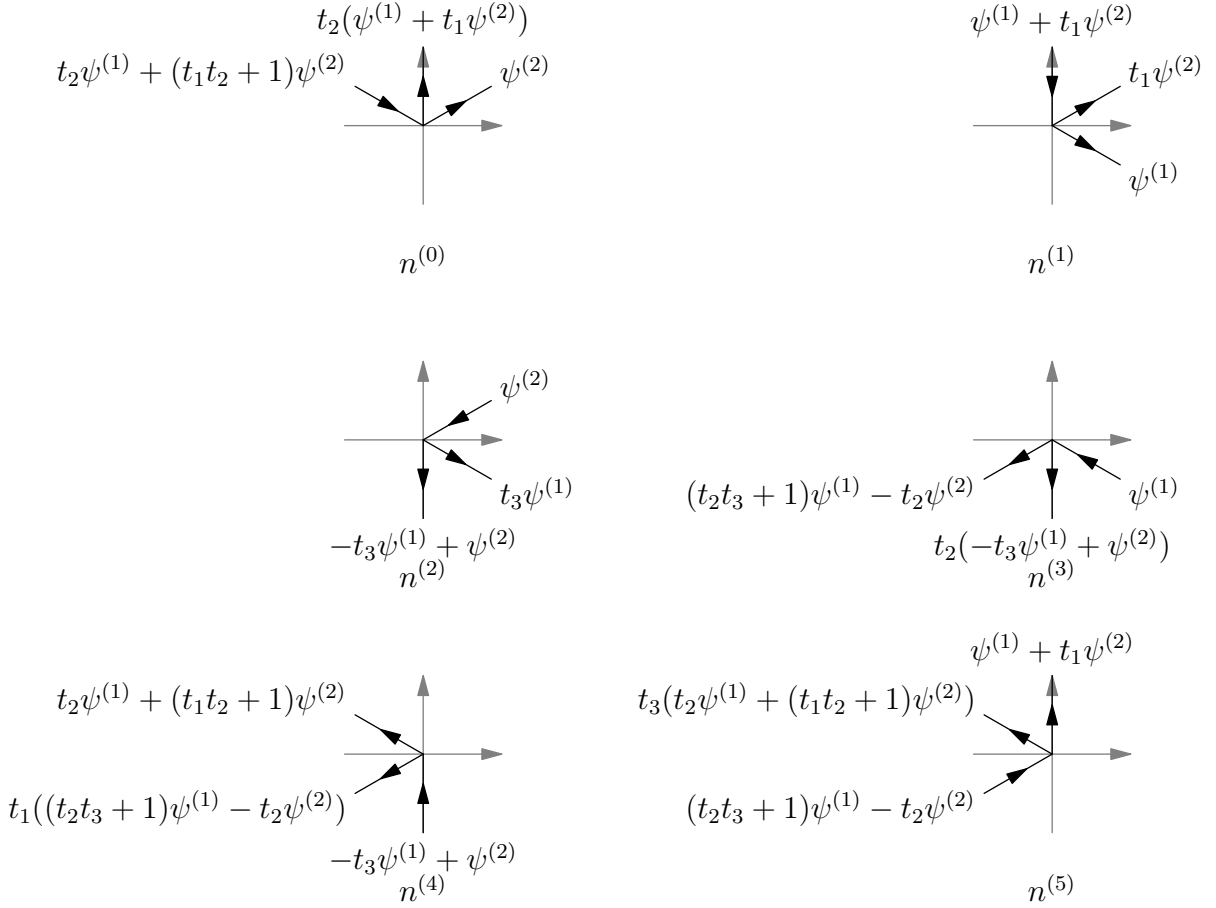}
  \caption{The contours $\Gamma^{(k)}$ for the integral representation of $n^{(0)}, \dotsc, n^{(5)}$. On each ray a two dimensional vector in the form of $c_1 \psi^{(1)} + c_2 \psi^{(2)}$ is given, and they are $f$, $g$, or $f + g$ in the integral formulas $\qcal_{\Gamma^{(k)}}$.}
  \label{fig:six_contours}
\end{figure}

The next result of the paper is that the solutions $n^{(0)}, \dotsc, n^{(5)}$ of \eqref{eq:diff_eq_4d} satisfy the asymptotics of the columns of the fundamental solutions $M^{(0)}, \dotsc, M^{(5)}$ in some of the sectors $\Omega_0, \dotsc, \Omega_5$, and thus these fundamental solutions can be built from the columns $n^{(0)}, \dotsc, n^{(5)}$. To state the proposition, we recall the functions $\theta_1(z)$ and $\theta_2(z)$ defined in \eqref{eq:theta}.

\begin{prop} \label{thm:main_asy}
  Suppose $\delta > 0$ is a small constant. For $z = r e^{i\theta}$ with $\theta$ fixed and $r \to +\infty$, we have the following asymptotic results, where all power functions take the principal branch $(-\pi, \pi)$:
  \begin{enumerate}
  \item \label{enu:thm:main_asy_1}
    Uniformly for $\theta \in (-\pi/3 + \delta, \pi/3 - \delta)$
    \begin{equation}
      n^{(0)}(z) = \frac{1}{\sqrt{2}}e^{-\theta_2(z) - \tau z}
      \left( \bigO(z^{-\frac{3}{4}}),
        z^{-\frac{1}{4}} + \bigO(z^{-\frac{3}{4}}),
        \bigO(z^{-\frac{1}{4}}),
        iz^{\frac{1}{4}} + \bigO(z^{-\frac{1}{4}}) \right)^T,
    \end{equation}
    uniformly for $\theta \in [0, \pi/3 - \delta)$
    \begin{equation}
      n^{(2)}(z) = -\frac{1}{\sqrt{2}}e^{\theta_2(z) - \tau z}
      \left( \bigO(z^{-\frac{3}{4}}),
        i z^{-\frac{1}{4}} + \bigO(z^{-\frac{3}{4}}),
        \bigO(z^{-\frac{1}{4}}),
        z^{\frac{1}{4}} + \bigO(z^{-\frac{1}{4}}) \right)^T,
    \end{equation}
    and uniformly for $\theta \in (-\pi/3 + \delta, 0]$
    \begin{equation}
      n^{(4)}(z) = \frac{1}{\sqrt{2}}e^{\theta_2(z) - \tau z}
      \left( \bigO(z^{-\frac{3}{4}}),
      i z^{-\frac{1}{4}} + \bigO(z^{-\frac{3}{4}}),
      \bigO(z^{-\frac{1}{4}}),
      z^{\frac{1}{4}} + \bigO(z^{-\frac{1}{4}}) \right)^T.
    \end{equation}
  \item \label{enu:thm:main_asy_2}
    Uniformly for $\theta \in (\pi/3 + \delta, \pi - \delta)$
    \begin{equation}
      n^{(2)}(z) = -\frac{1}{\sqrt{2}}e^{\theta_2(z) - \tau z}
      \left( \bigO(z^{-\frac{3}{4}}),
        i z^{-\frac{1}{4}} + \bigO(z^{-\frac{3}{4}}),
        \bigO(z^{-\frac{1}{4}}),
        z^{\frac{1}{4}} + \bigO(z^{-\frac{1}{4}}) \right)^T,
    \end{equation}
    uniformly for $\theta \in [2\pi/3, \pi - \delta)$
    \begin{equation}
      n^{(4)}(z) = -\frac{1}{\sqrt{2}}e^{-\theta_2(z) - \tau z}
      \left( \bigO(z^{-\frac{3}{4}}),
        z^{-\frac{1}{4}} + \bigO(z^{-\frac{3}{4}}),
        \bigO(z^{-\frac{1}{4}}),
        iz^{\frac{1}{4}} + \bigO(z^{-\frac{1}{4}}) \right)^T,
    \end{equation}
    and uniformly for $\theta \in (\pi/3 + \delta, 2\pi/3]$
    \begin{equation}
      n^{(0)}(z) = \frac{1}{\sqrt{2}}e^{-\theta_2(z) - \tau z}
      \left( \bigO(z^{-\frac{3}{4}}),
        z^{-\frac{1}{4}} + \bigO(z^{-\frac{3}{4}}),
        \bigO(z^{-\frac{1}{4}}),
        iz^{\frac{1}{4}} + \bigO(z^{-\frac{1}{4}}) \right)^T.
    \end{equation}
  \item \label{enu:thm:main_asy_3}
    Uniformly for $\theta \in (\pi + \delta, 5\pi/3 - \delta)$
    \begin{equation}
      n^{(4)}(z) = \frac{1}{\sqrt{2}}e^{\theta_2(z) - \tau z}
      \left( \bigO(z^{-\frac{3}{4}}),
        i z^{-\frac{1}{4}} + \bigO(z^{-\frac{3}{4}}),
        \bigO(z^{-\frac{1}{4}}),
        z^{\frac{1}{4}} + \bigO(z^{-\frac{1}{4}}) \right)^T,
    \end{equation}
    uniformly for $\theta \in [4\pi/3, 5\pi/3 - \delta)$
    \begin{equation}
      n^{(0)}(z) = \frac{1}{\sqrt{2}}e^{-\theta_2(z) - \tau z}
      \left( \bigO(z^{-\frac{3}{4}}),
        z^{-\frac{1}{4}} + \bigO(z^{-\frac{3}{4}}),
        \bigO(z^{-\frac{1}{4}}),
        iz^{\frac{1}{4}} + \bigO(z^{-\frac{1}{4}}) \right)^T,
    \end{equation}
    and uniformly for $\theta \in (\pi + \delta, 4\pi/3]$
    \begin{equation}
      n^{(2)}(z) = -\frac{1}{\sqrt{2}}e^{-\theta_2(z) - \tau z}
      \left( \bigO(z^{-\frac{3}{4}}),
        z^{-\frac{1}{4}} + \bigO(z^{-\frac{3}{4}}),
        \bigO(z^{-\frac{1}{4}}),
        iz^{\frac{1}{4}} + \bigO(z^{-\frac{1}{4}}) \right)^T.
    \end{equation}
  \item \label{enu:thm:main_asy_4}
    Uniformly for $\theta \in (\delta, 2\pi/3 - \delta)$
    \begin{equation}
      n^{(1)}(z) = \frac{1}{\sqrt{2}}e^{\theta_1(z) + \tau z}
      \left( -i(-z)^{-\frac{1}{4}} + \bigO(z^{-\frac{3}{4}}),
        \bigO(z^{-\frac{3}{4}}),
        (-z)^{\frac{1}{4}} + \bigO(z^{-\frac{1}{4}}),
        \bigO(z^{-\frac{1}{4}}) \right)^T,
          \end{equation}
    uniformly for $\theta \in([\pi/3, 2\pi/3 - \delta)$
    \begin{equation}
      n^{(3)}(z) = -\frac{1}{\sqrt{2}}e^{-\theta_1(z) + \tau z}
      \left( (-z)^{-\frac{1}{4}} + \bigO(z^{-\frac{3}{4}}),
      \bigO(z^{-\frac{3}{4}}),
      -i(-z)^{\frac{1}{4}} + \bigO(z^{-\frac{1}{4}}),
      \bigO(z^{-\frac{1}{4}}) \right)^T,
    \end{equation}
    and uniformly for $\theta \in (\delta, \pi/3]$
    \begin{equation}
      n^{(5)}(z) = \frac{1}{\sqrt{2}}e^{-\theta_1(z) + \tau z}
      \left( (-z)^{-\frac{1}{4}} + \bigO(z^{-\frac{3}{4}}),
        \bigO(z^{-\frac{3}{4}}),
        -i(-z)^{\frac{1}{4}} + \bigO(z^{-\frac{1}{4}}),
        \bigO(z^{-\frac{1}{4}}) \right)^T.
    \end{equation}
  \item \label{enu:thm:main_asy_5}
    Uniformly for $\theta \in (2\pi/3 + \delta, 4\pi/3 - \delta)$
    \begin{equation}
      n^{(3)}(z) = -\frac{1}{\sqrt{2}}e^{-\theta_1(z) + \tau z}
      \left( (-z)^{-\frac{1}{4}} + \bigO(z^{-\frac{3}{4}}),
      \bigO(z^{-\frac{3}{4}}),
      -i(-z)^{\frac{1}{4}} + \bigO(z^{-\frac{1}{4}}),
      \bigO(z^{-\frac{1}{4}}) \right)^T,
    \end{equation}
    uniformly for $\theta \in [\pi, 4\pi/3 - \delta)$
    \begin{equation}
      n^{(5)}(z) = -\frac{1}{\sqrt{2}}e^{\theta_1(z) + \tau z}
      \left( -i(-z)^{-\frac{1}{4}} + \bigO(z^{-\frac{3}{4}}),
      \bigO(z^{-\frac{3}{4}}),
      (-z)^{\frac{1}{4}} + \bigO(z^{-\frac{1}{4}}),
      \bigO(z^{-\frac{1}{4}}) \right)^T,
    \end{equation}
    and uniformly for $\theta \in (2\pi/3 + \delta \pi]$ 
    \begin{equation}
      n^{(1)}(z) = \frac{1}{\sqrt{2}}e^{\theta_1(z) + \tau z}
      \left( -i(-z)^{-\frac{1}{4}} + \bigO(z^{-\frac{3}{4}}),
      \bigO(z^{-\frac{3}{4}}),
      (-z)^{\frac{1}{4}} + \bigO(z^{-\frac{1}{4}}),
      \bigO(z^{-\frac{1}{4}}) \right)^T.
    \end{equation}
  \item \label{enu:thm:main_asy_6}
    Uniformly for $\theta \in (4\pi/3 + \delta, 2\pi - \delta)$
    \begin{equation}
      n^{(5)}(z) = -\frac{1}{\sqrt{2}}e^{\theta_1(z) + \tau z}
      \left( -i(-z)^{-\frac{1}{4}} + \bigO(z^{-\frac{3}{4}}),
        \bigO(z^{-\frac{3}{4}}),
        (-z)^{\frac{1}{4}} + \bigO(z^{-\frac{1}{4}}),
        \bigO(z^{-\frac{1}{4}}) \right)^T,
    \end{equation}
    uniformly for $\theta \in [5\pi/3, 2\pi - \delta)$
    \begin{equation}
      n^{(1)}(z) = \frac{1}{\sqrt{2}}e^{-\theta_1(z) + \tau z}
      \left( (-z)^{-\frac{1}{4}} + \bigO(z^{-\frac{3}{4}}),
        \bigO(z^{-\frac{3}{4}}),
        -i(-z)^{\frac{1}{4}} + \bigO(z^{-\frac{1}{4}}),
        \bigO(z^{-\frac{1}{4}}) \right)^T,
    \end{equation}
    and uniformly for $\theta \in (4\pi/3 + \delta, 5\pi/3]$
    \begin{equation}
      n^{(3)}(z) = -\frac{1}{\sqrt{2}}e^{-\theta_1(z) + \tau z}
      \left( (-z)^{-\frac{1}{4}} + \bigO(z^{-\frac{3}{4}}),
        \bigO(z^{-\frac{3}{4}}),
        -i(-z)^{\frac{1}{4}} + \bigO(z^{-\frac{1}{4}}),
        \bigO(z^{-\frac{1}{4}}) \right)^T.
    \end{equation}
  \end{enumerate}
\end{prop}
The proof is given in Section \ref{asymptotics}.

To visualize the result of Proposition \ref{thm:main_asy}, we summarize it in Figure \ref{fig:Prop1_3_summary}. 
\begin{figure}[ht]
  \centering
  \includegraphics{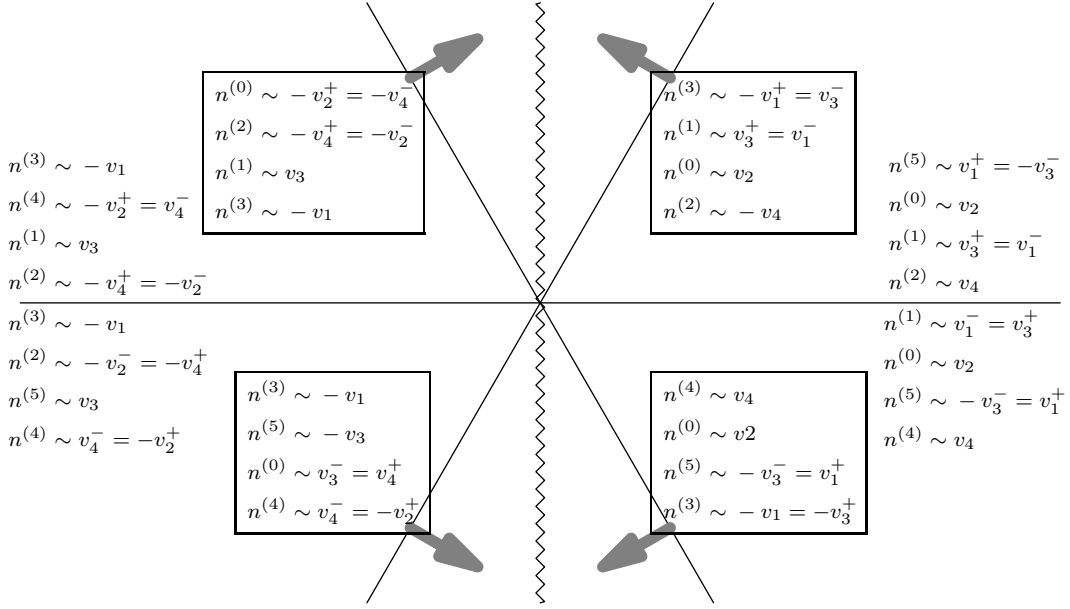}
  \caption{This figure summarizes the result of Proposition \ref{thm:main_asy}. The dividing lines are the real and imaginary axes, as well as $\arg z = \pm \pi/3$ and $\arg z= \pm 4\pi/3$. They separate the complex plane into 8 sectors, and within each sector the leading order behavior of some of the solutions $n^{(k)}(z)$ can be identified with the columns of $\acal^+$ and $\acal^-$. If the function $v_j$ is written without any superscript it means that $v^+_j = v^-_j$ throughout the given sector. The zigzag lines are the branch cuts for $\acal^+$ (the negative imaginary axis), and $\acal^-$ (the positive imaginary axis).}
  \label{fig:Prop1_3_summary}
\end{figure}

As a consequence of Proposition \ref{thm:main_asy}, we see that the vector $n^{(j)}$ is recessive in the sector $S_j$ shown in Figure \ref{fig:2x2jump}. We now describe the entries of the fundamental solutions defined in Proposition \ref{prop:Duits-Geudens} in terms of the solutions $n^{(0)}, \dots, n^{(5)}$.

\begin{theo}\label{fund_solutions}
Fix a PII solution $q(\sg)$ with Stokes multipliers $(t_1, t_2, t_3)$. For $j=0, \dots, 5$, let $M^{(j)}$ be the unique $4\times 4$ matrix-valued solution to \eqref{int:7_a} which satisfies \eqref{eq:M_asymptotics}. We have the following explicit formulas:
\begin{equation}\label{eq:fundamental_solutions}
\begin{aligned}
M^{(0)}&=\left( n^{(5)}+t_3n^{(0)}, n^{(0)}, n^{(1)}, -n^{(2)} \right), &
M^{(1)}&=\left( -n^{(3)}, n^{(0)}+t_2n^{(1)}, n^{(1)}, -n^{(2)} \right), \\
M^{(2)}&=\left( -n^{(3)}, -n^{(4)}, n^{(1)}+t_1 n^{(2)}, -n^{(2)} \right), &
M^{(3)}&=\left( -n^{(3)}, -n^{(2)}-t_3n^{(3)}, -n^{(5)}, n^{(4)} \right), \\
M^{(4)}&=\left( -n^{(3)}-t_2n^{(4)}, n^{(0)}, -n^{(5)}, n^{(4)} \right), &
M^{(5)}&=\left( n^{(1)}, n^{(0)}, -n^{(5)}, n^{(4)} +t_1n^{(5)} \right).
\end{aligned}
\end{equation}
\end{theo}

\medskip

Since we know the linear relations for the six fundamental solutions described above, we may describe them as the solution to a Riemann--Hilbert problem (RHP). In order to state the RHP, define the sectors $\De_j$ as
\begin{equation}\label{def:De_j}
  \De_j  := \left\{ z\in \C : \frac{j\pi}{3} < \arg z <\frac{(j+1)\pi}{3}\right\}, \quad j=0,\dots, 5,
\end{equation}
see Figure \ref{fig:4x4jump}. We then define the function $M(z)$ piecewise in the complex plane as
\begin{equation}\label{def:M}
M(z):=M^{(j)}(z), \quad \textrm{for} \ z \in \De_j, \ j=0, \dots, 5.
\end{equation}
Then $M(z)$ satisfies the following RHP.

\begin{RHP} \label{RHP:4x4}


  \begin{enumerate}[label=(\arabic*)]
 
 \item The $4 \times 4$ matrix-valued function $M$ is analytic in each of the sectors $\De_j$ defined in \eqref{def:De_j}, continuous up to the boundaries, and $M(z) = \bigO(1)$ as $z \to 0$.
  \item On the boundaries of the sectors $\De_j$, $M = M^{(j)}$ satisfies the jump conditions 
    \begin{equation} \label{eq:jump_matrices_4x4}
      M^{(j)}(z) = M^{(j - 1)}(z) J_j, \quad \text{for $j = 0, \dotsc, 5$}, \quad M^{(-1)}\equiv M^{(5)},
    \end{equation} 
    for the jump matrices $J_0, \dots, J_5$ with constant entries specified in Figure \ref{fig:4x4jump}.
  \item As $z\to \infty$, $M(z)$ satisfies the asymptotics
    \begin{equation}\label{eq:RHP_asymptotics}
      M(z) = \left( 1 + \bigO(z^{-1}) \right) \left( v_1(z), v_2(z), v_3(z), v_4(z) \right),
    \end{equation}
    where $v_1, v_2, v_3,$ and $v_4$ are defined in \eqref{eq:defn_v_1-v_4}.
  \end{enumerate}
\end{RHP}
\begin{figure}[htb]
  \begin{minipage}[b]{0.2\linewidth}
    \centering
    \includegraphics{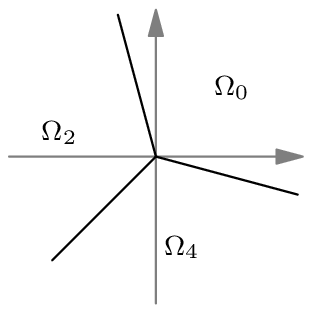}
    \includegraphics{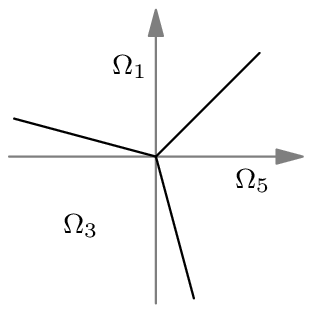}
    \caption{The sectors $\Omega_0, \dotsc, \Omega_5$.}
    \label{fig:Omega_k}
  \end{minipage}
  \hspace{\stretch{1}}
  \begin{minipage}[b]{0.8\linewidth}
    \includegraphics{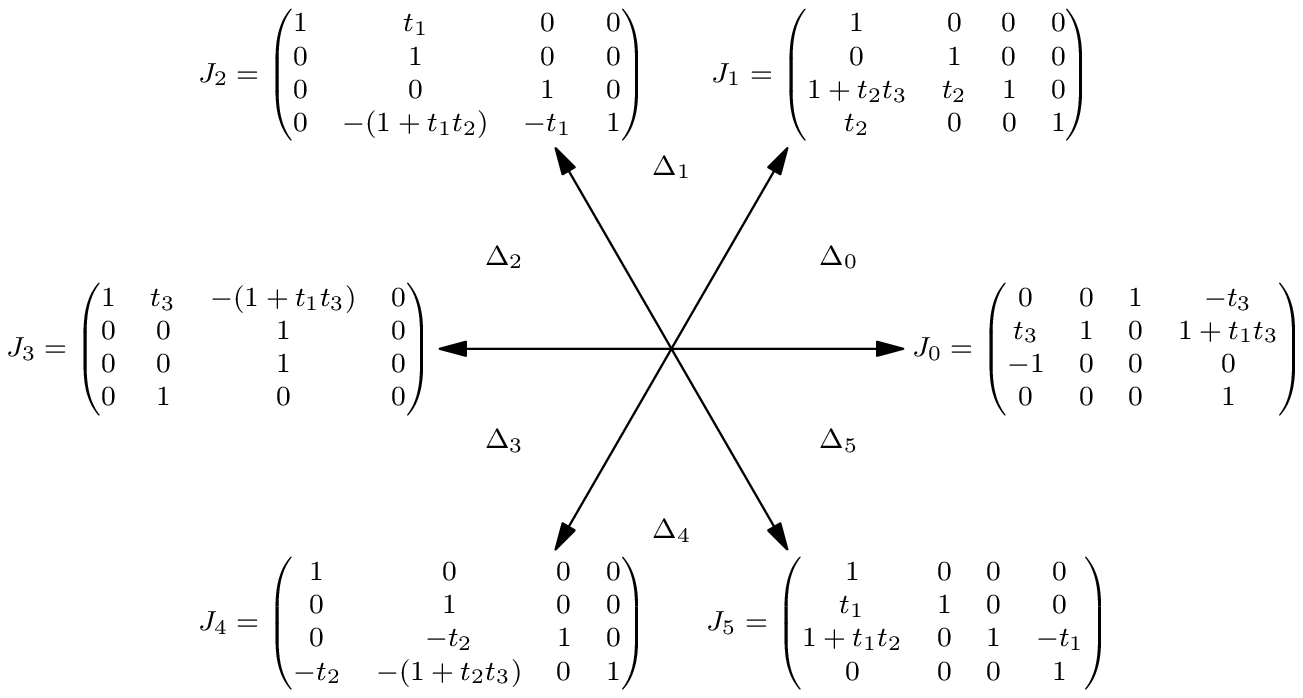}
    \caption{Jump matrices $J_0, \dotsc, J_5$ for RHP \ref{RHP:4x4}.}
    \label{fig:4x4jump}
  \end{minipage}
\end{figure}
It is not hard to see that there is at most one $M$ that satisfies RHP \ref{RHP:4x4}. So we have
\begin{cor}
  $M^{(j)}$ ($j = 0, \dotsc, 5$) are uniquely determined by Riemann--Hilbert problem \ref{RHP:4x4} and \eqref{def:M}.
\end{cor}

Notice that the jump matrices satisfy the symmetry
\begin{equation}
J_k=\begin{pmatrix} 0 & 1 & 0 & 0 \\ 1 & 0 & 0 & 0 \\ 0 & 0 & 0 & -1 \\ 0 & 0 & -1 & 0 \end{pmatrix} J_{k+3}\begin{pmatrix} 0 & 1 & 0 & 0 \\ 1 & 0 & 0 & 0\\ 0 & 0 & 0 & -1 \\ 0 & 0 & -1 & 0 \end{pmatrix},
\end{equation}
where $J_{6+k}\equiv J_k$. This implies the symmetry of the solutions
\begin{equation}
M(-z; r_1, r_2, s_1, s_2,\tau)=\begin{pmatrix} 0 & 1 & 0 & 0 \\ 1 & 0 & 0 &0 \\ 0 & 0 & 0 & -1 \\ 0 & 0 & -1 & 0 \end{pmatrix}M(z; r_2, r_1, s_2, s_1,\tau)\begin{pmatrix} 0 & 1 & 0 & 0 \\ 1 & 0 & 0 & 0 \\ 0 & 0 & 0 & -1 \\ 0 & 0 & -1 & 0 \end{pmatrix},
\end{equation}
which appears in \cite[Lemma 5.1(b)]{Delvaux-Kuijlaars-Zhang11} for the Hastings--McLeod case. (Their result is for the tacnode RHP that is equivalent to RHP \ref{RHP:4x4} in the Hastings--McLeod case, see Section \ref{subsec:tacnode_RHP}.) In this special case, there are additional symmetries with respect to complex conjugation which are not present in the general case.

\begin{rem}
 As informed by an anonymous referee, our integral representation of the fundamental solution of the Lax system \eqref{int:7} by the fundamental solution of the Flaschka--Newell Lax pair \eqref{int:2} is essentially similar to the representation of fundamental solution of the Harnad--Tracy--Widom Lax pair by the fundamental solution of the Jimbo--Miwa Lax pair in \cite[Theorem 3.1]{Joshi-Kitaev-Treharne09}. The idea of using a generalized Laplace transform to produce a Lax pair which is linear in the spectral variable from another which is quadratic in the spectral variable is in fact quite general and has been applied to other \Painleve\ equations as well, see \cite{Joshi-Kitaev-Treharne07}.
\end{rem}

\subsection{Contour integral formulas for two-matrix critical kernel and tacnode kernel}

In the special case $(t_1, t_2, t_3)=(1, 0, -1)$, $M(z)$ is the solution to the Lax system for the Hastings--McLeod solution to PII, and we refer to this special case as $M^{\HM}(z)$ below. In this section we discuss two occurrences of the entries of $M^{\HM}$, one in the two-matrix model critical kernel and the other in the tacnode kernel in the non-intersecting Brownian motion model. Originally these two kernels are expressed in terms of the tacnode RHP, which differs from the Hastings--McLeod case of our Riemann--Hilbert problem \ref{RHP:4x4} only by a constant matrix factor, see \eqref{eq:tacnode_conjugation} below. The integral formulas for the entries of $M^{\HM}$ yield contour integral formulas for these two kernels.

\subsubsection{Critical kernel in two-matrix model} \label{subsec:2MM}

Consider the two-matrix model in which two $n \times n$ random Hermitian matrices $M_1$ and $M_2$ have the joint probability measure
\begin{equation}
  \frac{1}{C_n} \exp(-n \Tr(V(M_1) + W(M_2) - \tau M_1 M_2)) dM_1 dM_2,
\end{equation}
where $V$ and $W$ are potentials and $\tau$ is the coupling constant. We concentrate on the distribution of eigenvalues of $M_1$, which is a determinantal process and is thus characterized by a correlation kernel. In the case that $V(x) = x^2/2$, $W(y) = y^4/4 + \alpha y^2/2$ and $n \to \infty$, the model is in the critical phase if $\alpha = -1$ and $\tau = 1$. As $n \to \infty$, under the double scaling limit $\alpha = -1 + 2a n^{-1/3} - b n^{-2/3}$ and $\tau = 1 + a n^{-1/3} + 2b n^{-2/3}$, where $a$ and $b$ are constants, the correlation kernel for the eigenvalues of $M_1$, at $x n^{-2/3}$ and $y n^{-2/3}$ converges to $K^{\crit}_2(x, y; (a^2 - 5b)/4, -a)$, whose formula is expressed by the tacnode RHP. See \cite{Duits-Geudens13} for the derivation, and also \cite{Duits14}.

Similar to the critical kernel $K^{\crit}_1$ in the one-matrix model, the limiting kernel $K^{\crit}_2$ is believed to be universal, and it should occur in very general settings of the two-matrix model. If $V$ is a quadratic polynomial, the forthcoming paper \cite{Claeys-Kuijlaars-Liechty-Wang17} will show that $K^{\crit}_2$ occurs for a large class of potentials $W$.

\subsubsection{Tacnode kernel in nonintersecting Brownian motion model} \label{subsec:NIBM}

Consider $(1 + \lambda)n$ non-colliding particles in Brownian bridges, with diffusion parameter $n^{-1/2}$. Suppose the particles are in two groups, such that the left $n$ of them are in the first group and the right $\lambda n$ of them are in the second. Let particles in the first group start at $a_1$ at time $0$, and end at $a_1$ at time $1$, and let particles in the second group start at $a_2$ at time $0$, and end at $a_2$ at time $1$. The particles in this model are a determinantal process, and their multi-time correlation functions are given by the multi-time correlation kernel.

If $a_1 = -1$ and $a_2 = \sqrt{\lambda}$, the model is in the critical phase as $n \to \infty$, such that the right-most particle in the first group meets narrowly the left-most particle in the second group at time $0.5$, and their trajectories touch each other like a tacnode. As $n \to \infty$, under the double scaling $a_1 = -(1 + (\Sigma/2) n^{-2/3})$ and $a_2 = \sqrt{\lambda}(1 + (\Sigma/2) n^{-2/3})$, the multi-time correlation kernel at positions $(x/2) n^{-2/3}$ and $(y/2) n^{-2/3}$ and times $(1 + \tau_1 n^{-1/3})/2$ and $(1 + \tau_2 n^{-1/3})/2$ converges to $\lcal^{\lambda, \Sigma}_{\tac}(\tau_1, x; \tau_2, y)$, which is expressed by the tacnode Riemann--Hilbert problem.

The derivation of $\lcal^{\lambda, \Sigma}_{\tac}$ was achieved by several groups of people: Adler, Ferrari and van Moerbeke got a multi-time tacnode kernel formula with $\lambda = 1$ from a discrete random walk model \cite{Adler-Ferrari-van_Moerbeke13}; Delvaux, Kuijlaars and Zhang got a single time tacnode kernel formula from the nonintersecting Brownian motion model \cite{Delvaux-Kuijlaars-Zhang11}; Johansson got a multi-time tacnode kernel formula with $\lambda = 1$ from the nonintersecting Brownian motion model \cite{Johansson13}; Ferrari and \Veto\ generalized Johansson's result for general $\lambda > 0$ \cite{Ferrari-Veto12} \footnote{Our notation $\lcal^{\lambda, \Sigma}_{\tac}$ follows that in \cite{Ferrari-Veto12}, but with their $\sigma$ replaced by $\Sigma$. The reason is that $\sigma$ occurs in our paper everywhere as the argument of $q$, the solution of \eqref{int:1}. We note that in \cite{Delvaux13}, the author took the same change of notation, as explained in \cite[Formula 2.15]{Delvaux13}.}. The formulas of Adler--Ferrari--van Moerbeke and Johansson were both expressed in terms of Airy resolvents, but quite different in structure. They were later was proved to be equivalent \cite{Adler-Johansson-van_Moerbeke14}. Delvaux--Kuijlaars--Zhang's formula was expressed by the tacnode RHP (their paper first defined the tacnode RHP, and the RHP is named thereby). Later Delvaux showed in \cite{Delvaux13} the equivalence of the results in \cite{Delvaux-Kuijlaars-Zhang11} and \cite{Ferrari-Veto12}, and furthermore wrote the general multi-time tacnode kernel in the tacnode Riemann-Hilbert problem. See also \cite{Kuijlaars14}. A variation of the model where the nonintersecting Brownian bridges are on a circle was studied by the current authors in \cite{Liechty-Wang14-2}.

\subsubsection{Tacnode Riemann--Hilbert problem revisited} \label{subsec:tacnode_RHP}

The tacnode RHP which is mentioned in Sections \ref{subsec:2MM} and \ref{subsec:NIBM} was defined in \cite{Delvaux-Kuijlaars-Zhang11}, \cite{Duits-Geudens13}, and \cite{Kuijlaars14}, with minor variations in generality and formality. The definition in \cite[Section 2.1]{Kuijlaars14} resembles the Hastings--McLeod case of our Riemann--Hilbert problem \ref{RHP:4x4}. Let us denote the solution to the RHP in \cite[Section 2.1]{Kuijlaars14} by $M^{\tac}$. Then $M^{\tac}$ is defined on regions $\Delta_0, \dotsc, \Delta_5$, with the same boundary conditions and asymptotics at $\infty$, but the jump matrices are slightly different from those of $M^{\HM}$. It seems perplexing that two similar but different Riemann--Hilbert problems are associated to Lax sytem \eqref{int:7} (in the Hastings--McLeod case), but the reason is simple: Although the solutions $M^{(0)}, \dots, M^{(5)}$ are the unique solutions to \eqref{int:7_a} satisfying the asymptotics \eqref{eq:M_asymptotics} throughout the sectors $\Om_j$, as stated in Proposition \ref{prop:Duits-Geudens}, if the sectors are shrunk to $\Delta_j$, the asymptotics \eqref{eq:M_asymptotics} do not uniquely determine the solution in each sector. The RHP \ref{RHP:4x4} and the tacnode RHP in \cite{Kuijlaars14} require only asymptotics in sectors $\De_j$, so there is some freedom to choose the jump matrices corresponding to different solutions to \eqref{int:7_a}. The relation between the RHP \ref{RHP:4x4} and the tacnode RHP is as follows. 

\begin{equation} \label{eq:tacnode_conjugation}
  \begin{split}
    M^{\tac}(z)= {}&
    M^{\HM}(z)
    \begin{pmatrix}
      1 & 0 & 0 & 0 \\
      1 & 1 & 0 & 0 \\
      0 & 0 & 1 & -1 \\
      0 & 0 & 0 &1
    \end{pmatrix}
    \quad \text{for $z\in \De_0$}, \quad
    M^{\tac}(z) = M^{\HM}(z)
    \begin{pmatrix}
      1 & 1 & 0 & 0 \\
      0 & 1 & 0 & 0 \\
      0 & 0 & 1 & 0 \\
      0 & 0 & -1 &1
    \end{pmatrix}
    \quad \text{for $z\in \De_3$}, \\
    M^{\tac}(z)= {}& M^{\HM}(z) \quad \text{otherwise}.
\end{split}
\end{equation}

In the paper \cite{Kuijlaars14}, Kuijlaars found explicit formulas for the entries of the solution to the tacnode RHP in terms of Airy functions and related operators. He found six solutions to the differential equation \eqref{int:7_a} with $q(\sg)$ being the Hastings--McLeod solution to PII, which were labeled $m^{(0)}, \dots, m^{(5)}$. Let us remark here that in the case $(t_1, t_2, t_3)=(1, 0, -1)$, the solutions $n^{(0)}, \dots, n^{(5)}$ constructed in this paper agree with the ones constructed by Kuijlaars up to sign. Specifically we have
\begin{equation} \label{eq:Kuijlaars_compare}
  \begin{aligned}
    n^{(0)} = {}& m^{(0)}, & n^{(1)} = {}& m^{(1)}, & n^{(2)} = {}& -m^{(2)}, \\
    n^{(3)} = {}& -m^{(3)}, & n^{(4)} = {}& m^{(4)}, & n^{(5)} = {}& -m^{(5)},
\end{aligned}
\end{equation}
which follows from comparing  \cite[Figure 2]{Kuijlaars14} with \eqref{eq:tacnode_conjugation} in light of Theorem \ref{fund_solutions}.

\subsubsection{Contour integral formulas}

We can write the critical kernel $K^{\crit}_2$ for two-matrix model in a contour integral formula where the integrand is expressed by entries of $\Psi^{(0)}$ in \eqref{eq:defn_psi^1_psi^2}, the solution to the Flaschka--Newell Lax pair.  Let $\Sg_{\tac}$ be a contour consisting of two infinite pieces: the first passing from $e^{-5\pi i/6} \cdot \infty$ to $e^{-\pi i/6} \cdot \infty$; and the second passing from $e^{\pi i/6} \cdot \infty$ to $e^{5\pi i/6} \cdot \infty$, as pictured in Figure \ref{fig:Sigma_T}.
Also let $\Sg_{\tMM}$ be the contour
\begin{equation}
  \Sg_{\tMM} =[1, e^{i\pi/6}\cdot \infty) \cup [1, e^{-i\pi/6}\cdot \infty)\cup[-1, e^{5i\pi/6}\cdot \infty) \cup [-1, e^{-5i\pi/6}\cdot \infty)\cup[-1,1],
\end{equation}
oriented as shown in Figure \ref{fig:Sigma_T_deformed}.
Define the functions $f$ and $g$ on $\C$ in a piecewise way:
\begin{equation}\label{tac:8}
  f(\zeta; \sg) :=
  \begin{cases}
    -\Psi^{(0)}_{1,2}(\zeta; \sg) & \text{if $\Im \zeta > 0$}, \\
    \Psi^{(0)}_{1,1}(\zeta; \sg) & \text{if $\Im \zeta < 0$}, \\
    \Phi_{1}(\zeta; \sg) & \text{if $\Im \zeta = 0$},
  \end{cases}
  \qquad
  g(u; \sg) :=
  \begin{cases}
    -\Psi^{(0)}_{2,2}(\zeta; \sg) & \text{if $\Im \zeta > 0$}, \\
    \Psi^{(0)}_{2,1}(\zeta; \sg) & \text{if $\Im \zeta < 0$}, \\
     \Phi_{2}(\zeta; \sg) & \text{if $\Im \zeta = 0$},
  \end{cases}
\end{equation}
where
\begin{equation} \label{eq:formula_Phi_1_2}
  \Phi_{1}(\zeta; \sg):= \Psi^{(0)}_{1,1}(\zeta; \sg)+\Psi^{(0)}_{1,2}(\zeta; \sg), \qquad  \Phi_{2}(\zeta; \sg):=\Psi^{(0)}_{2,1}(\zeta; \sg)+\Psi^{(0)}_{2,2}(\zeta; \sg).
\end{equation}
We then have the following theorem.

\begin{figure}[ht]
  \begin{minipage}[t]{0.45\linewidth}
    \centering
    \includegraphics{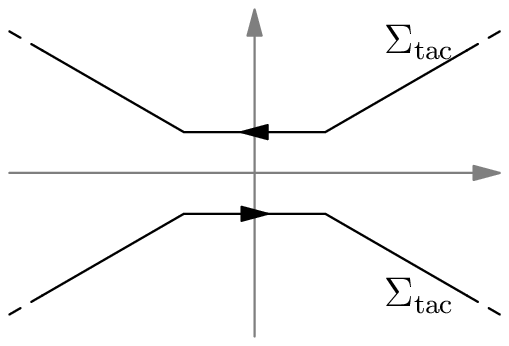}
    \caption{The contour $\Sg_{\tac}$. The rays make the angles $\pi/6$ with the real axis.}
    \label{fig:Sigma_T}
  \end{minipage}
  \hspace{\stretch{1}}
  \begin{minipage}[t]{0.45\linewidth}
    \centering
    \includegraphics{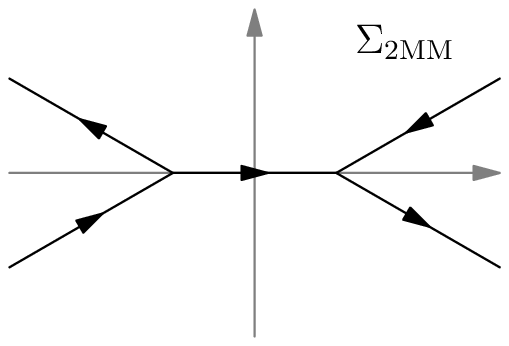}
    \caption{The contour $\Sg_{\tMM}$. The rays make the angles $\pi/6$ with the real axis, and the horizontal segment is the interval $[-1,1]$.}
    \label{fig:Sigma_T_deformed}
  \end{minipage}
\end{figure}

\begin{theo}\label{critical}
The critical kernel in the two-matrix model of Duits--Geudens \cite[equation (2.15)]{Duits-Geudens13} can be written as
  \begin{equation} \label{critkernel} 
    \begin{aligned}
      K_2^{\crit}(x,y; s,\tau)&=\frac{1}{2^{1/3}\pi} \int_{-i\infty}^{i\infty} \,du\,\int_{\Sg_{\tMM}} \,dv\, e^{-2^{4/3}\tau (u^2-v^2)+2^{2/3} (xu-y v)} \\
      &\hspace{4cm}\times \left(\frac{\Phi_1(u;\sg)g(v; \sg)-\Phi_2(u;\sg)f(v; \sg)}{2\pi(u-v)}\right) ,
    \end{aligned}
  \end{equation}
where $\sg$ is given as 
\begin{equation} \label{eq:expression_sigma_K^cr_2}
  \sg = 2^{2/3} (2s - \tau^2).
\end{equation}

\end{theo}
This theorem is proved in Section \ref{two_matrix_kernel_proof}.

\medskip

Similarly we can write the tacnode kernel $\lcal^{\lambda, \Sigma}_{\tac}$ in a double contour integral formula with integrand given by entries of $\Psi^{(0)}$ in \eqref{eq:defn_psi^1_psi^2}.  We have the following theorem.

\begin{theo}\label{tacnode}
The tacnode kernel of Ferrari--\Veto\ \cite{Ferrari-Veto12} can be written as
\begin{multline}\label{tac:16}
  \lcal_{\tac}^{\la, \Sigma}(\tau_1, x; \tau_2, y) = -1_{\tau_1 < \tau_2} \frac{1}{\sqrt{4\pi(\tau_2 - \tau_1)}} \exp \left( \frac{(y - x)^2}{4(\tau_2 - \tau_1)} \right) \\
  + \frac{1}{C\pi} \int_{\Sg_{\tac}}\,du\int_{\Sg_{\tac}}\,dv\, e^{-\frac{4i}{3}(\frac{1-\sqrt{\la}}{1+\sqrt{\la}})(u^3-v^3)+\frac{4}{C^2}(\tau_1u^2-\tau_2v^2)+\frac{2i}{C}(xu-yv)+\frac{i\sg(1-\sqrt{\la})}{1+\sqrt{\la}}(u-v)} \\
  \times \left( \frac{\left(f(u; \sg)g(v; \sg)-f(v; \sg)g(u; \sg)\right)}{2\pi i(u-v)} \right),
\end{multline}
where 
\begin{equation} \label{eq:defn_C_for_tacnode}
  C=\left(1+\frac{1}{\sqrt{\la}}\right)^{1/3} \quad \text{and} \quad \sigma = \lambda^{1/2} C^2 \Sigma.
\end{equation}
\end{theo}
This theorem is proved in Section \ref{tacnode_kernel_proof}.

\begin{rem}
  In the symmetric case $\la=1$, the formula \eqref{tac:16}, up to a rescaling, was derived by the current authors from a model of nonintersecting paths on the circle \cite{Liechty-Wang14-2}. That model was not robust enough to produce the asymmetric tacnode kernel, and the above theorem is new for $\la \ne 1$.
\end{rem}
\begin{rem}
  Formulas \eqref{critkernel} and \eqref{tac:16} are analogous, but \eqref{tac:16} is more general in the sense that (i) it has a $\lambda$ parameter and (ii) the $\tau_1, \tau_2$ parameters corresponding to $\tau$ in \eqref{critkernel} can be different. In \cite{Claeys-Kuijlaars-Liechty-Wang17}, a more general two-matrix model as well as its dynamical version is considered, and \eqref{critkernel} is generalized to a formula containing $\lambda, \tau_1, \tau_2$ parameters like \eqref{tac:16}.
\end{rem}

\subsection{Outlook} \label{subsec:outlook}

In this paper we concentrate on the homogeneous \Painleve\ II equation \eqref{int:1}. The general PII equation has a constant term: $y'' = xy + 2y^3 - \alpha$. 
Both the the Flaschka--Newell Lax pair \eqref{int:2} and the $4\times 4$ Lax system can be generalized to the inhomogeneous PII equation, and each appears in the kernel for a determinantal process.

The Flaschka--Newell Lax pair for the Hastings--McLeod solution of inhomogeneous PII equation occurs in the one-matrix model with logarithmic perturbation, see \cite{Claeys-Kuijlaars-Vanlessen08} and a brief discussion in Section \ref{subsubsec:1MM}. On the other hand, a $4 \times 4$ Riemann--Hilbert problem associated with the Hastings--McLeod solution of inhomogeneous PII equation occurs in the limiting critical correlation kernel for the nonintersecting squared Bessel processes. The (nonintersecting) squared Bessel processes are in some sense a generalization of the (nonintersecting) Brownian motions, and the limiting critical process for the nonintersecting squared Bessel process is a ``hard-edge'' generalization of the tacnode process. Hence the aforementioned $4 \times 4$ RHP, which is then called the hard-edge tacnode RHP, is a natural generalization of the tacnode RHP. This hard-edge tacnode RHP is also associated with a Lax system that is analogous to and more general than \eqref{int:7}. This hard-edge tacnode RHP is also related to a chiral two-matrix model. See \cite{Delvaux13a} and \cite{Delvaux-Geudens-Zhang13}.

It is tempting to conjecture that our construction of the $4 \times 4$ Lax system from the $2 \times 2$ Flaschka--Newell Lax pair to can be applied to the inhomogeneous case as well, thereby giving formulas  for the hard-edge tacnode RHP in terms of solutions to the Flaschka--Newell Lax pair. However, we have so far not been able to derive the relation in a straightforward way.

\subsection*{Organization of the paper}

The algebraic result Proposition \ref{diffeq} is proved in Section \ref{diff_eq_proof}, and next the analytic result Proposition \ref{thm:main_asy} is proved in Section \ref{asymptotics}. Then the main result Theorem \ref{fund_solutions} is proved in Section \ref{sec:proof_main}, based on Proposition \ref{thm:main_asy}. As the applications of the main theorem, Theorems \ref{critical} and \ref{tacnode} are proved in Section \ref{sec:proof_of_contour_int_form}.

\subsection*{Acknowledgements}

We thank the anonymous referees for valuable comments, especially for pointing out related literature on isomonodromy deformations. We also thank Peter Forrester for comments related to random matrix theory.

\section{The proof of Proposition \ref{diffeq}}\label{diff_eq_proof}

First we note that if $m = (m_1, m_2, m_3, m_4)^T$ satisfies equation \eqref{eq:diff_eq_4d}, then the components $m_3, m_4$ are expressed in terms of $m_1, m_2$ by
\begin{align}
  ir_1m_3&=m_1'-\left(\tau-s_1^2+\frac{u}{C}\right)m_1-\frac{\sqrt{r_2}q}{\ga\sqrt{r_1} C} m_2, \label{int:12} \\
  ir_2m_4&=m_2'+\left(\tau-s_2^2+\frac{u}{C}\right)m_2+\ga\frac{\sqrt{r_1}q}{\sqrt{r_2} C} m_1, \label{int:12_2}
\end{align}
and then equation \eqref{eq:diff_eq_4d} is reduced to the equations in $m_1, m_2$:
  \begin{align}
    m_1''  = {}& +2\tau m_1'+\frac{r_1^{3/2}\sqrt{r_2}q(\sg) C^2}{\ga} m_2' \notag \\
    & + \big(Cq(\sg)^2 r_1^2-r_1^2 z+2r_1s_1-\tau^2\big)m_1-\frac{r_1^{3/2}\sqrt{r_2}}{\ga C}\left[q(\sg) \tau\left(\frac{1}{r_1^2}-\frac{1}{r_2^2}\right)+C^2q'(\sg)\right]m_2, \label{int:13} \\
    m_2'' = {}& -2\tau m_2'-r_2^{3/2}\sqrt{r_1}q(\sg) C^2\ga m_1' \notag \\
    & + \big(Cq(\sg)^2 r_2^2+r_2^2 z+2r_2s_2-\tau^2\big)m_2+\frac{r_2^{3/2}\sqrt{r_1}\ga}{C}\left[q(\sg) \tau\left(\frac{1}{r_1^2}-\frac{1}{r_2^2}\right)-C^2q'(\sg)\right]m_1. \label{int:13_2}
  \end{align}
Conversely, if the four components of $m$ satisfy \eqref{int:12}, \eqref{int:12_2}, \eqref{int:13}, and \eqref{int:13_2}, then $m$ is a solution to \eqref{eq:diff_eq_4d}.

Suppose
\begin{equation} \label{eq:temporary_m}
  m(z) = \qcal_{\Gamma}(f, g), 
\end{equation}
where $f$ and $g$ are any two $2$-dimensional vector-valued functions that make $\qcal_{\Gamma}$ well defined on them. We denote functions $I_k = I_k(z)$, $k = 1, 2, 3, 4$, where
\begin{equation}
  I_k = I_{k, 1} + I_{k, 2} + I_{k, 3},
\end{equation}
such that for $j = 1, 2, 3$, with $\sgn(k) = 1$ if $k = 1, 3$ and $\sgn(k) = 2$ if $k = 2, 4$,
\begin{equation}
  I_{k, j} = \int_{\Gamma_j} h^{(j)}_{\sgn(k)}(\zeta) e^{\frac{2iz\z}{C}}G_k(\zeta) d\zeta, \quad \text{where} \quad h^{(1)}(\zeta) = f(\zeta), \quad h^{(2)}(\zeta) = g(\zeta), \quad h^{(3)}(\zeta) = f(\zeta) + g(\zeta).
\end{equation}

Then by definition \eqref{Q_def},
\begin{equation} \label{eq:ir_1m_3_in_I_k}
  ir_1 m_3(z) = e^{-\tau z \left( \frac{r^2_1 - r^2_2}{r^2_1 + r^2_2} \right)} \left( -\left( \tau \frac{r^2_1 - r^2_2}{r^2_1 + r^2_2} + \tau - s^2_1 + \frac{u}{C} \right) I_1 - \frac{\sqrt{r_1} q}{\gamma \sqrt{r_2} C} I_2 + I_3 \right).
\end{equation}
On the other hand, also by definition \eqref{Q_def},
\begin{equation} \label{eq:formulas_for_m_1_m_2_m_1'}
   m_1(z) = e^{-\tau z \left( \frac{r^2_1 - r^2_2}{r^2_1 + r^2_2} \right)} I_1, \quad m_2(z) = e^{-\tau z \left( \frac{r^2_1 - r^2_2}{r^2_1 + r^2_2} \right)} I_2,
\end{equation}
and we can evaluate $m'_1(z)$ as follows. Since $\frac{d}{dz} e^{\frac{2iz\z}{C}} G_1(\zeta) = e^{\frac{2iz\z}{C}}G_3(\zeta)$ by \eqref{eq:defn_of_G_1-G_4}, we have for $j = 1, 2, 3$,
\begin{equation} \label{eq:I_1_prime_id}
  \frac{d}{dz} I_{1, j} = \int_{\Gamma_j} h^{(j)}_1(\zeta)\left( \frac{d}{dz} e^{\frac{2iz\z}{C}}\right)G_1(\zeta) d\zeta = \int_{\Gamma_j} h^{(j)}_j(\zeta) G_3(\zeta) d\zeta = I_{3, j}.
\end{equation}
Thus
\begin{equation} \label{eq:I_1_prime_id_more}
  \frac{d}{dz} I_1 = I_3,
\end{equation}
and we have
\begin{equation} \label{eq:m'_1(z)_final}
  m'_1(z) = \frac{d}{dz} \left( e^{-\tau z \left( \frac{r^2_1 - r^2_2}{r^2_1 + r^2_2} \right)} I_1 \right) = -\tau \frac{r^2_1 - r^2_2}{r^2_1 + r^2_2} m_1(z) + e^{-\tau z \left( \frac{r^2_1 - r^2_2}{r^2_1 + r^2_2} \right)} \frac{d}{dz} I_1 = -\tau \frac{r^2_1 - r^2_2}{r^2_1 + r^2_2} m_1(z) + m_3(z).
\end{equation}
Using expressions \eqref{eq:ir_1m_3_in_I_k}, \eqref{eq:formulas_for_m_1_m_2_m_1'} and \eqref{eq:m'_1(z)_final}, we check that \eqref{int:12} holds. Similarly, we can check that \eqref{int:12_2} holds.

Next we show that if the $f$ and $g$ in \eqref{eq:temporary_m} are chosen to be the solutions $\phi(\z)$ and $\varphi(\z)$ to \eqref{eq:diff_eq_4d}, as in Proposition \ref{diffeq}, then identities \eqref{int:13} and \eqref{int:13_2} also hold.

Consider first $m_1(z)$. We have
\begin{equation} \label{eq:m''_1(z)_start}
  \begin{split}
    m''_1(z) = {}& \frac{d^2}{dz^2} \left( e^{-\tau z \left( \frac{r^2_1 - r^2_2}{r^2_1 + r^2_2} \right)} I_1 \right) \\
    = {}& \tau^2 \left( \frac{r^2_1 - r^2_2}{r^2_1 + r^2_2} \right)^2 m_1(z) - 2\tau \frac{r^2_1 - r^2_2}{r^2_1 + r^2_2} e^{-\tau z \left( \frac{r^2_1 - r^2_2}{r^2_1 + r^2_2} \right)} \frac{d}{dz} I_1 + e^{-\tau z \left( \frac{r^2_1 - r^2_2}{r^2_1 + r^2_2} \right)} \frac{d^2}{dz^2} I_1.
  \end{split}
\end{equation}
The first derivative of $I_1$ is already evaluated in \eqref{eq:I_1_prime_id_more}, and the second derivative can be computed similarly. We consider $I_{1, 1}, I_{1, 2}, I_{1, 3}$ individually, and have
\begin{equation} \label{eq:2nd_deriv_I_1_pre}
  \frac{d^2}{dz^2} I_{1, j} = \int_{\Gamma_j} h^{(j)}_1(\zeta) \left(\frac{d^2}{dz^2}e^{\frac{2iz\z}{C}}\right) G_1(\zeta) dz = \frac{-4}{C^2} \int_{\Gamma_j} h^{(j)}_1(\zeta) \zeta^2 e^{\frac{2iz\z}{C}} G_1(\zeta) d\zeta.
\end{equation}
Now we use the property that $h^{(j)}(\zeta)$ is a solution to \eqref{eq:diff_eq_4d}, and have
\begin{equation}\label{diffeq:7}
  \z^2 h^{(j)}_1(\z) = \frac{i}{4}\left[ \frac{d}{d\zeta} h^{(j)}_1(\z) + i(\sg + 2q(\sg)^2) h^{(j)}_1(\z) - (4\z q(\sg) + 2i q'(\sg)) h^{(j)}_2(\z) \right].
\end{equation}
We therefore have, using \eqref{eq:defn_of_G_1-G_4},
\begin{multline} \label{eq:2nd_deriv_I_1}
  \frac{d^2}{dz^2} I_{1, j} = \frac{1}{C^2} \left[ (\sigma + 2q(\sigma)^2) I_{1, j} - 2q'(\sigma) \frac{\gamma_1 \sqrt{r_2}}{\gamma_2 \sqrt{r_1}} I_{2, j} + 2Cq(\sigma) \frac{\gamma_1 \sqrt{r_2}}{\gamma_2 \sqrt{r_1}} I_{4, j} \vphantom{\int_{\Gamma_j}} \right. \\
  - \left. i \int_{\Gamma_j} \left( \frac{d}{d\zeta} h^{(j)}_1(\zeta) \right) e^{\frac{2iz\z}{C}}G_1(\zeta) d\zeta \right].
\end{multline}
Furthermore, using integration by parts, we have
\begin{equation}
  \int_{\Gamma_j} \left( \frac{d}{d\zeta} h^{(j)}_1(\zeta) \right) e^{\frac{2iz\z}{C}}G_1(\zeta) d\zeta = -\int_{\Gamma_j} h^{(j)}_1(\zeta) \left[ \frac{d}{d\zeta} \left(e^{\frac{2iz\z}{C}}G_1(\zeta) \right)\right] d\zeta +
  \begin{cases}
    -h^{(j)}_1(0) G_1(0) & j = 1, 2, \\
    h^{(j)}_1(0) G_1(0) & j = 3.
  \end{cases}
\end{equation}
Noting that
\begin{equation}
  -h^{(1)}_1(0) - h^{(2)}_1(0) + h^{(3)}_1(0) = 0,
\end{equation}
we obtain that
\begin{equation}
  \sum^3_{j = 1} \int_{\Gamma_j} \left( \frac{d}{d\zeta} h^{(j)}_1(\zeta) \right) e^{\frac{2iz\z}{C}}G_1(\zeta) d\zeta = \sum^3_{j = 1} \int_{\Gamma_j} h^{(j)}_1(\zeta) \left[ \frac{d}{d\zeta} \left(e^{\frac{2iz\z}{C}}G_1(\zeta) \right)\right] d\zeta,
\end{equation}
and then summing up the $j = 1, 2, 3$ cases of \eqref{eq:2nd_deriv_I_1},
\begin{multline} \label{eq:alt_2nd_deriv_I_1}
  \frac{d^2}{dz^2} I_{1} = \frac{1}{C^2} \left[ (\sigma + 2q(\sigma)^2) I_1 - 2q'(\sigma) \frac{\gamma_1 \sqrt{r_2}}{\gamma_2 \sqrt{r_1}} I_2 + 2Cq(\sigma) \frac{\gamma_1 \sqrt{r_2}}{\gamma_2 \sqrt{r_1}} I_4 \vphantom{\int_{\Gamma_j}}
  \vphantom{i \sum^3_{j = 1} \int_{\Gamma_j} h^{(j)}_1(\zeta) \left( \frac{d}{d\zeta} G_1(\zeta) \right) d\zeta} \right. \\
  + \left. i \sum^3_{j = 1} \int_{\Gamma_j} h^{(j)}_1(\zeta) \left[ \frac{d}{d\zeta} \left(e^{\frac{2iz\z}{C}}G_1(\zeta) \right)\right] d\zeta \right].
\end{multline}
Since
\begin{multline}
  \frac{d}{d\zeta} \left(e^{\frac{2iz\z}{C}}G_1(\zeta)\right) = \left( 3ia \zeta^2 + 2b\zeta + ic + \frac{2iz}{C} \right) e^{\frac{2iz\z}{C}}G_1(\zeta)  \\
  = -\frac{3iC^2 a}{4} \frac{d^2}{dz^2} \left(e^{\frac{2iz\z}{C}}G_1(\zeta)\right) - ibC e^{\frac{2iz\z}{C}}G_3(\zeta) + i \left( c + \frac{2z}{C} \right) e^{\frac{2iz\z}{C}}G_1(\zeta),
\end{multline}
we have, using \eqref{eq:2nd_deriv_I_1_pre}, that
\begin{equation} \label{eq:id_with_2nd_deriv_I_1}
  \sum^3_{j = 1} \int_{\Gamma_j} h^{(j)}_1(\zeta) \left[ \frac{d}{d\zeta} \left(e^{\frac{2iz\z}{C}}G_1(\zeta) \right)\right] d\zeta = -\frac{3iC^2 a}{4} \frac{d^2}{dz^2} I_1 - ibC I_3(\zeta) + i \left( c + \frac{2z}{C} \right) I_1.
\end{equation}
Combining \eqref{eq:alt_2nd_deriv_I_1} and \eqref{eq:id_with_2nd_deriv_I_1}, we solve that
\begin{equation} \label{eq:2nd_deriv_I_1_final}
  \left( 1 - \frac{3a}{4} \right) \frac{d^2}{dz^2} I_{1} = \left( \frac{\sigma + 2q(\sigma)^2 - c}{C^2} - \frac{2z}{C^3} \right) I_1 - \frac{2q'(\sigma)}{C^2} \frac{\gamma_1 \sqrt{r_2}}{\gamma_2 \sqrt{r_1}} I_2 + \frac{b}{C} I_3 + \frac{2q(\sigma)}{C} \frac{\gamma_1 \sqrt{r_2}}{\gamma_2 \sqrt{r_1}} I_4.
\end{equation}
Plugging \eqref{eq:2nd_deriv_I_1_final} and \eqref{eq:I_1_prime_id} into \eqref{eq:m''_1(z)_start}, and using the formulas \eqref{int:6} and \eqref{diffeq:2} for the coefficients, we have
\begin{multline} \label{eq:express_m''_2}
  m''_1(z) = \left( \tau^2 \frac{r^2_2 - 3r^2_1}{r^2_1 + r^2_2} + 2r_1s_1 + r^2_1 C q(\sigma)^2 - r^2_1 z \right) m_1 - r^2_1 C \frac{\gamma_1 \sqrt{r_2}}{\gamma_2 \sqrt{r_1}} q'(\sigma) m_2 \\
  + 2\tau m_3 + r^2_1 C^2 \frac{\gamma_1 \sqrt{r_2}}{\gamma_2 \sqrt{r_1}} q(\sigma) m_4.
\end{multline}
On the other hand, among the terms on the right-hand side of \eqref{int:13}, $m_1(z), m_2(z), m'_1(z)$ are already evaluated in \eqref{eq:formulas_for_m_1_m_2_m_1'} and \eqref{eq:m'_1(z)_final}, while $m'_2(z)$ can be evaluated similar to $m'_1(z)$ as
\begin{equation} \label{eq:m'_2(z)_final}
  m'_1(z) = -\tau \frac{r^2_1 - r^2_2}{r^2_1 + r^2_2} m_2(z) + m_4(z). 
\end{equation}
It is not hard to see that the right-hand side of \eqref{int:13} can also be expressed as the right-hand side of \eqref{eq:express_m''_2}. Thus we prove \eqref{int:13}. In the same way we can prove \eqref{int:13_2}.

\section{Proof of Proposition \ref{thm:main_asy}}\label{asymptotics}

Since parts \ref{enu:thm:main_asy_1} -- \ref{enu:thm:main_asy_6} are similar, we prove part \ref{enu:thm:main_asy_1} in detail in Section \ref{subsec:part_1_of_1.3}, and explain how the proof is adapted to other cases in Section \ref{subsec:other_parts_1.3}. Parts \ref{enu:thm:main_asy_2} and \ref{enu:thm:main_asy_3} can be proved by the computation as in part \ref{enu:thm:main_asy_1}. For parts \ref{enu:thm:main_asy_4}, \ref{enu:thm:main_asy_5}, and \ref{enu:thm:main_asy_6}, although the same method works, the computation should be adjusted because $f^{(k)}, g^{(k)}$ in \eqref{eq:def_of_nj2} and $f^{(k)} + g^{(k)}$ have different asymptotic behavior at $\infty$ for even and odd $k$.

In the proof of parts \ref{enu:thm:main_asy_1}, \ref{enu:thm:main_asy_2}, and \ref{enu:thm:main_asy_3}, for a computational reason that will be clear later, we take a change of variable
\begin{equation} \label{eq:zeta_xi_change}
  \zeta = \xi + \frac{ib}{3a + 4} = \xi + \frac{i\tau}{C^2 r^2_1},
\end{equation}
where $a, b$ are defined in \eqref{diffeq:2}, $\tau$ and $r_1$ are defined in \eqref{int:5}, and $C$ is defined in \eqref{int:6}, and define the cubic polynomial $F$ as
\begin{multline} \label{eq:defn_F(xi)}
  F(\xi) = i\tilde{a} \xi^3 + i\tilde{c} \xi, \quad \text{where}  \quad   \tilde{a} = a + \frac{4}{3} = \frac{8 r^2_1}{3(r^2_1 + r^2_2)} , \\
   \tilde{c}\equiv \tilde{c}(z) = \frac{b^2}{3a + 4} + c + \frac{2z}{C} + \sigma = \frac{2z + 4s_2/r_2}{C},
\end{multline}
and $c$ is defined in \eqref{diffeq:2}, $\sigma$ is defined in \eqref{int:9}, and $r_2, s_2$ are defined in \eqref{int:5}. We note that the leading coefficient of $F$ satisfies $\tilde{a} > 0$. We are interested in the asymptotics of the functions $n^{(0)}(z), \dots, n^{(5)}(z)$ as $z\to \infty$ in various sectors of the complex plane. Note that as $z\to\infty$ at a certain angle, the parameter $\tilde{c} \equiv \tilde{c}(z)$ also approaches $\infty$ at the same angle. Thus we will consider the asymptotic behavior of the integrals which define $n^{(0)}(z), \dots, n^{(5)}(z)$ as $\tilde{c} \to \infty$. For brevity we will often use the notation $\tilde{c}$ rather than $\tilde{c}(z)$, and we trust the reader can keep in mind that $\tilde{c}$ is related to $z$ by a scaling and shift of fixed size.

\begin{rem} \label{rem:meaning_of_F}
  The function $F(\xi)$ will be useful in the proof of parts \ref{enu:thm:main_asy_1}, \ref{enu:thm:main_asy_2}, and \ref{enu:thm:main_asy_3}, because the essential part of asymptotic analysis is the integrals on $\Gamma^{(k)}_1 \cup \Gamma^{(k)}_3$ ($k = 0, 2, 4$), where the contours $\Gamma^{(k)}_j$ are deformed, as explained later in this section. The integrands on $\Gamma^{(k)}_1 \cup \Gamma^{(k)}_3$, although various in explicit formulas, all have the asymptotic behavior
  \begin{equation} \label{eq:asy_integrands_parts_123}
    e^{\frac{4}{3} i\zeta^3 + i \sigma \zeta + \frac{2iz\zeta}{C}} G(\zeta) \times (\text{factor growing at most linearly at $\infty$}),
  \end{equation}
  and under the change of variable \eqref{eq:zeta_xi_change},
  \begin{equation} \label{eq:alt_defn_F}
    \log \left( e^{\frac{4}{3} i\zeta^3 + i \sigma \zeta + \frac{2iz\zeta}{C}} G(\zeta) \right) =   F(\xi) - \log \gamma_2-  \frac{2r^2_2}{r^2_1 + r^2_2} \tau z,
  \end{equation}
  where
 $\gamma_2$ is defined in \eqref{diffeq:2}.
\end{rem}

Recall the sectors $\Delta_0, \dotsc, \Delta_5$ defined in \eqref{def:De_j}. In what follows we consider them on the $\zeta$-plane and $\xi$-plane by replacing $z$ by $\zeta$ and $\xi$ respectively in their definitions.

\subsection{Proof of part \ref{enu:thm:main_asy_1}} \label{subsec:part_1_of_1.3}

Note that if $z = re^{i\theta}$ where $\theta \in (-\pi/3 + \delta, \pi/3 - \delta)$, then for large enough $r$, $\tilde{c}$ defined in \eqref{eq:defn_F(xi)} has its argument in a compact subset of $(-\pi/3, \pi/3)$. Below in the proof we assume that
\begin{equation} \label{eq:arg_tilde_c}
  \arg(\tilde{c}) \in [-\pi/3 + \delta', \pi/3 - \delta'], \quad \delta' > 0,
\end{equation}
even if $\lvert \tilde{c} \rvert$ is not large. To be concrete, we may take $\delta' = \delta/2$.

Before giving the rigorous argument of the proof, we describe the strategy.

\begin{enumerate}[label=\bf Step \arabic*, leftmargin=0pt, labelwidth=!, align=left]
\item \label{enu:Step_1}
  Find the critical points of $F(\xi)$. There are two of them, which are denoted as $\xi_+$ (on the upper-half plane) and $\xi_-$ (on the lower-half plane). Then denote
  \begin{equation}
    \zeta_{\pm} = \xi_{\pm} + \frac{ib}{3a + 4}.
  \end{equation}
\item \label{enu:Step_2}
  Deform the contour $\Gamma^{(k)} = \Gamma^{(k)}_1 \cup \Gamma^{(k)}_2 \cup \Gamma^{(k)}_3$ defined in \eqref{eq:defn_Gamma^(k)} for $k = 0, 2, 4$, such that $\Gamma^{(0)}_1 \cup \Gamma^{(0)}_3$ is a contour from $e^{5\pi i/6} \cdot \infty$ to $e^{\pi i/6} \cdot \infty$ and passes through $\zeta_+$, $\Gamma^{(2)}_1 \cup \Gamma^{(2)}_3$ is a contour from $e^{\pi i/6} \cdot \infty$ to $e^{-\pi i/2} \cdot \infty$ and passes through $\xi_-$, and $\Gamma^{(4)}_1 \cup \Gamma^{(4)}_3$ is a contour from $e^{-\pi i/2} \cdot \infty$ to $e^{5\pi i/6} \cdot \infty$ and passes through $\xi_-$. Then $\Gamma^{(k)}_2$ goes from a point on $\Gamma^{(k)}_1 \cup \Gamma^{(k)}_3$ to $e^{(1/2 - k/3)\pi i} \cdot \infty$, for $k = 0, 2, 4$. Furthermore, we require that for $\lvert z \rvert$ large enough,
  \begin{align}
    \Gamma^{(0)}_1 \in {}& \{ \epsilon \leq \arg \zeta \leq 2\pi/3 - \epsilon \}, & \Gamma^{(0)}_3 \in {}& \{ \pi/3 + \epsilon \leq \arg \zeta \leq \pi - \epsilon \}, \label{eq:condition_Gamma_0} \\
    \Gamma^{(2)}_1 \in {}& \{ 4\pi/3 + \epsilon \leq \arg \zeta \leq 2\pi - \epsilon \}, & \Gamma^{(2)}_3 \in {}& \{ -\pi/3 + \epsilon \leq \arg \zeta \leq \pi/3 - \epsilon \}, \label{eq:condition_Gamma_2} \\
    \Gamma^{(4)}_1 \in {}& \{ 2\pi/3 + \epsilon \leq \arg \zeta \leq 4\pi/3 - \epsilon \}, & \Gamma^{(4)}_3 \in {}& \{ \pi + \epsilon \leq \arg \zeta \leq 5\pi/3 - \epsilon \}, \label{eq:condition_Gamma_4}
  \end{align}
  \begin{equation} \label{eq:condition_Gamma_middles}
    \Gamma^{(0)}_2 \in \{ \frac{\pi}{3} + \epsilon \leq \arg \zeta \leq \frac{2\pi}{3} - \epsilon \}, \quad \Gamma^{(2)}_2 \in \{ -\frac{\pi}{3} + \epsilon \leq \arg \zeta \leq -\epsilon \}, \quad \Gamma^{(4)}_2 \in \{ \pi + \epsilon \leq \arg \zeta \leq \frac{4\pi}{3} - \epsilon \},
  \end{equation}
  and
  \begin{equation} \label{eq:Gamma_away_from_0}
    \dist(\Gamma^{(k)}, 0) > \epsilon \lvert z \rvert^{1/2},
  \end{equation}
  where $\epsilon > 0$ is a constant depending on $\delta$. Note that we only define $\Gamma^{(2)}_2$ for $\arg z \geq 0$ and only define $\Gamma^{(4)}_2$ for $\arg z \leq 0$. For the saddle point analysis, we require
  \begin{equation} \label{eq:saddle_pt_property_of_Gamma}
    \Re \log \left( e^{\frac{4}{3} i\zeta^3 + i \sigma \zeta + \frac{2iz\zeta}{C}} G(\zeta) \right) \text{ attains its maximum on $\Gamma^{(k)}_1 \cup \Gamma^{(k)}_3$ at $\zeta_{\pm}$, \quad }k = 0, 2, 4,
  \end{equation}
  where $\pm$ is $+$ for $k = 1$ and $-$ for $k = 2, 4$. See Figure \ref{fig:contours_part_1} for a schematic graph of the contours. The existence of the contours will be carefully justified later.
  \begin{figure}[htb]
    \begin{minipage}[t]{0.3\linewidth}
      \centering
      \includegraphics{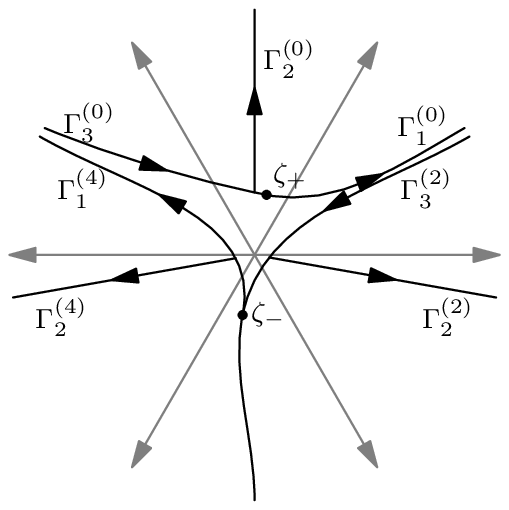}
      \caption{Schematic graphs of $\Gamma^{(0)}$, $\Gamma^{(2)}$ and $\Gamma^{(4)}$, in the proof of part \ref{enu:thm:main_asy_1} of Proposition \ref{thm:main_asy}. $\Gamma^{(2)}_1$ and $\Gamma^{(4)}_3$ are not labelled, because their major parts overlap.}
      \label{fig:contours_part_1}
    \end{minipage}
    \hspace{\stretch{1}}
    \begin{minipage}[t]{0.3\linewidth}
      \centering
      \includegraphics{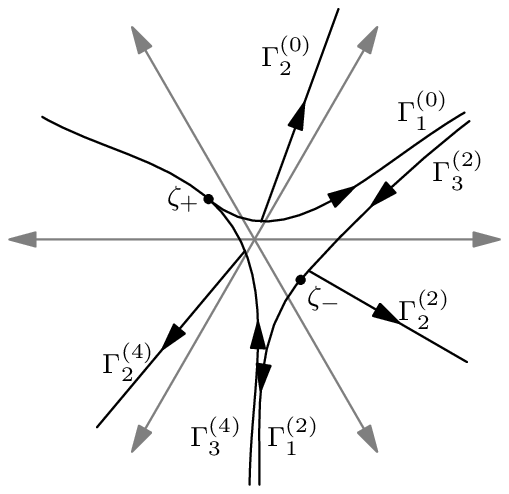}
      \caption{Schematic graphs of $\Gamma^{(0)}$, $\Gamma^{(2)}$ and $\Gamma^{(4)}$, in the proof of part \ref{enu:thm:main_asy_2} of Proposition \ref{thm:main_asy}. $\Gamma^{(4)}_1$ and $\Gamma^{(0)}_3$ are not labelled, because their major parts overlap.}
      \label{fig:contours_part_2}
    \end{minipage}
    \hspace{\stretch{1}}
    \begin{minipage}[t]{0.3\linewidth}
      \centering
      \includegraphics{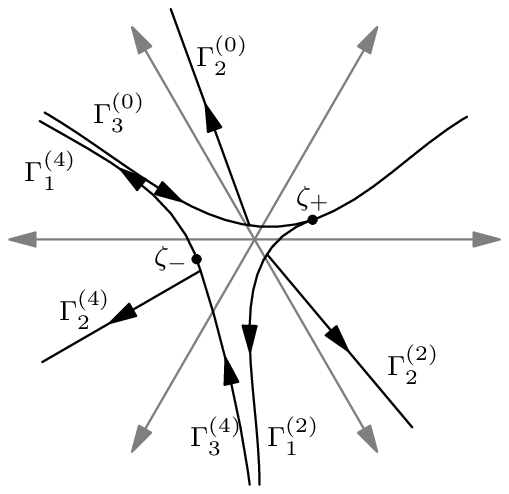}
      \caption{Schematic graphs of $\Gamma^{(0)}$, $\Gamma^{(2)}$ and $\Gamma^{(4)}$, in the proof of part \ref{enu:thm:main_asy_3} of Proposition \ref{thm:main_asy}. $\Gamma^{(0)}_1$ and $\Gamma^{(2)}_3$ are not labelled, because their major parts overlap.}
      \label{fig:contours_part_3}
    \end{minipage}
  \end{figure}
\item \label{enu:Step_3}
  Use the standard saddle point analysis to prove the result. In particular, for all integrals on $\Gamma^{(k)}_1$ and $\Gamma^{(k)}_3$ ($k = 0, 2, 4)$ in the entries of the $4 \times 4$ matrix shown in \eqref{Q_def}, the integrands are expressed as $(\text{linear polynomial in } \zeta \text{ plus } \bigO(\zeta^{-1})) \times e^{F(\xi)}$ where $\xi$ is related to $\zeta$ by \eqref{eq:zeta_xi_change} and $F(\xi)$ is given in \eqref{eq:defn_F(xi)} and \eqref{eq:alt_defn_F}. Thus $\zeta_{\pm}$ are the saddle points giving the major contributions to the integrals. We also show that the integrals over $\Gamma^{(k)}_2$ are negligible.
\end{enumerate}
Below we give the details of the steps.

\subsubsection{\ref{enu:Step_1}: Critical points}

For each $z \in \compC$, the equation $\frac{dF}{d\xi} = 0$ has two solutions,
\begin{equation} \label{eq:xi_pm}
  \xi_{\pm} = \pm i \sqrt{\frac{\tilde{c}(z)}{3\tilde{a}}}.
\end{equation}
By \eqref{eq:arg_tilde_c}, we have $\xi_+ \in \Delta_1$ and $\xi_- \in \Delta_4$.

\subsubsection{\ref{enu:Step_2}(a): Preliminary lemmas}

For the construction of the contours, we are going to use some planar dynamical system techniques. We interpret the complex $\xi$-plane as a two-dimensional real coordinate plane by the standard relation $\xi = x + yi$. The function $\Re F(\xi)$ is then harmonic in $\xi$, or equivalently in $x, y$, and it has only two critical points $\xi_{\pm}$. By condition \eqref{eq:arg_tilde_c}, we have
\begin{equation} \label{eq:inequality_critical_pt}
  \Re F(\xi_+) < \Re F(0) < \Re F(\xi_-).
\end{equation}
Consider the curve $L_0$ with differentiable parametrization $(x(t), y(t))$ such that
\begin{equation}
  \left( \frac{\partial}{\partial x} F(x(t) + iy(t)), \frac{\partial}{\partial x} F(x(t) + iy(t)) \right) \cdot
  \begin{pmatrix}
    x'(t) \\
    y'(t)
  \end{pmatrix}
  = 0, \quad \text{and} \quad x(0) = y(0) = 0.
\end{equation}
This curve is the \emph{level curve through $0$}. Since by \eqref{eq:inequality_critical_pt} $\xi_{\pm}$ are not on this level curve, the level curve can be extended to $\infty$ in both directions, and we assume it below. Then we have the following result on the directions that $L_0$ approaches $\infty$. A numerical plotting of $L_0$ is shown in Figure \ref{fig:L+-}. This plot demonstrates the following result.

\begin{lem} \label{lem:level_curve_0}
  $L_0$ lies in $\overline{\Delta_0} \cup \overline{\Delta_5} \cup \overline{\Delta_2} \cup \overline{\Delta_3}$, and it goes to $\infty$ in directions $e^0 \cdot \infty$ and $e^{\pi i} \cdot \infty$.
\end{lem}
\begin{proof}
  By the behavior of $\Re F(\xi)$ at $\infty$, we know that a level curve, on which $\Re F(\xi)$ is finite, can only go to $\infty$ in six possible directions: $k\pi/3$, $k = 0, \dotsc, 5$. For $L_0$, we also know that the tangent direction at $0$ is $-\arg \tilde{c} \in (\pi/3 + \delta', \pi/3 - \delta')$ and $\pi - \arg \tilde{c}$.

  In the remaining part of the proof, we consider three cases separately:
  \begin{enumerate*}[label=(\alph*)]
  \item \label{enu:case_3_zero}
    $\arg \tilde{c} = 0$, \linebreak
  \item \label{enu:case_1_neg}
    $\arg \tilde{c} \in (-\pi/3, 0)$, and
  \item \label{enu:case_2_pos}
    $\arg \tilde{c} \in (0, \pi/3)$.
  \end{enumerate*}

  In Case \ref{enu:case_3_zero}, $L_0$ is exactly the real axis and the result of the lemma is obvious.
 
  In Case \ref{enu:case_1_neg}, we have that one part of $L_0$ goes from $0$ to sector $\Omega_0$ and the other part goes from $0$ to sector $\Omega_3$. We denote them $L_{0, +}$ and $L_{0, -}$ respectively. We can see that $L_{0, +}$ does not go out of $\Omega_0$, because on one boundary of $\Omega_0$, $\{ \xi \neq 0 \mid \arg \xi = 0 \}$, $\Re F(\xi) > \Re F(0)$, and on the other boundary of $\Omega_0$, $\{ \xi \neq 0 \mid \arg \xi = \pi/3 \}$, $\Re F(\xi) < \Re F(0)$. Similarly, $L_{0, -}$ does not go out of $\Omega_3$.
  
  Now the possible directions for $L_{0, +}$ to approach $\infty$ is limited to $0$ and $\pi/3$. Next we exclude $\pi/3$. For any $\epsilon > 0$, we have by direct calculation that $\Re F(\rho e^{\alpha i}) < \Re F(0)$ for all $\alpha \in (\pi/3 - \epsilon, \pi/3)$ and large enough $\rho$, so for $\xi \in L_{0, +}$, $\arg \xi \notin (\pi/3 - \epsilon, \pi)$ if $\lvert \xi \rvert$ is large enough. Thus $L_{0, +}$ goes to $e^0 \cdot \infty$. By a similar reason, $L_{0, -}$ goes to $e^{\pi i} \cdot \infty$.

  Case \ref{enu:case_2_pos} is converted to Case \ref{enu:case_1_neg} by the change of variables $\xi \to \bar{\xi}$.
\end{proof}

The next technical lemma is proved by straightforward calculation.
\begin{lem} \label{lem:monotone}
  Let $\rho$ be a big enough positive number. Then on the circle $\{ \rho e^{i\alpha} \mid 0 \leq \alpha < 2\pi \}$, $\Re F(\xi)$ has three local maxima,  $z_2(\rho) = \rho e^{\pi i/2} + \bigO(\rho^{-1})$, $z_4(\rho) = \rho e^{7\pi i/6} + \bigO(\rho^{-1})$ and $z_6(\rho) = \rho e^{-\pi i/6} + \bigO(\rho^{-1})$, and three local minima, around $z_1(\rho) = \rho e^{\pi i/6} + \bigO(\rho^{-1})$, $z_3(\rho) = \rho e^{5\pi i/6} + \bigO(\rho^{-1})$ and $z_5(\rho) = \rho e^{-\pi i/2} + \bigO(\rho^{-1})$. Furthermore, on each arc $A_k(\rho)$ between $z_{k - 1}(\rho)$ and $z_k(\rho)$, ($k = 1, \dotsc, 6$ and $z_0(\rho) = z_6(\rho)$), the value of $\Re F(\xi)$ is monotonic as $\xi$ moves along the arc.
\end{lem}

Now consider the level curves through $\xi_+$ and $\xi_-$, which we denote by $L_+$ and $L_-$ respectively. Note that locally around $\xi_{\pm}$, $L_{\pm}$ is the union of two smooth local level curves, and $L_{\pm}$ goes to $\infty$ as it extends along the four ends of the smooth local level curves. See Figure \ref{fig:L+-} for a numerical plotting of these level curves. The plotting demonstrates the results of the following two lemmas.
\begin{lem} \label{lem:4_crosses}
  \begin{enumerate}[label=(\alph*)]
  \item \label{enu:level_curves_through_crit_pt_1}
    The branches of $L_+$ go to $e^0 \cdot \infty$, $e^{\pi i/3} \cdot \infty$, $e^{2\pi i/3} \cdot \infty$ and $e^{\pi i} \cdot \infty$, and we denote them $L_{+, 1}$, $L_{+, 2}$, $L_{+, 3}$, and $L_{+, 4}$ respectively.
  \item \label{enu:level_curves_through_crit_pt_2}
    The branches of $L_-$ go to $e^0 \cdot \infty$, $e^{-\pi i/3} \cdot \infty$, $e^{-2\pi i/3} \cdot \infty$ and $e^{\pi i} \cdot \infty$, and we denote them $L_{-, 1}$, $L_{-, 2}$, $L_{-, 3}$, and $L_{-, 4}$ respectively.
  \end{enumerate}
\end{lem}
\begin{figure}[htb]
  \begin{minipage}[t]{0.3\linewidth}
    \centering
    \includegraphics[width=\linewidth]{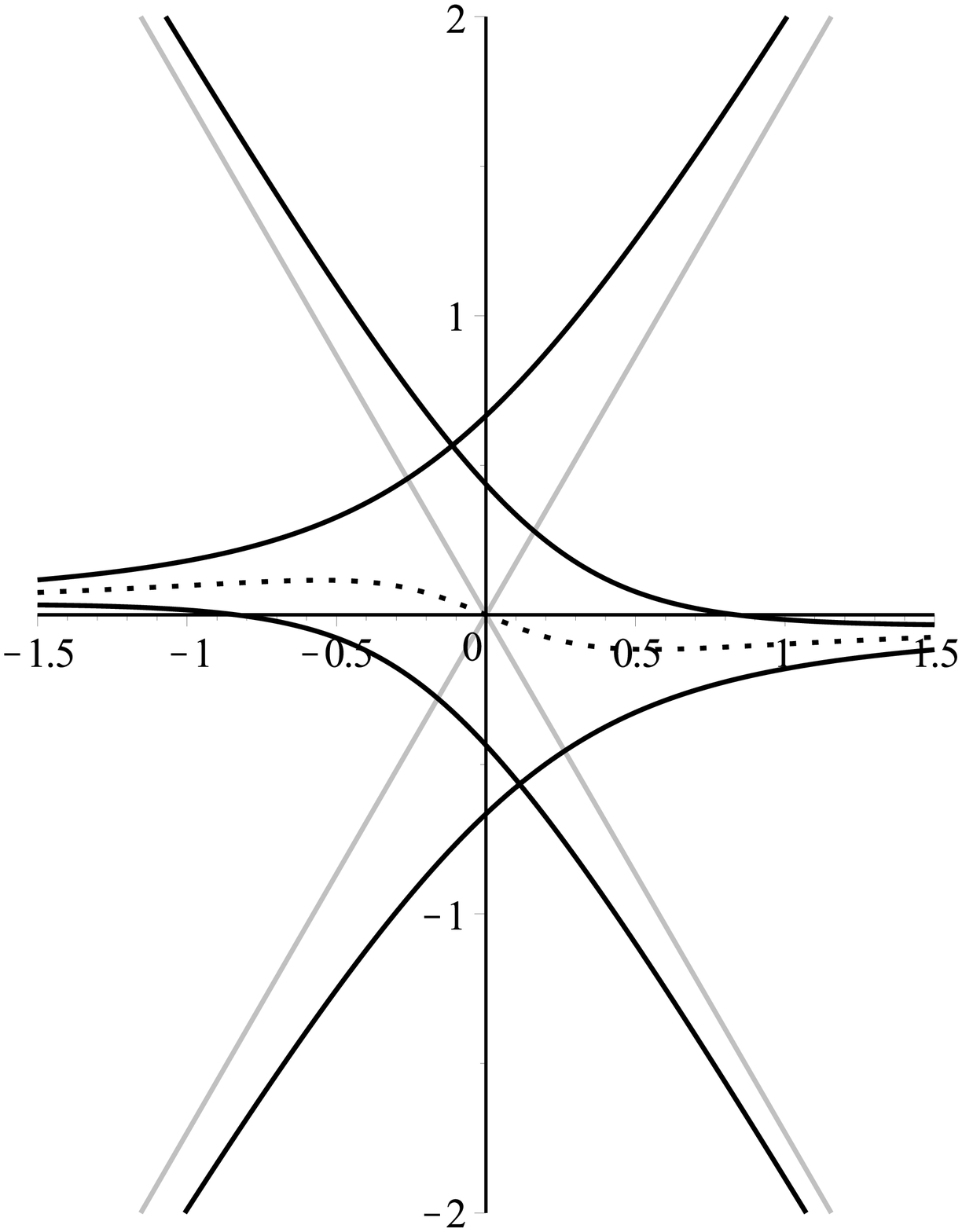}
    \caption{The dotted curve is $L_0$. $L_+$ consists of the curves above $L_0$, and $L_-$ consists of the curves below $L_0$. Here are $\tilde{a} = 1$ and $\tilde{c} = \pi/4$.}
    \label{fig:L+-}
  \end{minipage}
  \hspace{\stretch{1}}
  \begin{minipage}[t]{0.3\linewidth}
    \centering
    \includegraphics[width=\linewidth]{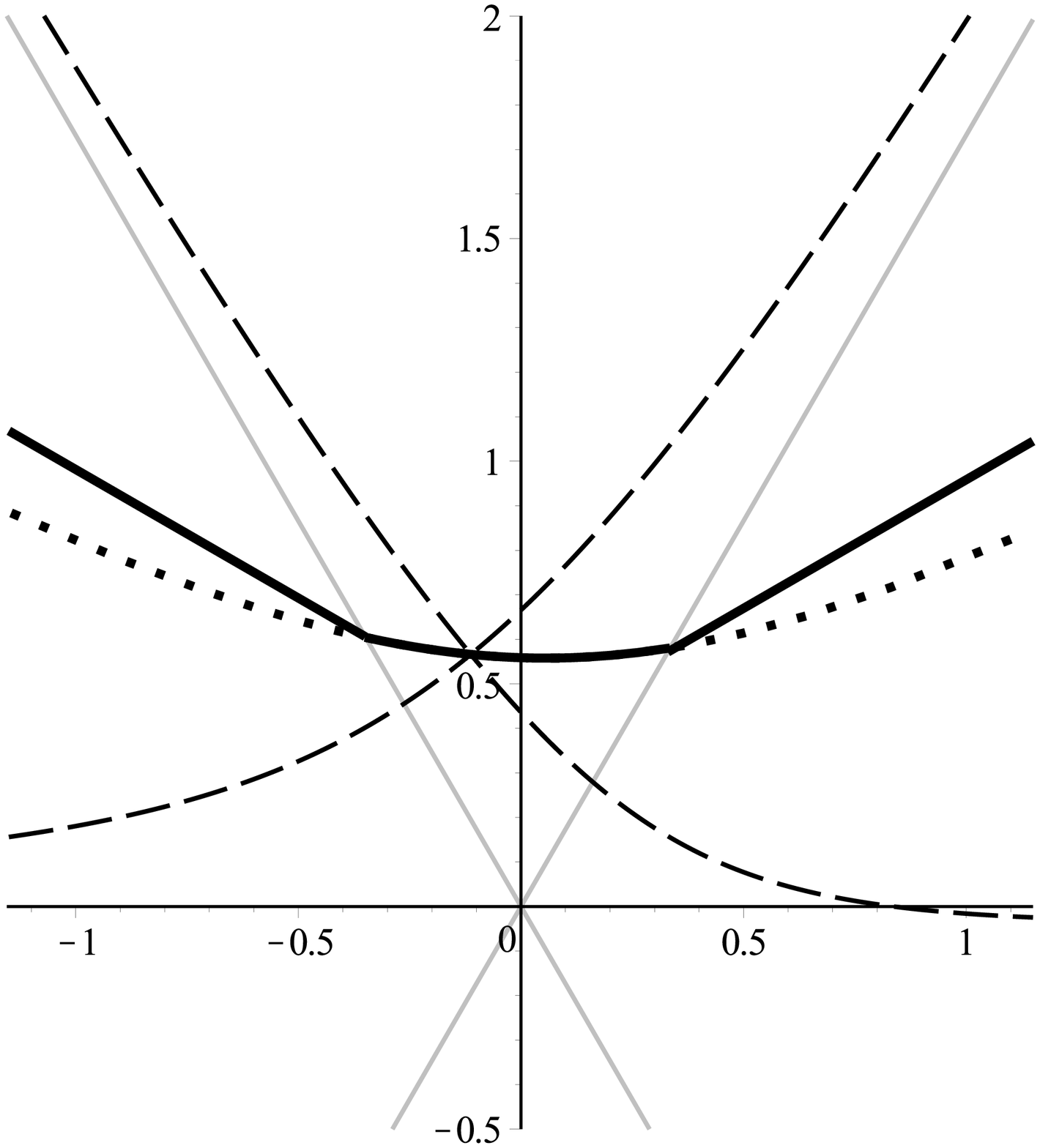}
    \caption{The solid curve is $\tilde{\Gamma}^{(0)}_1 \cup \tilde{\Gamma}^{(0)}_3$. The dotted curves are $\gamma_{+, 1}$ and $\gamma_{+, 2}$, and the dashed curves are $L_+$. Here $\tilde{a} = 1$ and $\tilde{c} = e^{\pi i/8}$.}
    \label{fig:Gamma_0}
  \end{minipage}
  \hspace{\stretch{1}}
  \begin{minipage}[t]{0.3\linewidth}
    \centering
    \includegraphics[width=\linewidth]{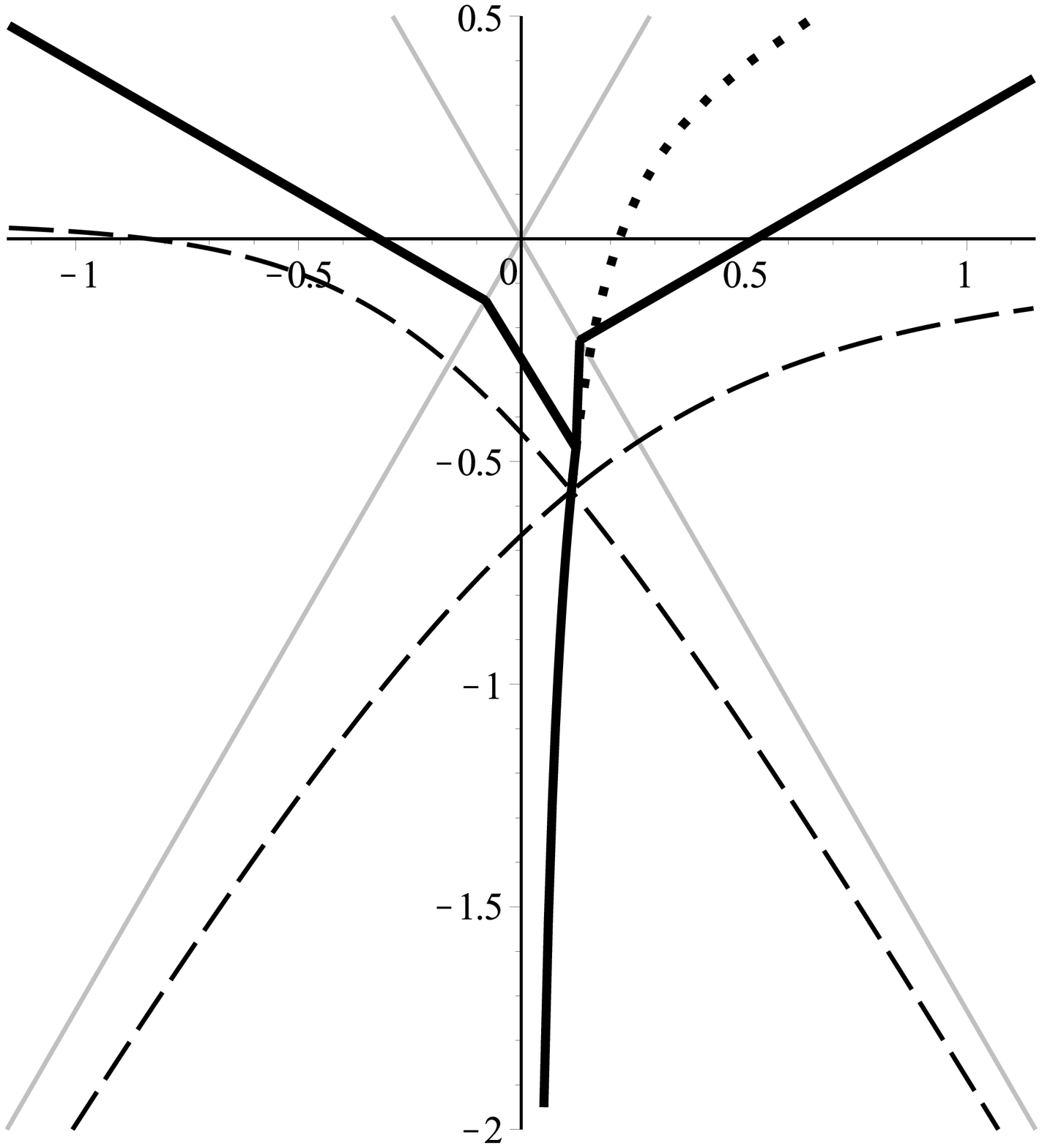}
    \caption{The solid curves are $\tilde{\Gamma}^{(2)}_1 \cup \tilde{\Gamma}^{(2)}_3$ and $\tilde{\Gamma}^{(4)}_1 \cup \tilde{\Gamma}^{(4)}_3$. (They overlap in the bottom part.) The dotted curve is $\gamma_{-, 2}$, and the dashed curves are $L_-$. Here $\tilde{a} = 1$ and $\tilde{c} = e^{\pi i/8}$.}
    \label{fig:Gamma_24}
  \end{minipage}
\end{figure}
\begin{proof}
  We give the proof to part \ref{enu:level_curves_through_crit_pt_1}, and that to part \ref{enu:level_curves_through_crit_pt_2} is analogous.
  
  By the argument in the beginning of the proof of Lemma \ref{lem:level_curve_0}, we know that the local level curves through $\xi_+$ go to infinity in the directions $k\pi/3$. By inequality \eqref{eq:inequality_critical_pt}, we know that $L_+$, the level curve through $\xi_+$, does not intersect $L_0$. Since we assume that $\xi_+$ is in sector $\Omega_1$, we have that $L_+$ can only go to $\infty$ in directions above $L_0$, that is, $0, \pi/3, 2\pi/3, \pi$.

  Recall the notations in Lemma \ref{lem:monotone}. Let $\rho$ be a large enough positive number, then $L_+$ intersects the circle $\{ \lvert \xi \rvert = \rho \}$ at four distinct points. These intersecting points are on $A_1(\rho) \cup A_2(\rho) \cup A_3(\rho) \cup A_4(\rho)$, since they are above $L_0$. By the monotone property stated in Lemma \ref{lem:monotone}, these intersecting points are on distinct arcs $A_k(\rho)$, and then they are around $\rho$, $\rho e^{\pi i/3}$, $\rho e^{2\pi i/3}$ and $-\rho$ respectively. Thus $L_+$ goes to infinity in the four distinct directions.
\end{proof}

\begin{lem} \label{lem:L+-_intersections}
  \begin{enumerate}[label=(\alph*)]
  \item \label{enu:lem:L+-_intersections_a}
    $L_{+, 1}$ intersects with the ray $\{ \arg \xi = \pi/3 \}$ at a point, which we denote by $\xi_1$; $L_{+, 2}$ stays in sector $\Delta_1$ and has the ray as its asymptote, but does not intersect it.
  \item \label{enu:lem:L+-_intersections_b}
    $L_{+, 4}$ intersects with the ray $\{ \arg \xi = 2\pi/3 \}$ at a point, which we denote by $\xi_2$; $L_{+, 3}$ stays in sector $\Delta_1$ and has the ray as its asymptote, but does not intersect it.
  \item \label{enu:lem:L+-_intersections_c}
    $L_{-, 1}$ intersects with the ray $\{ \arg \xi = 5\pi/3 \}$ at a point, which we denote by $\xi_3$; $L_{-, 2}$ stays in sector $\Delta_4$ and has the ray as its asymptote, but does not intersect it.
  \item \label{enu:lem:L+-_intersections_d}
    $L_{-, 4}$ intersects with the ray $\{ \arg \xi = 4\pi/3 \}$ at a point, which we denote by $\xi_4$; $L_{-, 3}$ stays in sector $\Delta_4$ and has the ray as its asymptote, but does not intersect it.
  \end{enumerate}
\end{lem}
\begin{proof}
  We prove parts \ref{enu:lem:L+-_intersections_a} and \ref{enu:lem:L+-_intersections_b}, and the proof to parts \ref{enu:lem:L+-_intersections_c} and \ref{enu:lem:L+-_intersections_d} is similar.

  We note that $\Re F(\xi)$ is monotonically decreasing as $\xi$ moves to $\infty$ on either the ray $\{ \arg \xi = \pi/3 \}$ or $\{ \arg \xi = 2\pi /3 \}$, which are the two boundaries of $\Delta_1$. So $L_+ = L_{+, 1} \cup L_{+, 2} \cup L_{+, 3} \cup L_{+, 4}$ intersects either ray at one point at most. Since $L_{+, 1}$ goes from $\xi_+$ to $e^0 \cdot \infty$, $L_{+, 4}$ goes from $\xi_+$ to $e^{\pi i} \cdot \infty$, and they do not intersect, we have that $L_{+, 1}$ intersects with the ray $\{ \arg \xi = \pi/3 \}$, at a point, and $L_{+, 4}$ intersects with the ray $\{ \arg \xi = 2\pi/3 \}$ at a point. The results for $L_{+, 2}$ and $L_{+, 3}$ are deduced by their asymptotic property in Lemma \ref{lem:4_crosses} and the fact that they do not intersect with the two rays.
\end{proof}

\subsubsection{\ref{enu:Step_2}(b): Construction of $\Gamma^{(0)}$, $\Gamma^{(2)}$ and $\Gamma^{(4)}$} \label{subsubsec:construction_of_contours}

We consider the images of $\Gamma^{(0)}, \Gamma^{(2)}, \Gamma^{(4)}$ and their components under the change of variables \eqref{eq:zeta_xi_change}, which are
\begin{equation} \label{eq:Sigma_in_Gamma}
  \tilde{\Gamma}^{(k)}_* = \left\{ \xi \in \compC \mid \xi + \frac{ib}{3a + 4} \in \Gamma^{(k)}_* \right\}, \quad k = 0, 2, 4, \quad * = 1, 2, 3 \text{ or blank}.
\end{equation}
The construction of $\tilde{\Gamma}^{(k)}$ is equivalent to the construction of $\Gamma^{(k)}$.

The basic ingredient for the construction of the contours are the \emph{flow curves} with respect to the gradient field $\nabla \Re F$, i.e., smooth curves with parametrization $(x(t), y(t))$ such that
\begin{equation}
  (x'(t), y'(t)) = \left( \frac{\partial}{\partial x} F(x(t) + iy(t)), \frac{\partial}{\partial x} F(x(t) + iy(t)) \right).
\end{equation}
Through any point where $\nabla \Re F$ does not vanish, there is a unique flow curve. But from the critical point $\xi_+$ or $\xi_-$, there are four flow curves connecting to $\xi_{\pm}$, with two flowing out of $\xi_{\pm}$, that is, $\Re F$ increases along the flow curves away from $\xi_{\pm}$, and two flowing in $\xi_{\pm}$, that is, $\Re F$ increases along the flow curves towards $\xi_{\pm}$. Generically the flow lines connecting $\xi_{\pm}$ can be extended to $\infty$, but in some special cases a flow line may connect $\xi_+$ and $\xi_-$. If a flow curve extends to $\infty$, then it can either go into $\infty$ in directions $e^{\pi i/2} \cdot \infty$, $e^{7\pi i/6} \cdot \infty$, and $e^{-\pi i/6} \cdot \infty$, or go out of $\infty$ in directions $e^{\pi i/6} \cdot \infty$, $e^{5\pi i/6} \cdot \infty$ and $e^{-\pi i/2} \cdot \infty$.

Around $\xi_{\pm}$, the flow curves into and out of $\xi_{\pm}$ alternate with the level curves $\Re F(\xi) = \Re F(\xi_{\pm})$. Part \ref{enu:level_curves_through_crit_pt_1} of Lemma \ref{lem:4_crosses} shows that one flow curve that flows in $\xi_+$ lies between $L_{+, 1}$ and $L_{+, 2}$, and is from $e^{\pi i/6} \cdot \infty$, and the other flow curve that flows in $\xi_+$ lies between $L_{+, 3}$ and $L_{+, 4}$, and is from $e^{5\pi i/6} \cdot \infty$. We denote them by $\gamma_{+, 1}$ and $\gamma_{+, 2}$ respectively for later use. Part \ref{enu:level_curves_through_crit_pt_2} of Lemma \ref{lem:4_crosses} shows that one flow curve that flows out of $\xi_-$ lies between $L_{-, 2}$ and $L_{-, 3}$, and goes to $e^{-\pi i/2} \cdot \infty$, and the other flow curve that flows out of $\xi_-$ lies above $L_{-, 1}$ and $L_{-, 4}$. We denote them by $\gamma_{-, 1}$ and $\gamma_{-, 2}$ respectively. The flow curve $\gamma_{-, 2}$ may end at $e^{\pi i/6} \cdot \infty$, $e^{5\pi i/6} \cdot \infty$, or $\xi_+$, depending on the argument of $\tilde{c}$, but we do not need this piece of information.

Below we construct $\tilde{\Gamma}^{(0)}, \tilde{\Gamma}^{(2)}$ and $\tilde{\Gamma}^{(4)}$ in the special case that $\tilde{\alpha} = 1$ and $\tilde{c} = e^{i\theta}$ with $\theta \in (-\pi/3 + \delta', \pi/3 - \delta')$. This construction may seem impractical, since our interest is the limiting case that $z \to \infty$, or equivalently, $\tilde{c} \to \infty$. Actually if $\tilde{a} = 1 = \lvert \tilde{c} \rvert = 1$, $\theta = \arg z$ may not satisfy the condition in part \ref{enu:thm:main_asy_1} of Proposition \ref{thm:main_asy}. But the construction for general $\tilde{a}$ and $\tilde{c}$, particularly for large $\tilde{c}$, will be derived by a scaling transform of the special case.

\paragraph{Construction of $\tilde{\Gamma}^{(0)}_1 \cup \tilde{\Gamma}^{(0)}_3$ for $\tilde{a} = 1$ and $\lvert \tilde{c} \rvert = 1$}

Recall that $L_{+, 1}$ intersects with the ray $\{ \arg \xi = \pi/3 \}$ at $\xi_1$. Since the flow curve $\gamma_{+, 1}$ lies above $L_{+, 1}$ and extends to $e^{\pi i/6} \cdot \infty$, it hits the ray $\{ \xi \mid \arg \xi = \pi/3 \text{ and } \lvert \xi \rvert > \lvert \xi_1 \rvert \}$ at a point, which we denote by $\xi'_1$. By the property of flow curve, $\Re F(\xi)$ decreases as $\xi$ moves from $\xi_+$ to $\xi'_1$ along $\gamma_{+, 1}$. A simple calculation shows that as $\xi$ moves to $e^{\pi i/6} \cdot \infty$ along the ray $\{ \xi'_1 + \rho e^{\pi i/6} \mid \rho \geq 0 \}$, $\Re F(\xi)$ is also decreasing. Similarly, $\gamma_{+, 2}$ hits the ray $\{ \xi \mid \arg \xi = 2\pi/3 \text{ and } \lvert \xi \rvert > \lvert \xi_2 \rvert \}$ at a point, which we denote by $\xi'_2$. $\Re F(\xi)$ decreases as $\xi$ moves along $\gamma_{+, 2}$ from $\xi_+$ to $\xi'_2$, and furthermore it decreases as $\xi$ moves to $\infty$ along the ray $\{ \xi'_2 + \rho e^{5\pi i/6} \mid \rho \geq 0 \}$. We define $\tilde{\Gamma}^{(0)}_1 \cup \tilde{\Gamma}^{(0)}_3$ by the concatenation of (i) the ray $\{ \xi'_2 + \rho e^{5\pi i/6} \mid \rho \geq 0 \}$, (ii) the part of $\gamma_{+, 2}$ between $\xi_+$ and $\xi'_2$, (iii) the part of $\gamma_{+, 1}$ between $\xi_+$ and $\xi'_1$, and (iv) the ray $\{ \xi'_1 + \rho e^{\pi i/3} \mid \rho \geq 0 \}$, with the orientation from $e^{5\pi i/6} \cdot \infty$ to $e^{\pi i/6} \cdot \infty$. For a numerical plotting of $\tilde{\Gamma}^{(0)}_1 \cup \tilde{\Gamma}^{(0)}_3$, see Figure \ref{fig:Gamma_0}. 

\begin{rem}
  \begin{itemize}
  \item
    We have not constructed $\tilde{\Gamma}^{(0)}_1$ and $\tilde{\Gamma}^{(0)}_3$ individually yet, since the dividing point between them is not given. 
  \item 
    It seems that we can let $\tilde{\Gamma}^{(0)}_1 \cup \tilde{\Gamma}^{(0)}_3$ simply be $\gamma_{+, 1} \cup \gamma_{+, 2}$. But then it is not easy to show that $\gamma_{+, 1}$ (resp.~$\gamma_{+, 2}$) stay in $\Delta_0 \cup \Delta_1$ (resp.~$\Delta_1 \cup \Delta_2$), and then it is a problem to verify \eqref{eq:condition_Gamma_0} for $\Gamma^{(0)}_1$ and $\Gamma^{(0)}_3$ later.
  \end{itemize}
\end{rem}

\paragraph{Construction of $\tilde{\Gamma}^{(2)}_1 \cup \tilde{\Gamma}^{(2)}_3$ and $\tilde{\Gamma}^{(4)}_1 \cup \tilde{\Gamma}^{(4)}_3$ for $\tilde{a} = 1$ and $\lvert \tilde{c} \rvert = 1$}

First we note that the flow curve $\gamma_{-, 1}$ stays in $\Delta_4$. Next we note that the level curves $L_{-, 1}$ and $L_{-, 4}$, the line segment between $\xi_3$ and $0$, and the line segment between $\xi_4$ and $0$ enclose a region, which we call $R$. On the boundary of $R$, $\Re F(\xi)$ keeps the same on the level curves and decreases as $\xi$ moves above to $0$ along either of the two line segments. The flow curve $\gamma_{-, 2}$ goes into region $R$. Letting $\epsilon$ be a small enough positive constant, we take $\xi'_-$ as the point on $\gamma_{-, 2}$ such that $\lvert \xi'_- - \xi_- \rvert = \epsilon$, and denote the part of $\gamma_{-, 2}$ between $\xi_-$ and $\xi'_0$ by $\gamma_{\epsilon}$. Then there exists a smooth curve lying in region $R$ and connecting $\xi'_-$ and $\xi_3/2$, which we denote by $C_3$, such that $\Re F(\xi)$ decreases monotonically as $\xi$ moves along $C_3$ from $\xi'_-$ to $\xi_3/2$. Similarly, there exists a smooth curve lying in region $R$ and connecting $\xi'_-$ and $\xi_4/2$, which we denote by $C_4$, such that $\Re F(\xi)$ decreases monotonically as $\xi$ moves along $C_4$ from $\xi'_-$ to $\xi'_4/2$. At last, by direct calculation, we find that $\Re F(\xi)$ decreases as $\xi$ moves along the ray $\{ \xi_3/2 + \rho e^{\pi i/6} \mid \rho \geq 0 \}$ to $e^{\pi i/6} \cdot \infty$, and analogously that $\Re F(\xi)$ decreases as $\xi$ moves along the ray $\{ \xi_4/2 + \rho e^{5\pi i/6} \mid \rho \geq 0 \}$ to $e^{5\pi i/6} \cdot \infty$.

Thus we define $\tilde{\Gamma}^{(2)}_1 \cup \tilde{\Gamma}^{(2)}_3$ by the concatenation of (i) the ray $\{ \xi_3/2 + \rho e^{\pi i/6} \mid \rho \geq 0 \}$, (ii) the curve $C_3$, (iii) the curve $\gamma_{\epsilon}$, and (iv) the flow curve $\gamma_{-, 1}$, with the orientation from $e^{\pi i/6} \cdot \infty$ to $e^{-\pi i/2} \cdot \infty$. Similarly, we define $\tilde{\Gamma}^{(4)}_1 \cup \tilde{\Gamma}^{(4)}_3$ by the concatenation of (i) the flow curve $\gamma_{-, 1}$, (ii) the curve $\gamma_*$, (iii) the curve $C_4$, and (iv) the ray $\{ \xi_4/2 + \rho e^{5\pi i/6} \mid \rho \geq 0 \}$, with the orientation from $e^{-\pi i/2} \cdot \infty$ to $e^{5\pi i/6} \cdot \infty$.

For a numerical plotting of $\tilde{\Gamma}^{(2)}_1 \cup \tilde{\Gamma}^{(2)}_3$ and $\tilde{\Gamma}^{(4)}_1 \cup \tilde{\Gamma}^{(4)}_3$, see Figure \ref{fig:Gamma_24}. Note that they have overlap $\gamma_{-, 1} \cup \gamma_{\epsilon}$, which explains the overlap in the schematic Figure \ref{fig:contours_part_1}.

\paragraph{Construction of $\tilde{\Gamma}^{(0)}_2$, $\tilde{\Gamma}^{(2)}_2$ and $\tilde{\Gamma}^{(4)}_2$ for $\tilde{a} = 1$ and $\lvert \tilde{c} \rvert = 1$}

In the construction, we define the function
\begin{equation} \label{eq:defn_hat_F}
  \hat{F}(\xi) = F(\xi) - i\frac{8}{3} \xi^3 = i \left( \tilde{a} - \frac{8}{3} \right) \xi^3 + i\tilde{c} \xi,
\end{equation}
and note that the leading coefficient of $\hat{F}$ satisfies $\tilde{a} - 8/3 < 0$. Notice that in the integral formulas \eqref{Q_def}, integrands on $\Sigma^{(0)}_2$, $\Sigma^{(2)}_2$ and $\Sigma^{(4)}_2$, although different, can all be written in the form of $\exp(-4i \zeta^3/3 -i\sigma\zeta + 2iz\zeta/C) G(\zeta) \times (\text{factor growing at most linearly in $\zeta$})$, and we have
\begin{equation}
  \exp(-4i \zeta^3/3 -i\sigma\zeta + 2iz\zeta/C) G(\zeta) - \hat{F}(\xi) = \exp(\text{quadratic polynomial in $\zeta$}).
\end{equation}

The ray $\{ \rho e^{\pi i/2} \mid \rho > 0 \}$ intersects with $\tilde{\Gamma}^{(0)}_1 \cup  \tilde{\Gamma}^{(0)}_3$ at a point, which we denote by $\xi''_0$. Then we define $\tilde{\Gamma}^{(0)}_2$ to be the ray $\{ \rho e^{\pi i/2} \mid \rho \geq \lvert \xi''_0 \rvert \}$. The following properties can be checked by direct computation: (i) $\tilde{\Gamma}^{(0)}_2$ is contained in $\Omega_1$, and (ii) 
\begin{equation} \label{eq:property_hat_F_Tamma^0_2}
  \Re \hat{F}(\xi''_0) < \Re F(\xi''_0) < \Re F(\xi_+) \quad \text{and} \quad \Re \text{$\hat{F}(\xi)$ decreases as $\xi$ moves along $\tilde{\Gamma}^{(0)}_2$ from $\xi''_0$ to $\infty$}.
\end{equation}

Let $\varphi \in (0, \pi/6)$ be a small enough positive number such that
\begin{equation} \label{eq:ineq_phi_tilde_Gamma^2_2}
  \frac{2 \sin(\varphi/2)^{3/2}}{\sqrt{(8 - 3\tilde{a}) \sin(3\varphi)}} < \frac{2}{3} (3\tilde{a})^{-1/2} \sin(3\delta'/2).
\end{equation}
Then the ray $\{ \rho e^{-i\varphi} \mid \rho \geq 0 \}$ intersects with $\tilde{\Gamma}^{(2)}_1 \cup \Gamma^{(2)}_3$ at a point, which we denote by $\xi''_2$, and the ray $\{ \rho e^{i(\pi + \varphi)} \mid \rho \geq 0 \}$ intersects $\tilde{\Gamma}^{(4)}_1 \cup \tilde{\Gamma}^{(4)}_3$ at a point, which we denote by $\xi''_4$. We define $\tilde{\Gamma}^{(2)}_2$ by the ray $\{ \rho e^{-i\varphi} \mid \rho \geq \lvert \xi''_2 \rvert \}$ if $\arg(\tilde{c}) \in (-\varphi/2, \pi/3 - \delta')$, and define $\tilde{\Gamma}^{(4)}_2$ by the ray $\{ \rho e^{i(\pi/2 + \varphi)} \mid \rho \geq \lvert \xi''_4 \rvert \}$ if $\arg(\tilde{c}) \in (-\pi/3 + \delta', \varphi/2)$. Then we have that 
\begin{equation} \label{eq:hat_F_on_tilde_Gamma^2_2}
  \Re \hat{F}(\xi) < \Re F(\xi_-)
  \begin{cases}
    \text{for all $\xi \in \tilde{\Gamma}^{(2)}_2$} & \text{if $\arg(\tilde{c}) \in (-\varphi/2, \pi/3 - \delta')$,} \\
    \text{for all $\xi \in \Gamma^{(4)}_2$} & \text{if $\arg(\tilde{c}) \in (-\pi/3 + \delta', \varphi/2)$}.
  \end{cases}
\end{equation}
  Below we check \eqref{eq:hat_F_on_tilde_Gamma^2_2} in the case that $\xi \in \tilde{\Gamma}^{(2)}_2$, and the case $\xi \in \tilde{\Gamma}^{(4)}_2$ is analogous. We first note that for all $\tilde{c} = e^{i\theta}$ with $\theta \in (-\varphi/2, \pi/3 - \delta')$ and for all $\rho > 0$,
  \begin{equation}
    \Re \left( i \tilde{c} \rho e^{-i\varphi} \right) = \rho \Re e^{i(\theta - \varphi + \pi/2)} \leq \rho \Re e^{i(\pi/2 - \varphi/2)} = \rho \sin(\varphi/2).
  \end{equation}
  So the value $\Re \hat{F}(\xi)$ for $\xi$ on the ray $\{ \rho e^{-i\theta} \mid \rho \geq 0 \}$ satisfies
\begin{equation}
  \begin{split}
    \Re \hat{F}(\xi) = \Re \hat{F}(\rho e^{-i\theta}) = {}& \Re \left( i \left( \tilde{a} - \frac{8}{3} \right) \rho^3 e^{-3i\theta} + i\tilde{c} \rho e^{-i\theta} \right) \\
    \leq {}& \left( \tilde{a} - \frac{8}{3} \right) \rho^3_0 \sin (3\theta) + \rho_0 \sin (\theta/2) = \frac{2 \sin(\varphi/2)^{3/2}}{\sqrt{(8 - 3\tilde{a}) \sin(3\varphi)}}.
  \end{split}
\end{equation}
On the other hand,
\begin{equation} \label{eq:value_of_F(xi_+)}
  \Re F(\xi_-) = \frac{2}{3} (3\tilde{a})^{-1/2} \Re(\tilde{c}^{3/2}) > \frac{2}{3} (3\tilde{a})^{-1/2} \sin(3\delta'/2), 
\end{equation}
since $\lvert \tilde{c} \rvert = 1$ and $\arg(\tilde{c}) \in (-\varphi/2, \pi/3 - \delta')$. So inequalities \eqref{eq:ineq_phi_tilde_Gamma^2_2} and \eqref{eq:value_of_F(xi_+)} imply \eqref{eq:hat_F_on_tilde_Gamma^2_2} in the case that $\xi \in \tilde{\Gamma}^{(2)}_2$.

\begin{rem} \label{rem:general_shape_of_Gammas}
  Although our construction depends on the value of $\arg(\tilde{c})$, by the compactness argument it is clear that for all $\tilde{c}$ that satisfy \eqref{eq:arg_tilde_c}, there exists $\epsilon > 0$ such that we can make:
  \begin{enumerate}
  \item
    For $\arg(\tilde{c} \in (-\pi/3 + \delta', \pi/3 - \delta')$, $\tilde{\Gamma}^{(0)}_1 \in \overline{\Delta_0 \cup \Delta_1}$, $\tilde{\Gamma}^{(0)}_2 \in \Delta_1$, $\tilde{\Gamma}^{(0)}_3 \in \overline{\Delta_1 \cup \Delta_2}$, and \linebreak[4] $\dist(\tilde{\Gamma}^{(0)}_1, \partial(\overline{\Delta_0 \cup \Delta_1})) > \epsilon$, $\dist(\tilde{\Gamma}^{(0)}_2, \partial(\Delta_1)) > \epsilon$, $\dist(\tilde{\Gamma}^{(0)}_3, \partial(\overline{\Delta_1 \cup \Delta_2})) > \epsilon$.
  \item
    For $\arg(\tilde{c} \in (-\varphi/2, \pi/3 - \delta')$, $\tilde{\Gamma}^{(2)}_1 \in \overline{\Delta_4 \cup \Delta_5}$, $\tilde{\Gamma}^{(2)}_2 \in \Delta_5$, $\tilde{\Gamma}^{(2)}_3 \in \overline{\Delta_5 \cup \Delta_0}$, and \linebreak[4] $\dist(\Sigma^{(2)}_1, \partial(\overline{\Omega_4 \cup \Omega_5})) > \epsilon$, $\dist(\tilde{\Gamma}^{(2)}_2, \partial(\Delta_5)) > \epsilon$, $\dist(\tilde{\Gamma}^{(2)}_3, \partial(\overline{\Delta_5 \cup \Delta_0})) > \epsilon$.
  \item
    For $\arg(\tilde{c} \in (-\pi/2 + \delta', \varphi/2)$, $\tilde{\Gamma}^{(4)}_1 \in \overline{\Delta_2 \cup \Delta_3}$, $\tilde{\Gamma}^{(4)}_2 \in \Delta_3$, $\tilde{\Gamma}^{(4)}_3 \in \overline{\Delta_3 \cup \Delta_4}$, and \linebreak[4] $\dist(\tilde{\Gamma}^{(4)}_1, \partial\overline{\Delta_2 \cup \Delta_3})) > \epsilon$, $\dist(\tilde{\Gamma}^{(4)}_2, \partial(\Delta_3) > \epsilon$, $\dist(\tilde{\Gamma}^{(4)}_3, \partial(\overline{\Delta_3 \cup \Delta_4})) > \epsilon$.
  \end{enumerate}
\end{rem}

\paragraph{Construction for the contours with general $\tilde{a}$ and $\tilde{c}$}

At last we consider the general case that $\tilde{a}$ is any positive number between $0$ and $8/3$, and $\tilde{c}$ is any number such that $\arg \tilde{c} \in (-\pi/3 + \delta', \pi/3 - \delta')$. We first construct the contours $\tilde{\Gamma}^{(0)}_{\scaled}$, $\tilde{\Gamma}^{(2)}_{\scaled}$ and $\tilde{\Gamma}^{(4)}_{\scaled}$ with respect to the parameters $1$ and $\tilde{c}/\lvert \tilde{c} \rvert$ in place of $\tilde{\alpha}$ and $\tilde{c}$, and then scale the contours in $\xi$-plane by the factor $\sqrt{\lvert \tilde{c} \rvert / \tilde{a}}$, that is, $\xi \in \tilde{\Gamma}^{(k)}$ if and only if $\xi/\sqrt{\lvert \tilde{c} \rvert/\tilde{a}} \in \tilde{\Gamma}^{(k)}_{\scaled}$. Then $\tilde{\Gamma}^{(0)}_1 \cup \tilde{\Gamma}^{(0)}_3$ is still through the point $\xi_+ = i \sqrt{\frac{\tilde{c}}{3\tilde{a}}}$, and $\tilde{\Gamma}^{(2)}_1 \cup \tilde{\Gamma}^{(2)}_3$ and $\tilde{\Gamma}^{(4)}_1 \cup \tilde{\Gamma}^{(4)}_3$ are still through the point $\xi_- = -i \sqrt{\frac{\tilde{c}}{3\tilde{a}}}$.

Our goal is to construct $\Gamma^{(0)}, \Gamma^{(2)}, \Gamma^{(4)}$, and it can be done by a translation of $\tilde{\Gamma}^{(0)}, \tilde{\Gamma}^{(2)}, \tilde{\Gamma}^{(4)}$ according to \eqref{eq:Sigma_in_Gamma}. Note that after a translation, the contours $\Gamma^{(k)}_j$ may not lie in the same sectors as $\Sigma^{(k)}_j$ do. But as $\lvert z \rvert \to \infty$, or equivalently, $\lvert \tilde{c} \rvert \to \infty$, the finite translation can be neglected. To be precise, if $\lvert z \rvert$ is large enough, then $\arg(z) \in (-\pi/3 + \delta, \pi/3 - \delta)$ implies that $\arg(\tilde{c})$ satisfies \eqref{eq:arg_tilde_c}, and $\arg(z) \in [0, \pi/3 - \delta)$ (\resp\ $\arg(z) \in (-\pi/3 + \delta, 0]$) implies that $\arg(\tilde{c}) \in (-\varphi/2, \pi/3 - \delta')$ (\resp\ $\arg(\tilde{c}) \in (-\pi/3 + \delta', \varphi/2)$). Thus we derive results \eqref{eq:condition_Gamma_0}--\eqref{eq:saddle_pt_property_of_Gamma} by properties stated in Remark \ref{rem:general_shape_of_Gammas} for the contours $\tilde{\Gamma}^{(k)}_{\scaled}$.

\subsubsection{\ref{enu:Step_3}: Saddle point analysis} \label{subsubsec:saddle_point_anal}

First we compute $n^{(0)}(z)$ as $\lvert z \rvert \to \infty$ with $\arg z \in (-\pi/3 + \delta, \pi/3 - \delta)$. As discussed in the beginning of this section, this condition is equivalent to \eqref{eq:arg_tilde_c} and $\lvert \tilde{c} \rvert \to \infty$.

We write
\begin{equation} \label{eq:formula_n^0}
  n^{(0)}(z) = \M (\tilde{n}^{(0)}(z) + \hat{n}^{(0)}(z)),
\end{equation}
where
\begin{align}
  \tilde{n}^{(0)}(z) = {}&
                           \begin{pmatrix}
                             \int_{\Gamma^{(0)}_1} e^{\frac{2iz \zeta}{C}} \Psi^{(1)}_{1, 2}(\zeta) G_1(\zeta) dz + \int_{\Gamma^{(0)}_3} e^{\frac{2iz \zeta}{C}} \Psi^{(2)}_{1, 2}(\zeta) G_1(\zeta) d\zeta \\
                             \int_{\Gamma^{(0)}_1} e^{\frac{2iz \zeta}{C}} \Psi^{(1)}_{2, 2}(\zeta) G_2(\zeta) dz + \int_{\Gamma^{(0)}_3} e^{\frac{2iz \zeta}{C}} \Psi^{(2)}_{2, 2}(\zeta) G_2(\zeta) d\zeta \\
                             \int_{\Gamma^{(0)}_1} e^{\frac{2iz \zeta}{C}} \Psi^{(1)}_{1, 2}(\zeta) G_3(\zeta) dz + \int_{\Gamma^{(0)}_3} e^{\frac{2iz \zeta}{C}} \Psi^{(2)}_{1, 2}(\zeta) G_3(\zeta) d\zeta \\
                             \int_{\Gamma^{(0)}_1} e^{\frac{2iz \zeta}{C}} \Psi^{(1)}_{2, 2}(\zeta) G_4(\zeta) dz + \int_{\Gamma^{(0)}_3} e^{\frac{2iz \zeta}{C}} \Psi^{(2)}_{2, 2}(\zeta) G_4(\zeta) d\zeta
                           \end{pmatrix}, \\
  \hat{n}^{(0)}(z) = {}&
                         \begin{pmatrix}
                           \int_{\Gamma^{(0)}_2} e^{\frac{2iz \zeta}{C}} \Psi^{(1)}_{1, 1}(\zeta) G_1(\zeta) d\zeta \\
                           \int_{\Gamma^{(0)}_2} e^{\frac{2iz \zeta}{C}} \Psi^{(1)}_{2, 1}(\zeta) G_2(\zeta) d\zeta \\
                           \int_{\Gamma^{(0)}_2} e^{\frac{2iz \zeta}{C}} \Psi^{(1)}_{1, 1}(\zeta) G_3(\zeta) d\zeta \\
                           \int_{\Gamma^{(0)}_2} e^{\frac{2iz \zeta}{C}} \Psi^{(1)}_{2, 1}(\zeta) G_4(\zeta) d\zeta
                         \end{pmatrix},
\end{align}
$\M$ is defined in \eqref{eq:defn_M4x4}, and $\Psi^{(k)}$ is the fundamental solution of \eqref{int:2a} that is expressed in $\psi^{(1)}$ and $\psi^{(2)}$ in Figure \ref{fig:2x2Phi}.

We note that $(\Psi^{(1)}_{1, 2}, \Psi^{(1)}_{2, 2})^T = \psi^{(2)}$ that is defined in \eqref{eq:defn_psi^1_psi^2}. By \eqref{eq:condition_Gamma_0} and \eqref{eq:Gamma_away_from_0}, we have that for all $\zeta \in \Gamma^{(0)}_1$, the asymptotic formula \eqref{eq:asy_psi^2} holds uniformly. Then we use the asymptotics of $\Psi^{(1)}_{1, 2}(\zeta), \Psi^{(1)}_{2, 2}(\zeta)$ to derive that uniformly
\begin{align}
  e^{\frac{2iz \zeta}{C}} \Psi^{(1)}_{1, 2}(\zeta) G_1(\zeta) = {}& \sqrt{\frac{2}{\pi}} \frac{\gamma_1}{C\ga_2 \sqrt{r_1}} e^{ - \frac{2r^2_2}{r^2_1 + r^2_2} \tau z} e^{F(\xi)} \bigO(\zeta^{-1}), \label{eq:asy_tilde_n^0_11} \\
  e^{\frac{2iz \zeta}{C}} \Psi^{(1)}_{2, 2}(\zeta) G_2(\zeta) = {}& \sqrt{\frac{2}{\pi}} \frac{1}{C \sqrt{r_2}} e^{ - \frac{2r^2_2}{r^2_1 + r^2_2} \tau z} e^{F(\xi)} (1 + \bigO(\zeta^{-1})),  \label{eq:asy_tilde_n^0_12} \\
  e^{\frac{2iz \zeta}{C}} \Psi^{(1)}_{1, 2}(\zeta) G_3(\zeta) = {}& 2i \sqrt{\frac{2}{\pi}} \frac{\gamma_1}{C^2 \ga_2 \sqrt{r_1}} \zeta e^{ - \frac{2r^2_2}{r^2_1 + r^2_2} \tau z} e^{F(\xi)} \bigO(\zeta^{-1}),\\
  e^{\frac{2iz \zeta}{C}} \Psi^{(1)}_{2, 2}(\zeta) G_4(\zeta) = {}& 2i \sqrt{\frac{2}{\pi}} \frac{1}{C^2 \sqrt{r_2}} \zeta e^{ - \frac{2r^2_2}{r^2_1 + r^2_2} \tau z} e^{F(\xi)} (1 + \bigO(\zeta^{-1})), \label{eq:asy_tilde_n^0_14}
\end{align}
where $\xi$ depends on $\zeta$ by \eqref{eq:zeta_xi_change}. Similarly, $(\Psi^{(2)}_{1, 2}, \Psi^{(2)}_{2, 2})^T = t_2 \psi^{(1)} + (t_1t_2 + 1)\psi^{(2)}$, and by \eqref{eq:condition_Gamma_0} and \eqref{eq:Gamma_away_from_0}, we also have that for all $\zeta \in \Gamma^{(0)}_3$, the asymptotic formula \eqref{eq:asy_psi^2} holds uniformly. Then similar to \eqref{eq:asy_tilde_n^0_11}--\eqref{eq:asy_tilde_n^0_14}, we have uniformly
\begin{align}
  e^{\frac{2iz \zeta}{C}} \Psi^{(2)}_{1, 2}(\zeta) G_1(\zeta) = {}& \sqrt{\frac{2}{\pi}} \frac{\gamma_1}{C\ga_2 \sqrt{r_1}} e^{ - \frac{2r^2_2}{r^2_1 + r^2_2} \tau z} e^{F(\xi)} \bigO(\zeta^{-1}), \\
  e^{\frac{2iz \zeta}{C}} \Psi^{(2)}_{2, 2}(\zeta) G_2(\zeta) = {}& \sqrt{\frac{2}{\pi}} \frac{1}{C \sqrt{r_2}} e^{ - \frac{2r^2_2}{r^2_1 + r^2_2} \tau z}e^{F(\xi)} (1 + \bigO(\zeta^{-1})), \label{eq:asy_tilde_n^0_13} \\
  e^{\frac{2iz \zeta}{C}} \Psi^{(2)}_{1, 2}(\zeta) G_3(\zeta) = {}& 2i \sqrt{\frac{2}{\pi}} \frac{\gamma_1}{C^2 \ga_2 \sqrt{r_1}} \zeta e^{ - \frac{2r^2_2}{r^2_1 + r^2_2} \tau z} e^{F(\xi)} \bigO(\zeta^{-1}),\\
  e^{\frac{2iz \zeta}{C}} \Psi^{(2)}_{2, 2}(\zeta) G_4(\zeta) = {}& 2i \sqrt{\frac{2}{\pi}} \frac{1}{C^2 \sqrt{r_2}} \zeta e^{ - \frac{2r^2_2}{r^2_1 + r^2_2} \tau z} e^{F(\xi)} (1 + \bigO(\zeta^{-1})).
\end{align}
We compute the second component of the $4$-dimensional vector $\tilde{n}^{(0)}(z)$ in detail. The uniform convergence asymptotics \eqref{eq:asy_tilde_n^0_12} and \eqref{eq:asy_tilde_n^0_13} imply that 
\begin{equation}
  \begin{split}
    \tilde{n}^{(0)}_2(z) = {}& \sqrt{\frac{2}{\pi}} \frac{1}{C \sqrt{r_2}} e^{ - \frac{2r^2_2}{r^2_1 + r^2_2} \tau z} \int_{\Gamma^{(0)}_1 \cup \Gamma^{(0)}_3} e^{F(\xi)} (1 + \bigO(\zeta^{-1})) d\zeta \\
    = {}& \sqrt{\frac{2}{\pi}} \frac{1}{C \sqrt{r_2}} e^{ - \frac{2r^2_2}{r^2_1 + r^2_2} \tau z} \int_{\tilde{\Gamma}^{(0)}_1 \cup \tilde{\Gamma}^{(0)}_3} e^{F(\xi)} (1 + \bigO(\xi^{-1})) d\xi.
  \end{split}
\end{equation}
According to the construction in Section \ref{subsubsec:construction_of_contours}, the contour $\tilde{\Gamma}^{(0)}_1 \cup \tilde{\Gamma}^{(0)}_3$ has the following property that $\Re F(\xi)$ attains its unique maximum on it at $\xi_+$, $\Re F(\xi)$ decreases fast as $\xi \to \infty$ along it, and locally around $\xi_+$ it is the steepest descent contour for $\Re F(\xi)$. Thus a standard application of the saddle point method yields
\begin{equation}
  \begin{split}
    \int_{\tilde{\Gamma}^{(0)}_1 \cup \tilde{\Gamma}^{(0)}_3} e^{F(\xi)} (1 + \bigO(\xi^{-1})) d\xi = {}& \sqrt{\frac{2\pi}{-F''(\xi_+)}} e^{F(\xi_+)} (1 + \bigO(\xi^{-1}_+)) \\
    = {}& \frac{\sqrt{\pi}}{2} \frac{(r^2_1 + r^2_2)^{1/3}}{r^{2/3}_1 r^{1/6}_2} z^{-1/4} e^{- \theta_2(z)} (1 + \bigO(z^{-1/2})).
  \end{split}
\end{equation}
Hence 
\begin{equation} \label{eq:approx_n^0_2}
  \tilde{n}^{(0)}_2(z) = \sqrt{\frac{2}{\pi}} \frac{1}{C \sqrt{r_2}} \frac{\sqrt{\pi}}{2} \frac{(r^2_1 + r^2_2)^{1/3}}{r^{2/3}_1 r^{1/6}_2} e^{ - \frac{2r^2_2}{r^2_1 + r^2_2} \tau z} e^{F(\xi_+)} = \frac{1}{\sqrt{2}} e^{-\frac{2r^2_2}{r^2_1 + r^2_2} \tau z} z^{-1/4} e^{- \theta_2(z)} (1 + \bigO(z^{-1/2})).
\end{equation}
Similarly, we have for the fourth component of $\tilde{n}^{(0)}(z)$
\begin{equation} \label{eq:approx_n^0_4}
  \begin{split}
    \tilde{n}^{(0)}_4(z) = {}& 2i \sqrt{\frac{2}{\pi}} \frac{1}{C^2 \sqrt{r_2}} e^{ - \frac{2r^2_2}{r^2_1 + r^2_2} \tau z} \int_{\Gamma^{(0)}_1 \cup \Gamma^{(0)}_3} \zeta e^{F(\xi)} (1 + \bigO(\zeta^{-1})) d\zeta \\
    = {}& 2i \sqrt{\frac{2}{\pi}} \frac{1}{C^2 \sqrt{r_2}} e^{ - \frac{2r^2_2}{r^2_1 + r^2_2} \tau z} \int_{\tilde{\Gamma}^{(0)}_1 \cup \tilde{\Gamma}^{(0)}_3} \left( \xi + \frac{ib}{3a + 4} \right) e^{F(\xi)} (1 + \bigO(\xi^{-1})) d\xi \\
    = {}& \frac{i r_2}{\sqrt{2}} e^{-\frac{2r^2_2}{r^2_1 + r^2_2} \tau z} z^{1/4} e^{- \theta_2(z)} (1 + \bigO(z^{-1/2})).
  \end{split}
\end{equation}
For the first and third components of $\tilde{n}^{(0)}(z)$, we can do the same computation, but we only need the estimates as follows
\begin{equation} \label{eq:approx_n^0_13}
  \tilde{n}^{(0)}_2(z) = e^{-\frac{2r^2_2}{r^2_1 + r^2_2} \tau z} z^{-1/4} e^{- \theta_2(z)} \bigO(z^{-1/2}), \quad \tilde{n}^{(0)}_2(z) = e^{-\frac{2r^2_2}{r^2_1 + r^2_2} \tau z} z^{1/4} e^{- \theta_2(z)} \bigO(z^{-1/2}).
\end{equation}

Next we consider the components of $\hat{n}^{(0)}(z)$, and give some detail in the estimate of the first component. We note that $(\Psi^{(1)}_{1, 1}, \Psi^{(1)}_{2, 1})^T = \psi^{(1)} + t_1 \psi^{(2)}$. By \eqref{eq:condition_Gamma_middles} and \eqref{eq:Gamma_away_from_0}, we have that for all $\zeta \in \Gamma^{(0)}_2$, the asymptotic formula \eqref{eq:asy_psi^1+t_1psi^2} holds uniformly. Then we use the asymptotics of $\Psi^{(1)}_{1, 1}(\zeta)$ to derive that uniformly
\begin{equation} \label{eq:approx_on_Gamma^0_2}
  e^{\frac{2iz \zeta}{C}} \Psi^{(1)}_{1, 2}(\zeta) G_1(\zeta) = \sqrt{\frac{2}{\pi}} \frac{\gamma_1}{C\ga_2 \sqrt{r_1}} e^{ - \frac{2r^2_2}{r^2_1 + r^2_2} \tau z} e^{\hat{F}(\xi) + f(\xi)} \bigO(\zeta^{-1}),
\end{equation}
where $\hat{F}(\xi)$ is defined in \eqref{eq:defn_hat_F} and
\begin{equation}
  f(\xi) = \frac{8b}{3a + 4}\xi^2 + i \left( \frac{8b^2}{(3a + 4)^2} - 2\sigma \right) \xi - \frac{8b^3}{3(3a + 4)^3} + \frac{2\sigma b}{3a + 4}.
\end{equation}
Note that the coefficients of $F(\xi)$ and $\hat{F}(\xi)$ are given in terms of $\tilde{a}$ and $\tilde{c}$. If we denote
\begin{equation}
  \F(\xi) = F(\xi) \Big\rvert_{\tilde{c} \to \tilde{c}/\lvert \tilde{c} \rvert} = i\tilde{a} \xi^3 + i\frac{\tilde{c}}{\lvert \tilde{c} \rvert} \xi, \quad \hat{\F}(\xi) = \hat{F}(\xi) \Big\rvert_{\tilde{c} \to \tilde{c}/\lvert \tilde{c} \rvert} = i \left( \tilde{a} - \frac{8}{3} \right) \xi^3 + i\frac{\tilde{c}}{\lvert \tilde{c} \rvert} \xi, 
\end{equation}
then
\begin{equation}
  F(\xi) = \lvert \tilde{c} \rvert^{3/2} \F \left( \frac{\xi}{\sqrt{\lvert \tilde{c} \rvert}} \right), \quad \hat{F}(\xi) = \lvert \tilde{c} \rvert^{3/2} \hat{\F} \left( \frac{\xi}{\sqrt{\lvert \tilde{c} \rvert}} \right).
\end{equation}
By the construction of $\Gamma^{(0)}_2$ and \eqref{eq:property_hat_F_Tamma^0_2}, we have that for all $\zeta \in \Gamma^{(0)}_2$, or equivalently $\xi \in \tilde{\Gamma}^{(0)}_2$, there exists $\epsilon > 0$ such that
\begin{equation}
  \Re \hat{\F} \left( \frac{\xi}{\sqrt{\lvert \tilde{c} \rvert}} \right) < \Re \F \left( \frac{\xi_+}{\sqrt{\lvert \tilde{c} \rvert}} \right) - \epsilon.
\end{equation}
As $z \to \infty$, we have $\tilde{c} = 2C^{-1} z + \bigO(1)$, and then for all $\zeta \in \Gamma^{(0)}_2$, or equivalently $\xi \in \tilde{\Gamma}^{(0)}_2$, 
\begin{equation}
  \Re \hat{F}(\xi) < \Re F(\xi_+) - \left( \frac{2\epsilon}{C} \right)^{3/2} \lvert z \rvert^{3/2}.
\end{equation}
Since $f(\xi)$ is independent of $z$ and $f(\xi) = \bigO(\xi^2)$ as $\xi \to \infty$, we have that as $z \to \infty$, if $\lvert \xi \rvert \leq \lvert z \rvert^{3/5}$, then $\lvert f(\xi) \rvert = \bigO(z^{6/5})$. Thus for $\xi \in \Sigma^{(0)}_2$ and $\lvert \xi \rvert \leq \lvert z \rvert^{3/5}$, there exists $\epsilon' > 0$ such that for large enough $z$
\begin{equation} \label{eq:ineq_hat_F}
  \Re \hat{F}(\xi) + f(\xi) < \Re F(\xi_+) - \epsilon' \lvert z \rvert^{3/2}.
\end{equation}
On the other hand, if $\xi \in \Sigma^{(0)}_2$ and $\lvert \xi \rvert > \lvert z \rvert^{3/5}$ and $z \to \infty$, then $\hat{F}(\xi)$ is dominated by the cubic term, and it is clear that inequality \eqref{eq:ineq_hat_F} still holds.

By the approximation \eqref{eq:approx_on_Gamma^0_2}, using \eqref{eq:ineq_hat_F} and that $\Re \hat{F}(\xi) \to -\infty$ fast as $\xi \to \infty$ along $\tilde{\Gamma}^{(0)}_2$, we estimate that
\begin{equation}
  \hat{n}^{(0)}_1(z) = \sqrt{\frac{2}{\pi}} \frac{\gamma_1}{C\ga_2 \sqrt{r_1}} e^{ - \frac{2r^2_2}{r^2_1 + r^2_2} \tau z} \int_{\tilde{\Gamma}^{(0)}_2} e^{\hat{F}(\xi) + f(\xi)} \bigO(\xi^{-1}) d\xi = e^{-\frac{2r^2_2}{r^2_1 + r^2_2} \tau z} e^{- \theta_2(z)} o(e^{-\epsilon' \lvert z \rvert^{3/2}}),
\end{equation}
where $\epsilon' > 0$ is a constant, which can be taken to be the same as in \eqref{eq:ineq_hat_F}. By the same method, we obtain the general result
\begin{equation} \label{approx_n^0_hat}
  \hat{n}^{(0)}_k(z) = e^{-\frac{2r^2_2}{r^2_1 + r^2_2} \tau z} e^{- \theta_2(z)} o(e^{-\epsilon' \lvert z \rvert^{3/2}}), \quad k = 1, 2, 3, 4.
\end{equation}

Plugging \eqref{eq:approx_n^0_2}, \eqref{eq:approx_n^0_4}, \eqref{eq:approx_n^0_13} and \eqref{approx_n^0_hat} into \eqref{eq:formula_n^0}, we derive that
\begin{equation}
  n^{(0)}(z) = \frac{1}{\sqrt{2}}
  \begin{pmatrix}
    e^{-\theta_2(z) - \tau z} \bigO(z^{-3/4}) \\
    z^{-1/4} e^{-\theta_2(z) - \tau z} (1 + \bigO(z^{-1/2})) \\
    e^{-\theta_2(z) - \tau z} \bigO(z^{-1/4}) \\
    z^{1/4} e^{-\theta_2(z) - \tau z} (1 + \bigO(z^{-1/2}))
  \end{pmatrix},
\end{equation}
and prove part \ref{enu:thm:main_asy_1} of Proposition \ref{thm:main_asy}.

\subsection{Sketch of the proof of parts \ref{enu:thm:main_asy_2} -- \ref{enu:thm:main_asy_6}} \label{subsec:other_parts_1.3}

\subsubsection{Proof of parts \ref{enu:thm:main_asy_2} and \ref{enu:thm:main_asy_3}}

The proof of parts \ref{enu:thm:main_asy_2} and \ref{enu:thm:main_asy_3} is parallel to that of part \ref{enu:thm:main_asy_1}. We also take the change of variables \eqref{eq:zeta_xi_change} and compute the critical point $\xi_{\pm}$ as in \eqref{eq:xi_pm}. But now $\xi_+ \in \Delta_2$ and $\xi_- \in \Delta_5$ in the setting of part \ref{enu:thm:main_asy_2}, and $\xi_+ \in \Delta_0$ and $\xi_- \in \Delta_3$ in the setting of part \ref{enu:thm:main_asy_3}. Also we use the method from planar dynamic systems to construct $L_0$, $L_{\pm}$, and the flow curves, and then $\tilde{\Gamma}^{(k)}_{\scaled}$, and finally $\tilde{\Gamma}^{(k)}$ and $\Gamma^{(k)}$ ($k = 0, 2, 4$). The results are shown in Figures \ref{fig:contours_part_2} and \ref{fig:contours_part_3}. An obvious $2\pi/3$ rotational symmetry can be ovserved in Figures \ref{fig:contours_part_1}, \ref{fig:contours_part_2} and \ref{fig:contours_part_3}, and it is a direct consequence of the symmetry among the settings in the three parts. At last, the saddle point analysis is applied, and the critical point $\zeta_+$ yields the result $e^{-\theta_2(\zeta) - \tau z} \bigO(z^{-1/4})$, and the critical point $\zeta_-$ yields the result $e^{\theta_2(\zeta) - \tau z} \bigO(z^{-1/4})$. The explicit leading terms of the $\bigO(z^{-1/4})$ factors are computed in the way of Section \ref{subsubsec:saddle_point_anal}.

\subsubsection{Proof of parts \ref{enu:thm:main_asy_4}, \ref{enu:thm:main_asy_5} and \ref{enu:thm:main_asy_6}}

Similar to the proof to parts \ref{enu:thm:main_asy_1}, \ref{enu:thm:main_asy_2} and \ref{enu:thm:main_asy_3}, the essential part of the asymptotic analysis in the proof of parts \ref{enu:thm:main_asy_4}, \ref{enu:thm:main_asy_5} and \ref{enu:thm:main_asy_6} is the integrals on $\Gamma^{(k)}_1 \cup \Gamma^{(k)}$ ($k = 1, 3, 5$), where the contours $\Gamma^{(k)}_j$ are deformed, in a similar way to the deformation of contours shown in Section \ref{subsec:part_1_of_1.3} for part \ref{enu:thm:main_asy_1}. The integrands on $\Gamma^{(k)}_1 \cup \Gamma^{(k)}$, although various in explicit formulas, all have the asymptotic behavior
\begin{equation}
  e^{-\frac{4}{3} i\zeta^3 - i \sigma \zeta + \frac{2iz\zeta}{C}} G(\zeta) \times (\text{factor growing at most linearly at $\infty$}),
\end{equation}
which is comparable to \eqref{eq:asy_integrands_parts_123} in Remark \ref{rem:meaning_of_F}.

Thus we take the change of variables, comparable to \eqref{eq:zeta_xi_change}
\begin{equation} \label{eq:relation_zeta_xi_parts456}
  \zeta = \xi - \frac{ib}{-3a + 4} = \xi - \frac{i\tau}{C^2 r^2_2},
\end{equation}
define
\begin{equation}
  F(\xi) = -i\tilde{a} \xi^3 + i\tilde{c} \xi, \quad \text{where} \quad \tilde{a} = \frac{4}{3} - a = \frac{8r^2_2}{3(r^2_1 + r^2_2)}, \quad \tilde{c} = -\frac{b^2}{-3a + 4} + c + \frac{2z}{C} - \sigma = \frac{2z - 4s_1/r_1}{C},
\end{equation}
and have
\begin{equation}
  \log \left( e^{-\frac{4}{3} i\zeta^3 - i \sigma \zeta + \frac{2iz\zeta}{C}} G(\zeta) \right) = F(\xi) +\log \gamma_1 - \frac{2r^2_1}{r^2_1 + r^2_2} \tau z.
\end{equation}
\begin{rem}
  Here and below notations like $\xi$, $F$, $\tilde{a}$, and $\tilde{c}$, are different from their counterparts in Section \ref{asymptotics} but serve the same purpose in the proof. We use the same notations to emphasize the identical use, while we trust that they do not lead to confusion. 
\end{rem}
Then we find the critical points of $F(\xi)$, and denote them 
\begin{equation}
  \xi_{\pm} = \pm \sqrt{\frac{\tilde{c}}{\tilde{a}}}, 
\end{equation}
and then let
\begin{equation}
  \zeta_{\pm} = \xi_{\pm} - \frac{ib}{-3a + 4}.
\end{equation}
We deform $\Gamma^{(k)}$ ($k = 1, 3, 5$) such that $\Gamma^{(k)}_1 \cup \Gamma^{(k)}_3$ are  through either $\zeta_+$ or $\zeta_-$, satisfy
  \begin{equation} 
    \Re \log \left( e^{-\frac{4}{3} i\zeta^3 - i \sigma \zeta + \frac{2iz\zeta}{C}} G(\zeta) \right) \text{ attains its maximum on $\Gamma^{(k)}_1 \cup \Gamma^{(k)}_3$ at $\zeta_{\pm}$, \quad} k = 1, 3, 5,
  \end{equation}
and the deformed contours satisfy conditions analogous to \eqref{eq:condition_Gamma_0}--\eqref{eq:Gamma_away_from_0}. Since the construction of the contours is different from the constructions in parts \ref{enu:thm:main_asy_1}, \ref{enu:thm:main_asy_2}, and \ref{enu:thm:main_asy_3} only in computational detail, we omit it, and only show Figures \ref{fig:contours_part_4}, \ref{fig:contours_part_5} and \ref{fig:contours_part_6} to indicate the result of the construction.
  \begin{figure}[htb]
    \begin{minipage}[t]{0.3\linewidth}
      \centering
      \includegraphics{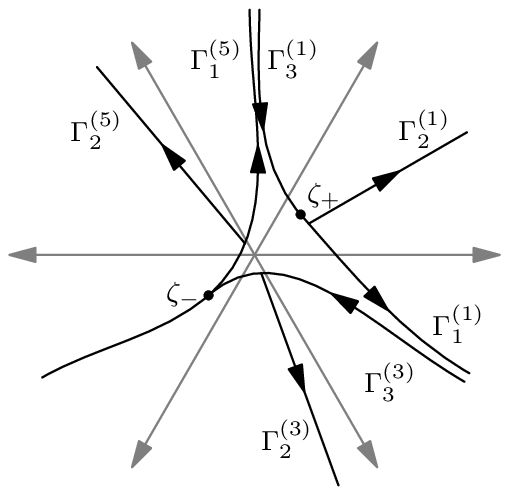}
      \caption{Schematic graphs of $\Gamma^{(1)}$, $\Gamma^{(3)}$ and $\Gamma^{(5)}$, in the proof of part \ref{enu:thm:main_asy_4} of Proposition \ref{thm:main_asy}. $\Gamma^{(3)}_1$ and $\Gamma^{(5)}_3$ are not labelled, because their major parts overlap.}
      \label{fig:contours_part_4}
    \end{minipage}
    \hspace{\stretch{1}}
    \begin{minipage}[t]{0.3\linewidth}
      \centering
      \includegraphics{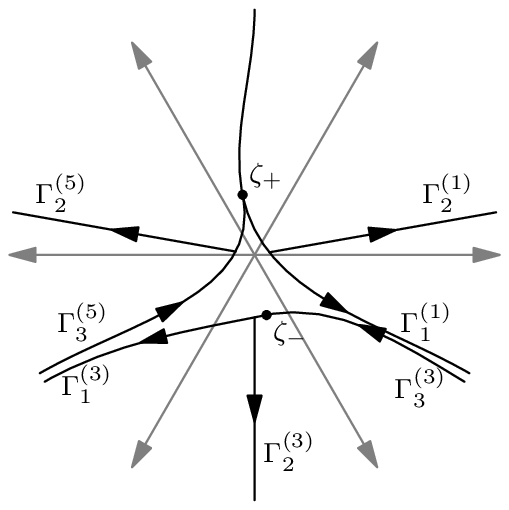}
      \caption{Schematic graphs of $\Gamma^{(1)}$, $\Gamma^{(3)}$ and $\Gamma^{(5)}$, in the proof of part \ref{enu:thm:main_asy_5} of Proposition \ref{thm:main_asy}. $\Gamma^{(5)}_1$ and $\Gamma^{(1)}_3$ are not labelled, because their major parts overlap.}
      \label{fig:contours_part_5}
    \end{minipage}
    \hspace{\stretch{1}}
    \begin{minipage}[t]{0.3\linewidth}
      \centering
      \includegraphics{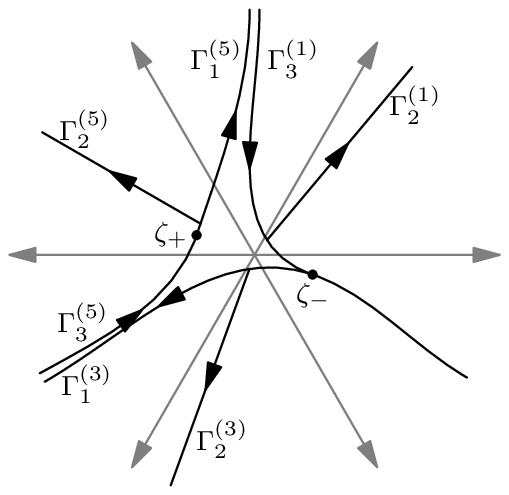}
      \caption{Schematic graphs of $\Gamma^{(1)}$, $\Gamma^{(3)}$ and $\Gamma^{(5)}$, in the proof of part \ref{enu:thm:main_asy_6} of Proposition \ref{thm:main_asy}. $\Gamma^{(1)}_1$ and $\Gamma^{(3)}_3$ are not labelled, because their major parts overlap.}
      \label{fig:contours_part_6}
    \end{minipage}
  \end{figure}

  At last we apply the saddle point analysis, and find that the critical point $\zeta_+$ yields the result $e^{\theta_1(\zeta) - \tau z} \bigO(z^{-1/4})$, and the critical point $\zeta_-$ yields the result $e^{-\theta_1(\zeta) - \tau z} \bigO(z^{-1/4})$, and prove the results. The detailed computation is omitted.

\section{Proof of Theorem \ref{fund_solutions}} \label{sec:proof_main}

The proof of Theorem \ref{fund_solutions} follows from combining the results of Propositions \ref{prop:Duits-Geudens}, \ref{diffeq}, and \ref{thm:main_asy}. We will write the detailed proof of the formula for $M^{(0)}$. The proofs for $M^{(1)}, M^{(2)},\dots,M^{(5)}$ are nearly identical, and we leave them to the reader. Throughout the proof, we refer the reader to Figure \ref{fig:A+Dominance}, which divides the complex plane into 12 sectors, each of size $\pi/6$. Within each of these sectors, the asymptotic dominance scheme of the columns of the matrix $\acal^+(z)$ is indicated. For example, in the sector $0<\arg z < \pi/6$ the sequence $4, 1, 3, 2$ means that as $z\to\infty$, 
 \begin{equation}
 v^+_1(z) = o\left( v^+_4(z)\right), \quad v^+_3(z) = o\left( v^+_1(z)\right), \quad v^+_2(z) = o\left( v^+_3(z)\right),
 \end{equation}
 where we recall that $v^+_j$ are columns of $\acal^+$.
It is easy to check this dominance scheme in each of the sectors from the definitions \eqref{eq:defn_v_1-v_4} and the relations \eqref{eq:v_jumps}.
The rays which separate the different dominance schemes are called the {\it Stokes rays}.

  \begin{figure}[ht]
    \begin{minipage}[t]{0.45\linewidth}
      \centering
      \includegraphics{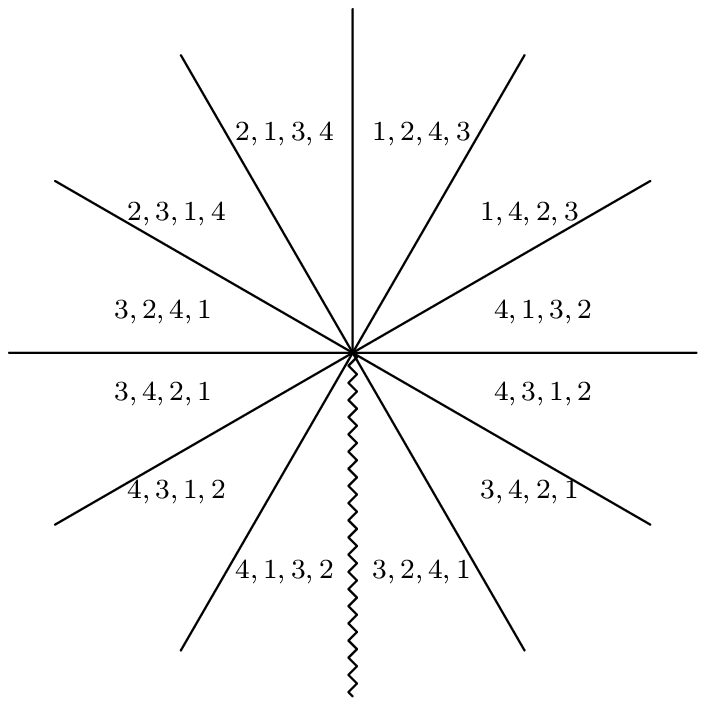}
      \caption{The dominance scheme for $\acal^+(z)$. The complex plane is separated into sectors of angle $\pi/6$. In each sector the relative dominance as $z\to\infty$ of the columns of $\acal^+(z)$ is indicated, e.g. the sequence 4, 1, 3, 2 in the sector $0<\arg z < \pi/6$ means that $\lvert v^+_4(z) \rvert \gg \lvert v^+_1(z)\rvert \gg \lvert v^+_3(z)\rvert \gg \lvert v^+_2(z)\rvert$ as $z\to \infty$. The zigzag line indicates the branch cut for $\acal^+(z)$. cf.~\cite[Figure 13]{Duits-Geudens13}.}
      \label{fig:A+Dominance}
    \end{minipage}
    \hspace{\stretch{1}}
    \begin{minipage}[t]{0.45\linewidth}
      \centering
      \includegraphics{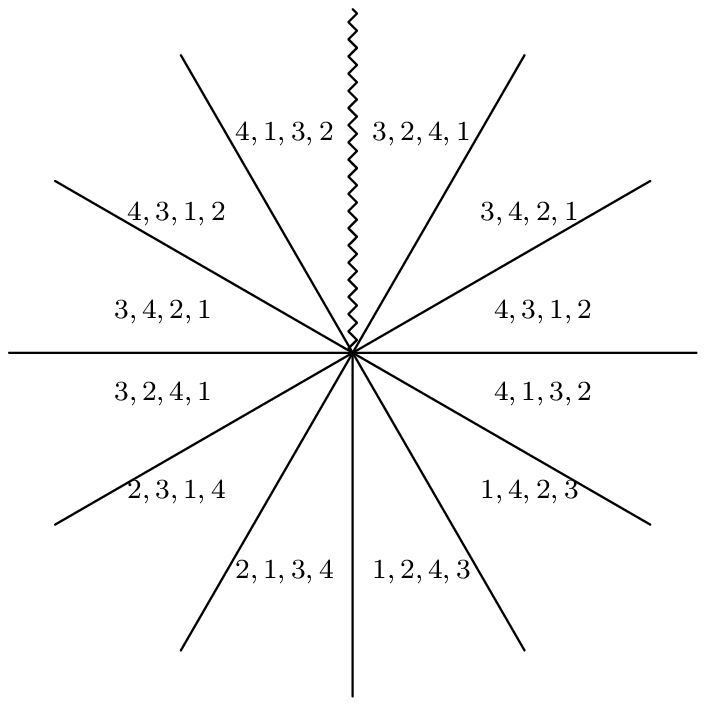}
      \caption{The dominance scheme for $\acal^-(z)$. The complex plane is separated into sectors of angle $\pi/6$. In each sector the relative dominance as $z\to\infty$ of the columns of $\acal^-(z)$ is indicated, e.g. the sequence 4, 3, 1, 2 in the sector $0<\arg z < \pi/6$ means that $\lvert v^-_4(z) \rvert \gg \lvert v^-_3(z)\rvert \gg \lvert v^-_1(z)\rvert \gg \lvert v^-_2(z)\rvert$ as $z\to \infty$. The zigzag line indicates the branch cut for $\acal^-(z)$.}
      \label{fig:A-Dominance}
    \end{minipage}
  \end{figure}

Denote the columns of $M^{(0)}$ by $m^{(0)}_1, m^{(0)}_2, m^{(0)}_3, m^{(0)}_4$, so that
 \begin{equation}
 M^{(0)}(z)= \left( m^{(0)}_1(z), m^{(0)}_2(z), m^{(0)}_3(z),  m^{(0)}_4(z) \right).
 \end{equation}
 According to Proposition \ref{prop:Duits-Geudens}, $ M^{(0)}(z)=(I+\bigO(z^{-1}))\acal^+(z)$ as $z\to\infty$ throughout the sector $\Om_0$, so the dominance scheme in Figure \ref{fig:A+Dominance} applies also to the columns of $M^{(0)}(z)$ in this sector. Notice that $\Om_0=\{ z\in \C : -\pi/12 < \arg z < 7\pi/12\}$ overlaps with 5 of the sectors shown in Figure \ref{fig:A+Dominance}: $-\pi/6<\arg z < 0$;  $0<\arg z <\pi/6$;  $\pi/6<\arg z < \pi/3$; $\pi/3< \arg z < \pi/2;$ and $\pi/2 < \arg z < 4\pi/3.$
 
We also refer to Figure \ref{fig:Prop1_3_summary}, which summarizes the results of Proposition \ref{thm:main_asy}. In that figure the asymptotic behavior as $z\to\infty$ of the solutions $n^{(j)}(z)$ is matched to that of the functions $v^{\pm}_1, \dots, v^{\pm}_4$. 
  
In the sector $0<\arg z < \pi/6$, we have that $n^{(0)}(z) \sim v^+_2(z)$, and $m^{(0)}_2(z) \sim v^+_2(z)$. According to the dominance scheme in this sector, this indicates that both $n^{(0)}(z)$ and $m^{(0)}_2(z)$ are {\it recessive solutions} to \eqref{int:7_a}  in this sector, i.e. they are not dominant over any other solutions to \eqref{int:7_a} as $z\to \infty$ in this sector except the trivial one. Recessive solutions are unique up to a constant factor. Since both $n^{(0)}(z)$ and $m^{(0)}_2(z)$ match the leading order behavior of $v^+_2(z)$ in this sector, we must have
 \begin{equation}\label{eq:m02}
 m^{(0)}_2(z)=n^{(0)}(z).
 \end{equation}
 
 By considering the recessive solutions in the sectors $\pi/3< \arg z < \pi/2$ and $\pi/2 < \arg z < 4\pi/3,$ we similarly obtain
  \begin{equation}\label{eq:m0304}
 m^{(0)}_3(z)=n^{(1)}(z), \qquad  m^{(0)}_4(z)=-n^{(2)}(z).
 \end{equation}
 It remains only to find $m^{(0)}_1(z)$, which is the only column of $M^{(0)}(z)$ which is not recessive in one of the sectors which overlap $\Om_0$. We look instead for a sector in which it is the least dominant possible amongst those overlapping $\Om_0$. Notice in the sector $-\pi/12 < \arg z < 0$, $m^{(0)}_1(z)$ dominates $m^{(0)}_2(z)$, but is dominated by the other columns of $M^{(0)}(z)$. According to the Figure \ref{fig:Prop1_3_summary}, in this sector we have $n^{(5)}(z) \sim v^+_1(z)$. 
 
 Since $n^{(5)}(z)$ solves \eqref{int:7_a} it is a linear combination of the rows of $M^{(0)}(z)$,
   \begin{equation}
   n^{(5)}(z)= c_1 m^{(0)}_1(z)+c_2 m^{(0)}_2(z)+c_3 m^{(0)}_3(z)+c_4 m^{(0)}_4(z),
 \end{equation}
 for some constants $c_1, c_2, c_3, c_4$.
 In the sector $-\pi/12 < \arg z < 0$, $n^{(5)}(z) \sim v^+_1(z)$ and $m^{(0)}_3(z) \sim v^+_3(z)$, which dominates $v^+_1(z)$ there, so $c_3 = 0$. Similarly, in the sector $0 < \arg z < \pi/6$, $n^{(5)}(z) \sim v^+_1(z)$, and $m^{(0)}_4(z) \sim v^+_4(z)$, which dominates $v^+_1(z)$ there, so $c_4 = 0$. Furthermore, in the sector $-\pi/12 < \arg z < 0$, we have $c_1 m^{(0)}_1(z) + c_2 m^{(0)}_2(z) \sim c_1 v^+_1(z) + c_2 v^+_2(z) \sim c_1 v^+_1(z)$, so $c_1 = 1$ by the comparison with the asymptotics of $n^{(5)}(z)$. Using \eqref{eq:m02} as well, we obtain
 \begin{equation}\label{eq:n5}
 n^{(5)}(z)= m^{(0)}_1(z)+c_2 n^{(0)}(z).
 \end{equation}
To find the value of $c_2$, we can use the linear relation \eqref{nj_relations_a},
\begin{equation}\label{nj_relations_a2}
 n^{(5)}(z) = -t_3 n^{(0)}(z)-(1+t_2t_3) n^{(1)}(z) + t_2n^{(2)}(z) -n^{(3)}(z),
\end{equation}
and consider the asymptotics of $n^{(5)}(z)$ in the sector $\pi/2 < \arg z < 7\pi/12.$ The leading order behavior of each of the functions on the right-hand side of \eqref{nj_relations_a2} is given in Proposition \ref{thm:main_asy} (see also Figure \ref{fig:Prop1_3_summary}). Inserting these asymptotics into 
\eqref{nj_relations_a2} gives, as $z\to \infty$ with $\pi/2 < \arg z < 7\pi/12,$
\begin{equation}\label{eq:n5asy1}
 n^{(5)}(z) \sim -t_3 v^+_2(z) - (1 + t_2 t_3) v^+_3(z) - t_2 v^+_4(z) + v^+_1(z).
\end{equation}
According to Figure \ref{fig:A+Dominance}, $v^+_2(z)$ is dominant as $z\to \infty$ in this sector, so \eqref{eq:n5asy1} becomes
\begin{equation}\label{eq:n5asy2}
 n^{(5)}(z) \sim -t_3 v^+_2(z),
\end{equation}
or equivalently, in the sector $\pi/2 < \arg z < 7\pi/12$,
\begin{equation}\label{eq:n5asy3}
 n^{(5)}(z) \sim -t_3 n^{(0)}(z).
\end{equation}
Comparing \eqref{eq:n5asy2} and \eqref{eq:n5}, and noting that $n^{(0)}(z) \sim v^+_2(z)$ dominates $m^{(0)}_1(z) \sim v^+_1(z)$ in the sector $\pi/2 < \arg z < 7\pi/12$, we find that $c_2=-t_3$. Thus \eqref{eq:n5} gives the formula for $m^{(0)}_1(z)$:
 \begin{equation}\label{eq:m01}
 m^{(0)}_1(z)= n^{(5)}(z)+t_3 n^{(0)}(z).
 \end{equation}
Combining \eqref{eq:m02}, \eqref{eq:m0304}, and \eqref{eq:m01} gives the formula for $M^{(0)}(z)$ in Theorem \ref{fund_solutions}.

The formulas for the rest of the solutions $M^{(1)}(z), \dots, M^{(5)}(z)$ can be obtained in a similar manner. Always three out of the four columns of $M^{(j)}$ can be identified as solutions to \eqref{int:7_a} which are recessive in some part of $\Om_j$. These recessive solutions can be identified with one of the functions $n^{(k)}(z)$ using Proposition \ref{thm:main_asy}, or equivalently referencing Figure \ref{fig:Prop1_3_summary}. There is one column which is never recessive in $\Om_j$, but it can be determined using the linear relations \eqref{nj_relations} in a manner similar to how $m^{(0)}$ was determined above. In Figure \ref{fig:A-Dominance} we include the dominance scheme for the columns of $\acal^-$, which should be consulted when considering $M^{(3)}(z)$, $M^{(4)}(z),$ and $M^{(5)}(z)$.

\section{Proof of contour integral formulas for kernels} \label{sec:proof_of_contour_int_form}

In this section, we assume $q(\sigma)$ is the Hastings--McLeod solution to the PII equation \eqref{int:1}. Then the solutions $\Psi^{(0)}(\zeta; \sigma), \dotsc, \Psi^{(0)}(\zeta; \sigma)$ to the Lax pair \eqref{int:2} that are defined in Section \ref{subsec:Flaschka-Newell_Lax} are also assumed to be associated with the Hastings--McLeod solution $q(\sigma)$. These solutions to \eqref{int:2} are related by the jump conditions \eqref{eq:jump_condition_2x2} which are in turn determined by the parameters $(t_1, t_2, t_3)$. In this section we assume $(t_1, t_2, t_3) = (1, 0, -1)$, associated with the Hastings--McLeod solution. Recall that for $j=0,\dots, 5$, $n^{(j)}(z) = n^{(j)}(z;r_1, r_2, s_1, s_2, \tau)$ defined in \eqref{eq:def_of_nj1} are vector-valued functions with the parameters $r_1, r_2, s_1, s_2, \tau$. The vectors $n^{(j)}$ also depend on a solution to the PII equation \eqref{int:1} by definition, and we assume it to be the Hastings--McLeod solution $q(\sg)$ in this section.

By the symmetry of equations \eqref{int:2}, and the identity $\Psi^{(0)}(\zeta; \sigma) = \Psi^{(3)}(\zeta; \sigma)$ that holds because $t_2 = 0$ for the Hastings--McLeod solution (see Figure \ref{fig:2x2Phi}), we have that
\begin{equation}
  \Psi^{(0)}(\zeta; \sigma) =
  \begin{pmatrix}
    0 & 1 \\
    1 & 0
  \end{pmatrix}
  \Psi^{(0)}(\zeta; \sigma)
  \begin{pmatrix}
    0 & 1 \\
    1 & 0
  \end{pmatrix}.
\end{equation}
It implies that, with functions $f(\zeta; \sigma), g(\zeta; \sigma), \Phi_1(\zeta; \sigma), \Phi_2(\zeta; \sigma)$ defined in \eqref{tac:8} and \eqref{eq:formula_Phi_1_2},
\begin{align}
  \Phi_1(\zeta; \sigma) = {}& \Phi_2(-\zeta; \sigma) & & \text{for all $\zeta \in \compC$}, \label{eq:u_to_-u_Phi} \\
  f(\zeta; \sigma) = {}& -g(-\zeta; \sigma) & & \text{for all $\zeta \in \compC \setminus \realR$}.
\end{align}
At last, we note that the differential equation \eqref{int:2b} implies that
\begin{equation} \label{tac:14}
  \begin{aligned}
  \frac{\partial f(\zeta; \sigma)}{\partial \sigma} = {}& -i\zeta f(\zeta; \sigma) + q(\sigma) g(\zeta; \sigma), & \frac{\partial g(\zeta; \sigma)}{\partial \sigma} = {}& q(\sigma) f(\zeta; \sigma) + i\zeta g(\zeta; \sigma), \\
  \frac{\partial \Phi_1(\zeta; \sigma)}{\partial \sigma} = {}& -i\zeta \Phi_1(\zeta; \sigma) + q(\sigma) \Phi_2(\zeta; \sigma), & \frac{\partial \Phi_2(\zeta; \sigma)}{\partial \sigma} = {}& q(\sigma) \Phi_1(\zeta; \sigma) + i\zeta \Phi_2(\zeta; \sigma).
  \end{aligned}
\end{equation}

\subsection{Proof of Theorem \ref{critical}} \label{two_matrix_kernel_proof}

The Duits--Guedens critical kernel for the two-matrix model was derived in \cite{Duits-Geudens13}, and our proof of Theorem \ref{critical} is based on the presentation in \cite{Kuijlaars14}. The critical kernel is described in terms of the tacnode RHP with parameters \cite[Formula (2.41)]{Kuijlaars14}
\begin{equation}\label{crit:1}
r_1=r_2=1, \quad s_1=s_2=s, \quad \tau \in \R.
\end{equation}
In \cite[Formulas (4.5) and (4.9)]{Kuijlaars14}, two vector-valued functions $\widehat{m}(z)$ and $\widetilde{m}(z)$, depending on parameters $s$ and $\tau$, are introduced as the linear combinations of the columns of the solution to the tacnode RHP. By the relation \eqref{eq:Kuijlaars_compare} between the tacnode RHP and RHP \ref{RHP:4x4}, we have in our notations
\begin{equation} \label{crit:3}
  \begin{split}
    \widetilde{m}(z; s, \tau) = {}& n^{(1)}(z; 1, 1, s, s, \tau) + n^{(2)}(z; 1, 1, s, s, \tau), \\
    \widehat{m}(z; s, \tau) = {}& n^{(0)}(z; 1, 1, s, s, \tau) -n^{(3)}(z; 1, 1, s, s, \tau).
  \end{split}
\end{equation}

The critical kernel $K^{\crit}_2(x, y; s, \tau)$ has the expression \cite[Theorem 2.9 and Formula (4.13)]{Kuijlaars14}
\begin{multline} \label{eq:first_formula_of_K^cr_2}
  K^{\crit}_2(x, y; s, \tau) = \frac{1}{2\pi i(x - y)} \left[ \widetilde{m}_1(ix; s, -\tau) \widehat{m}_4(iy; s, \tau) + \widetilde{m}_2(ix; s, -\tau) \widehat{m}_3(iy; s, \tau) \right. \\
  - \left. \widetilde{m}_3(ix; s, -\tau) \widehat{m}_2(iy; s, \tau) - \widetilde{m}_4(ix; s, -\tau) \widehat{m}_1(iy; s, \tau) \right],
\end{multline}
and for the proof of Theorem \ref{critical} we also need \cite[equation (4.13)]{Kuijlaars14}
\begin{equation}\label{crit:3a}
  \frac{\partial}{\partial s} K^{\crit}_2(x,y; s,\tau) = \frac{-1}{\pi i} \left( \widetilde{m}_1(ix; s, -\tau) \widehat{m}_1(iy; s, \tau) + \widetilde{m}_2(ix; s, -\tau)\widehat{m}_2(iy; s, \tau) \right) ds.
\end{equation}

Now we use the integral formulas of $n^{(3)}$ and $n^{(0)}$ to express the functions $\widehat{m}_1$ and $\widehat{m}_2$ in terms of the entries of $\Psi^{(0)}(\zeta; \sigma)$ which are defined in Section \ref{subsec:Flaschka-Newell_Lax}. First we consider $n^{(0)}$. By the definition given by \eqref{Q_def}--\eqref{eq:form_f_and_g}, we have that each component of $n^{(0)}(z)$ is expressed by a sum of integrals on $\Gamma^{(0)}_1$, $\Gamma^{(0)}_2$, and $\Gamma^{(0)}_3$. Since $n^{(0)}(z)$ is associated to the Hastings--McLeod solution, and then $(t_1, t_2, t_3) = (1, 0, -1)$, by \eqref{eq:def_of_nj2} for any component of $n^{(0)}(z)$, the integrand on $\Gamma^{(0)}_2$ vanishes, and the integrands on $\Gamma^{(0)}_1$ and $\Gamma^{(0)}_3$ are identical. Similarly, we have that each component of $n^{(3)}(z)$ is expressed by a sum of integrals on $\Gamma^{(3)}_1$, $\Gamma^{(3)}_2$, and $\Gamma^{(3)}_3$, such that the integrand on $\Gamma^{(3)}_2$ vanishes and the integrands on $\Gamma^{(3)}_1$ and $\Gamma^{(3)}_3$ are identical. Furthermore, with $\Sg_{\tac}$ as in Figure \ref{fig:Sigma_T}, we can deform $\Gamma^{(0)}_1 \cup \Gamma^{(0)}_3$ into the upper half of $\Sigma_{\tac}$ and $\Gamma^{(3)}_1 \cup \Gamma^{(3)}_3$ into the lower half of $\Sigma_{\tac}$, both with reversed orientation. Thus we can write, after expressing the integrands in \eqref{Q_def} by \eqref{eq:def_of_nj2}, \eqref{eq:form_f_and_g}, and \eqref{eq:jump_condition_2x2},
\begin{align}
  \widehat{m}_1(z; s, \tau) = {}& \sqrt{\frac{2}{\pi}} \frac{e^{-\frac{\tau^3}{3} + 2 s \tau}}{2^{1/3}} \int_{\Sg_{\tac}} e^{2^{4/3}\tau \zeta^2 + 2^{2/3}i z \zeta} f(\zeta; 2^{2/3}(2s - \tau^2)) \,d\zeta, \label{crit:7_1} \\
  \widehat{m}_2(z; s, \tau) = {}& \sqrt{\frac{2}{\pi}} \frac{e^{-\frac{\tau^3}{3} + 2 s \tau}}{2^{1/3}} \int_{\Sg_{\tac}} e^{2^{4/3}\tau \zeta^2 + 2^{2/3}i z \zeta} g(\zeta; 2^{2/3}(2s - \tau^2)) \,d\zeta, \label{crit:7_2}
\end{align}
where the functions $f(\zeta; 2^{2/3}(2s - \tau^2))$ and $g(\zeta; 2^{2/3}(2s - \tau^2))$ the contour $\Sg_{\tac}$ are defined in equation \eqref{tac:8} and in Figure \ref{fig:Sigma_T}, respectively. By the same argument, we have that $\widehat{m}_3(z; s, \tau)$ and $\widehat{m}_4(z; s, \tau)$ have similar but slightly more complicated formulas as integrals on $\Sigma_{\tac}$. Here we note that the contour $\Sigma_{\tac}$ can be replaced by $\Sigma_{\tMM}$ where the definition for $f$ and $g$ is still given by \eqref{tac:8}.

Next we use the integral formulas of $n^{(1)}$ and $n^{(2)}$ to express the functions $\widetilde{m}_1$ and $\widetilde{m}_2$. Similar to the argument for $\widehat{m}_1$ and $\widehat{m}_2$, because $n^{(1)}$ and $n^{(2)}$ are associated to the Hastings--McLeod solution, the integrands on $\Gamma^{(1)}_2$ and $\Gamma^{(1)}_3$ for the integral formula of the $j$-th component of $n^{(1)}$ are identical to the integrands on  $\Gamma^{(2)}_3$ and $\Gamma^{(2)}_1$ respectively for the integral formula of the $j$-th component of $n^{(2)}$, and the integrand on $\Gamma^{(1)}_1$ for the integral formula of the $j$-th component of $n^{(1)}$ is the negative of the integrand on $\Gamma^{(2)}_2$ for the integral formula of the $j$-th component of $n^{(2)}$. Using the contours $\Gamma^{(k)}_j$ shown in Figure \ref{fig:six_contours}, we find that the integrals on $\Gamma^{(1)}_1$ and $\Gamma^{(1)}_2$ cancel the integrals on $\Gamma^{(2)}_2$ and $\Gamma^{(2)}_3$ respectively when we compute $n^{(1)} + n^{(2)}$. Hence by \eqref{eq:def_of_nj2}, \eqref{eq:form_f_and_g}, and \eqref{eq:jump_condition_2x2}, we derive analogous to \eqref{crit:7_1} and \eqref{crit:7_2},
\begin{align}
  \widetilde{m}_1(z; s, -\tau) = {}& -\sqrt{\frac{2}{\pi}}\frac{e^{\frac{\tau^3}{3} - 2 s \tau}}{2^{1/3}} \int_{-i\infty}^{i\infty} e^{-2^{4/3}\tau \zeta^2 + 2^{2/3}i z \zeta} \Phi_1(\zeta; 2^{2/3}(2s - \tau^2)) \,d\zeta, \label{crit:6_1} \\
  \widetilde{m}_2(z; s, -\tau) = {}& -\sqrt{\frac{2}{\pi}}\frac{e^{\frac{\tau^3}{3} - 2 s \tau}}{2^{1/3}} \int_{-i\infty}^{i\infty} e^{-2^{4/3}\tau \zeta^2 + 2^{2/3}i z \zeta} \Phi_2(\zeta; 2^{2/3}(2s - \tau^2)) \,d\zeta,  \label{crit:6_2}
\end{align}
where $\Phi_1(\zeta; 2^{2/3}(2s - \tau^2))$ and $\Phi_2(\zeta; 2^{2/3}(2s - \tau^2))$ are defined in \eqref{eq:formula_Phi_1_2}. Similarly, $\widetilde{m}_3(z; s, -\tau)$ and $\widetilde{m}_4(z; s, -\tau)$ have similar but slightly more complicated formulas as integrals on the imaginary axis.

Plugging in \eqref{crit:7_1}, \eqref{crit:7_2}, \eqref{crit:6_1} and \eqref{crit:6_2} into \eqref{crit:3a}, we find
\begin{multline}\label{crit:11}
  \frac{\partial}{\partial s} K_2^{\crit}(x,y; s,\tau) = \frac{-2^{1/3}}{ i\pi^2}  \int_{-i\infty}^{i\infty} \,du\,\int_{\Sg_{\tac}} \,dv\, e^{-2^{4/3}\tau (u^2-v^2)-2^{2/3} (xu+y v)} \\
  \times  \bigg(\Phi_1(u;  2^{2/3}(2s - \tau^2)) f(v; 2^{2/3}(2s - \tau^2)) + \Phi_2(u; 2^{2/3}(2s - \tau^2)) g(v; 2^{2/3}(2s - \tau^2))\bigg).
\end{multline}
Next we make the change of variable $u\mapsto (-u)$. We note that the contour for $u$ changes direction under this transform. Hence by \eqref{eq:u_to_-u_Phi}, \eqref{crit:11} becomes
\begin{multline}\label{crit:12}
  \frac{\partial}{\partial s} K_2^{\crit}(x,y; s,\tau) = \frac{2^{1/3}}{i \pi^2} \int_{-i\infty}^{i\infty} \,du\,\int_{\Sg_{\tac}} \,dv\, e^{-2^{4/3}\tau (u^2-v^2)+2^{2/3} (xu-y v)} \\
  \times  \bigg(\Phi_2(u;  2^{2/3}(2s - \tau^2)) f(v; 2^{2/3}(2s - \tau^2)) + \Phi_1(u; 2^{2/3}(2s - \tau^2)) g(v; 2^{2/3}(2s - \tau^2))\bigg).
\end{multline}

The right-hand side of the above equation can also be expressed as a derivative with respect to $s$. Indeed, using \eqref{tac:14} we find
\begin{equation}\label{crit:13}
  \frac{\d}{\d \sigma} \left(\frac{\Phi_2(u; \sigma)f(v; \sigma)-\Phi_1(u; \sigma)g(v; \sigma)}{i(u-v)}\right)=\Phi_2(u; \sigma)f(v; \sigma)+\Phi_1(u; \sigma)g(v; \sigma).
\end{equation}
We also note that in \eqref{crit:12}, if we deform the integral contour $\Sigma_{\tac}$ into $\Sigma_{\tMM}$, with the definition of $f(v; \sg)$ and $g(v; \sg)$ given in \eqref{tac:8} when $v \in \realR$, the integral on the right-hand side does not change. Therefore in \eqref{crit:12} we can replace $\Sigma_{\tac}$ with $\Sigma_{\tMM}$. Let us denote
\begin{multline} \label{eq:witehat_K}
  \widehat{K}^{\crit}_2(x, y; s, \tau) = \frac{1}{2^{1/3}\pi} \int_{-i\infty}^{i\infty} \,du\,\int_{\Sg_{\tMM}} \,dv\, e^{-2^{4/3}\tau (u^2-v^2)+2^{2/3} (xu-y v)} \\
  \times \left(\frac{\Phi_1(u;\sg)g(v; \sg)-\Phi_2(u;\sg)f(v; \sg)}{2\pi(u-v)}\right) ,
\end{multline}
where $\sigma = 2^{2/3}(2s - \tau^2)$ as in \eqref{eq:expression_sigma_K^cr_2}. Notice that without the deformation of $\Sigma_{\tac}$ into $\Sigma_{\tMM}$, the integral in \eqref{eq:witehat_K} would be singular when the contours cross, and would have to be be regarded as a principal value integral. The deformation removes the singularity. Then \eqref{crit:12} and \eqref{crit:13} imply that
\begin{equation} \label{eq:derivatives_equal_2MM}
  \frac{d}{ds} K^{\crit}_2(x, y; s, \tau) = \frac{d}{ds} \widehat{K}^{\crit}_2(x, y; s, \tau).
\end{equation}
If we can show
\begin{align}
  \lim_{\sigma \to -\infty} K^{\crit}_2(x, y; s, \tau) = {}& 0, \label{eq:vanishing_2MM_old} \\
  \lim_{\sigma \to -\infty} \widehat{K}^{\crit}_2(x, y; s, \tau) = {}& 0,  \label{eq:vanishing_2MM_new}
\end{align}
then \eqref{critkernel} follows from \eqref{eq:derivatives_equal_2MM}. Below we prove \eqref{eq:vanishing_2MM_old} and \eqref{eq:vanishing_2MM_new}.

We need the asymptotic behavior of $\Psi^{(0)}(\zeta; \sigma), \dotsc, \Psi^{(5)}(\zeta; \sigma)$, the fundamental solutions to \eqref{int:2a}, as $\sigma \to -\infty$, when $q(\sigma)$ is the Hastings-McLeod solution to \eqref{int:1}, or equivalently, when $(t_1, t_2, t_3) = (1, 0, -1)$ in \eqref{eq:triple_number}. The result was obtained in \cite[Section 6]{Deift-Zhou95}, and we summarize it below.

Suppose $\sigma < 0$. We define the $2 \times 2$ matrix-valued function $m^{(23)}(z)$, where we follow the notational convention in \cite{Deift-Zhou95} and suppress the dependence on $\sigma$, by
\begin{equation}
  m^{(23)}(z) e^{-i \frac{\sqrt{2}}{3} (-\sigma)^{3/2} (z^2 - 1)^{3/2} \sigma_3} =
  \begin{cases}
    \Psi^{(0)} \left( \sqrt{\frac{-\sigma}{2}} z; \sigma \right) = \Psi^{(3)} \left( \sqrt{\frac{-\sigma}{2}} z; \sigma \right) & \text{if $z \in C_1 \cup C_3$}, \\
    \Psi^{(1)} \left( \sqrt{\frac{-\sigma}{2}} z; \sigma \right) = \Psi^{(2)} \left( \sqrt{\frac{-\sigma}{2}} z; \sigma \right) & \text{if $z \in C_2$}, \\
    \Psi^{(4)} \left( \sqrt{\frac{-\sigma}{2}} z; \sigma \right) = \Psi^{(5)} \left( \sqrt{\frac{-\sigma}{2}} z; \sigma \right) & \text{if $z \in C_4$},
  \end{cases}
\end{equation}
where $C_1, C_2, C_3, C_4$ are regions of the complex plane as shown in Figure \ref{fig:last_RHP}, and the power function $(z^2-1)^{3/2}$ has a cut on $[-1,1]$ taking the branch such that $(z^2-1)^{3/2}\sim z^3$ as $z\to\infty$. Then $m^{(23)}(z) = I + \bigO(z^{-1})$ as $z \to \infty$, and it satisfies the jump condition as shown in Figure \ref{fig:last_RHP}, cf.~\cite[Figure 6.18]{Deift-Zhou95}.
\begin{figure}[htb]
  \centering
  \includegraphics{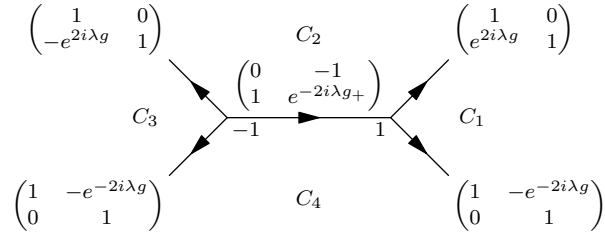}
  \caption{The jump condition for $m^{(23)}(z)$, where $\lambda = 2^{1/2} 3^{-1} (-\sigma)^{3/2}$ and $g = g(z) = (z^2 - 1)^{3/2}$.}
  \label{fig:last_RHP}
\end{figure}

Then as observed in \cite{Deift-Zhou95}, the RHP for $m^{(23)}(z)$ converges formally to a RHP on the interval $[-1, 1]$ with jump matrix $\left(
\begin{smallmatrix}
  0 & -1 \\
  1 & 0
\end{smallmatrix} \right)$, and this RHP has a simple solution:
\begin{equation}
  P^{(\infty)}(z) = \frac{1}{2}
  \begin{pmatrix}
    a(z) + a(z)^{-1} & i(a(z) - a(z)^{-1}) \\
    i(a(z)^{-1} - a(z) & a(z) + a(z)^{-1}
  \end{pmatrix},
  \quad \text{where} \quad a(z) = \frac{(z - 1)^{1/4}}{(z + 1)^{1/4}}.
\end{equation}
The function $a(z)$ has a cut on $[-1,1]$ taking the branch of the fractional power so that $a(z)\sim 1$ as $z\to\infty$.
By constructing local parametrices at $1$ and $-1$, the convergence can be made rigorous, as discussed in \cite[Section 3]{Deift-Zhou95}. By standard argument, we derive the following result:
\begin{lem} \label{lem:last}
  Fix $\epsilon > 0$. As $\sigma \to -\infty$, $\lVert m^{(23)}(z) - P^{(\infty)}(z) \rVert \to 0$ uniformly for all $\{ z \in \compC \mid \lvert z - 1 \rvert > \epsilon \text{ and } \lvert z + 1 \rvert > \epsilon \}$. Here if $z$ is on the jump curve, then the uniform convergence holds for both $m^{(23)}_+(z)$ and $m^{(23)}_-(z)$.
\end{lem}

By Lemma \ref{lem:last} and the asymptotics of $\Psi^{(k)}(\zeta; \sigma)$ implied by it, we use the explicit formulas \eqref{crit:6_1} and \eqref{crit:6_2} for $\widetilde{m}_1(z; s, -\tau)$ and $\widetilde{m}_2(z; s, -\tau)$ and derive that as $s \to -\infty$, $\widetilde{m}_1(z; s, -\tau)$ and $\widetilde{m}_2(z; s, -\tau)$ vanishes exponentially. In a similar way, we have that $\widetilde{m}_3(z; s, -\tau)$ and $\widetilde{m}_4(z; s, -\tau)$ also vanishes exponentially.

To estimate $\widehat{m}_j(z; s, \tau)$, as $s \to -\infty$, it is better to replace the contour $\Sigma_{\tac}$ by $\Sigma^{\sigma}_{\tMM}$ that depends on $\sigma = 2^{2/3}(2s - \tau^2) < 0$, which is simply the contour $\Sigma_{\tMM}$ scaled by factor $\sqrt{-\sigma}/2$:
\begin{multline}
  \Sg^{\sigma}_{\tMM} =[\sqrt{-\sigma}/2, e^{i\pi/6}\cdot \infty) \cup [\sqrt{-\sigma}/2, e^{-i\pi/6}\cdot \infty)\cup \\
  [-\sqrt{-\sigma}/2, e^{5i\pi/6}\cdot \infty) \cup [-\sqrt{-\sigma}/2, e^{-5i\pi/6}\cdot \infty)\cup[-\sqrt{-\sigma}/2, \sqrt{-\sigma}/2].
\end{multline}
Then by direct computation based on \eqref{crit:7_1} and \eqref{crit:7_2}, we can find that $\widehat{m}_1(z; s, \tau)$ and $\widehat{m}_2(z; s, \tau)$ vanish exponentially as $s \to -\infty$. Similarly, we can show that $\widehat{m}_1(z; s, \tau)$ and $\widehat{m}_2(z; s, \tau)$ also do. Hence we prove \eqref{eq:vanishing_2MM_old} by plugging in the exponential vanishing property stated above to \eqref{eq:first_formula_of_K^cr_2}.

Similarly, if we replace $\Sigma_{\tMM}$ by $\Sigma^{\sigma}_{\tMM}$ when we evaluate the double contour integral in \eqref{eq:witehat_K}, we find it vanishes exponentially as $s \to -\infty$, and \eqref{eq:vanishing_2MM_new} holds. The detail is left to the reader. Thus Theorem \ref{critical} is proved.

\subsection{Proof of Theorem \ref{tacnode}}\label{tacnode_kernel_proof}

The multi-time correlation kernel for the tacnode process was first derived in the form of Airy resolvents, and it was presented in the most general form in \cite{Ferrari-Veto12}, where the notation $\lcal^{\lambda, \Sigma}_{\tac}(\tau_1, x; \tau_2, y)$ is defined. We use this notation in Theorem \ref{tacnode}. In \cite{Delvaux13}, it was shown that the kernel can also be represented in the form of the tacnode RHP. In this section we prove Theorem \ref{tacnode} based on the formula of in \cite{Delvaux13}. 

First we recall that in \cite[RHP 2.1 in Section 2.1.1]{Delvaux13} a Riemann--Hilbert problem is defined, and it is essentially the same as the tacnode RHP in \cite{Kuijlaars14} that is discussed in Section \ref{subsec:tacnode_RHP}. In \cite[Section 2.1.1]{Delvaux13}, a matrix $\widehat{M}(z)$ is defined with parameters $r_1, r_2, s_1, s_2, \tau$, and without much difficulty we can check that
\begin{equation}
  \widehat{M}(z; r_1, r_2, s_1, s_2, \tau) = e^{\frac{r^2_1 - r^2_2}{2} \tau z} \left. M^{(1)} \left( z; r_1, r_2, s_1, s_2, \frac{r^2_1 + r^2_2}{2} \tau \right) \right\rvert_{(t_1, t_2, t_3) = (1, 0, -1)},
\end{equation}
where $M^{(1)}(z; r_1, r_2, s_1, s_2, \tau)$ is the solution to RHP \ref{RHP:4x4} in sector $\Delta_1$, associated to the Hastings--McLeod. Then the vector function $\mathbf{p}(z)$ defined in\cite[Formula (2.9)]{Delvaux13}, which is the sum of the first and second columns of $\widehat{M}(z)$, become
\begin{multline}
  \mathbf{p}(z) = \left( p_j(z; r_1, r_2, s_1, s_2, \tau) \right)^T_{j = 1, \dotsc, 4} = \\
  e^{\frac{r^2_1 - r^2_2}{2} \tau z} \left( n^{(0)} \left( z; r_1, r_2, s_1, s_2, \frac{r^2_1 + r^2_2}{2} \tau \right) - n^{(3)} \left( z; r_1, r_2, s_1, s_2, \frac{r^2_1 + r^2_2}{2} \tau \right) \right).
\end{multline}
Then the functions $\widehat{p}_j(z; \tilde{s}, \tau)$ ($j = 1, 2$), which are denoted in \cite[Formula (2.26)]{Delvaux13} and are related to the first two components of $\mathbf{p}(z)$ by \cite[Formula 4.37]{Delvaux13}, becomes
\begin{multline} \label{eq:phat_1}
  \widehat{p}_1(z; s, \tau) = \frac{\lambda^{-1/24}}{\sqrt{2\pi}} \exp \left( -\lambda \tau \left( \lambda^{-1/2} C^{-2} s + \frac{2}{3} \tau^2 \right) \right) \\
  \times p_1 \left( z; \lambda^{1/4}, 1, \frac{\lambda^{3/4}}{2} \left( \lambda^{-1/2} C^{-2} s + \tau^2 \right), \frac{1}{2} \left( \lambda^{-1/2} C^{-2} s + \tau^2 \right), \tau \right),
\end{multline}
\begin{multline} \label{eq:phat_2}
  \widehat{p}_2(z; s, \tau) = \frac{1}{\sqrt{2\pi}} \exp \left( -\tau \left( \lambda^{-1/2} C^{-2} s + \frac{2}{3} \tau^2 \right) \right) \\
  \times p_2 \left( z; \lambda^{1/4}, 1, \frac{\lambda^{3/4}}{2} \left( \lambda^{-1/2} C^{-2} s + \tau^2 \right), \frac{1}{2} \left( \lambda^{-1/2} C^{-2} s + \tau^2 \right), \tau \right),
\end{multline}
where $p_1, p_2$ are components of $\mathbf{p}$, $C$ is defined in \eqref{eq:defn_C_for_tacnode} and where $\lambda > 0$ is the parameter in the correlation kernel formula $\lcal^{\lambda, \Sigma}_{\tac}$. Then by \cite[Theorem 2.9]{Delvaux13}, (noting that our notation $\lcal^{\lambda, \Sigma}_{\tac}$ means the same as $\lcal^{\lambda, \sigma}_{\tac}$ in \cite{Delvaux13} if $\sigma$ and $\Sigma$ are related by \cite[Formula (2.15)]{Delvaux13}, or equivalently \eqref{eq:defn_C_for_tacnode}),
\begin{equation}
  \lcal_{\tac}^{\la, \Sigma}(\tau_1, x; \tau_2, y) = -1_{\tau_1 < \tau_2} \frac{1}{\sqrt{4\pi(\tau_2 - \tau_1)}} \exp \left( \frac{(y - x)^2}{4(\tau_2 - \tau_1)} \right) + \tilde{\lcal}^{\lambda, \Sigma}_{\tac}(\tau_1, x; \tau_2, y),
\end{equation}
where, with $\sigma = \lambda^{1/2} C^2 \Sigma$ as specified in \eqref{eq:defn_C_for_tacnode},
\begin{equation}
  \tilde{\lcal}^{\lambda, \Sigma}_{\tac}(\tau_1, x; \tau_2, y) = \frac{1}{C^2} \int^{\infty}_{\sigma} \left( \lambda^{1/3} \widehat{p}_1(x; s, \tau_1) \widehat{p}_1(y; s, -\tau_2) + \lambda^{-1/2} \widehat{p}_2(x; s, \tau_1) \widehat{p}_2(y; s, -\tau_2) \right) ds.
\end{equation}
The vector-valued function $\mathbf{p}(z)$ is in the same form as $\widehat{m}(z)$ defined in \eqref{crit:3} with more general parameters. Thus similar to \eqref{crit:7_1} and \eqref{crit:7_2}, we can write $p_1(z)$ and $p_2(z)$ as integrals on contour $\Sigma_{\tac}$, and then have by \eqref{eq:phat_1} and \eqref{eq:phat_2}
\begin{align}
  \widehat{p}_1(z; s, \tau) = {}& \frac{1}{\pi (1 + \sqrt{\lambda})^{1/3}} \int_{\Sigma_{\tac}} f(\zeta; s) \exp \left[ -\frac{4 i}{3} \frac{1 - \sqrt{\lambda}}{1 + \sqrt{\lambda}} \zeta^3 + \frac{4\tau}{C^2} \zeta^2 + \left( \frac{2iz}{C} + i \frac{1 - \sqrt{\lambda}}{1 + \sqrt{\lambda}} s \right) \zeta \right] d\zeta, \\
  \widehat{p}_2(z; s, \tau) = {}& \frac{1}{\pi (1 + 1/\sqrt{\lambda})^{1/3}} \int_{\Sigma_{\tac}} g(\zeta; s) \exp \left[ -\frac{4 i}{3} \frac{1 - \sqrt{\lambda}}{1 + \sqrt{\lambda}} \zeta^3 + \frac{4\tau}{C^2} \zeta^2 + \left( \frac{2iz}{C} + i \frac{1 - \sqrt{\lambda}}{1 + \sqrt{\lambda}} s \right) \zeta \right] d\zeta,
\end{align}
where $f$ and $g$ are defined in \eqref{tac:8}. 

In this section, we need a technical lemma:
\begin{lem} \label{lem:tacnode_vanishing}
  Let $\epsilon > 0$ and $N \in \realR$. Then there exists $C(\epsilon, N) > 0$ such that for all $\sigma > N$, and $k = 1, 2$
  \begin{align}
    \Psi^{(0)}_{k, 2}(\zeta; \sigma) = {}& e^{i \left( \frac{4}{3} \zeta^3 + \sigma \zeta \right)} \bigO(1), & \text{for all $\zeta \in \compC$ such that $\Im(\zeta) > \epsilon$}, \label{eq:vanishing_tacnode_1} \\
    \Psi^{(0)}_{k, 1}(\zeta; \sigma) = {}& e^{-i \left( \frac{4}{3} \zeta^3 + \sigma \zeta \right)} \bigO(1), & \text{for all $\zeta \in \compC$ such that $\Im(\zeta) < -\epsilon$}. \label{eq:vanishing_tacnode_2}
  \end{align}
\end{lem}
\begin{proof}
  We prove \eqref{eq:vanishing_tacnode_1}, and the proof of \eqref{eq:vanishing_tacnode_2} is analogous. The Airy resolvent formulas in \cite{Baik-Liechty-Schehr12} yield (see \cite[Formulas (2.38) and (2.39)]{Kuijlaars14})
  \begin{equation}
    \Psi^{(0)}_{1, 2}(\zeta; \sigma) = e^{i \left( \frac{4}{3} \zeta^3 + \sigma \zeta \right)} \tilde{f}(\zeta; \sigma),\quad \Psi^{(0)}_{1, 2}(\zeta; \sigma) = e^{i \left( \frac{4}{3} \zeta^3 + \sigma \zeta \right)}  \tilde{g}(\zeta; \sigma),
  \end{equation}
  where
  \begin{equation}
     \tilde{f}(\zeta; \sigma) = -\int^{\infty}_{\sigma} e^{2i(x - \sigma)\zeta} Q_{\sigma}(x) dx, \quad \tilde{g}(\zeta; \sigma) = 1 + \int^{\infty}_{\sigma} e^{2i(x - \sigma)\zeta} R_{\sigma}(x, \sigma) dx,
  \end{equation}
  and the definitions of $Q_{\sigma}(x)$ and $R_{\sigma}(x, \sigma)$ are given in \cite[Formulas (2.18) and (2.19)]{Kuijlaars14}. For all $\sigma > N$, $Q_{\sigma}(x)$ and $R_{\sigma}(x, \sigma)$ decays at the speed comparable to the Airy function, so that that $\tilde{f}(\zeta; \sigma)$ and $\tilde{g}(\zeta; \sigma)$ are bounded and \eqref{eq:vanishing_tacnode_1} is proved.
\end{proof}

Then we can write the kernel $\tilde{\lcal}^{\lambda, \Sigma}_{\tac}(\tau_1, x; \tau_2, y)$ as
\begin{multline} \label{tac:12_new}
  \tilde{\lcal}^{\lambda, \Sigma}_{\tac}(\tau_1, x; \tau_2, y) = \frac{1}{\pi C} \int_{\Sigma_{\tac}} du \int_{\Sigma_{\tac}} dv\, e^{-\frac{4i}{3} \frac{1 - \sqrt{\lambda}}{1 + \sqrt{\lambda}} (u^3 + v^3) + \frac{4(\tau_1 u^2 - \tau_2 v^2)}{C^2} + \frac{2i(xu + yv)}{C}} \\
  \times \frac{1}{\pi C^3} \int^{\infty}_{\sigma} \exp \left( i\frac{1 - \sqrt{\lambda}}{1 + \sqrt{\lambda}} s (u + v) \right) \left[ f(u; s) f(v; s) + \frac{1}{\sqrt{\lambda}} g(u; s) g(v; s) \right] ds,
\end{multline}
where the change of order of integrations is justified by Lemma \ref{lem:tacnode_vanishing}. Next we make the change of variable $v\mapsto (-v)$. Note now that the contour $\Sg_{\tac}$ is invariant under this transform. Hence by \eqref{eq:u_to_-u_Phi}, \eqref{tac:12_new} becomes
\begin{multline} \label{tac:12_new_new}
  \tilde{\lcal}^{\lambda, \Sigma}_{\tac}(\tau_1, x; \tau_2, y) = \frac{1}{\pi C} \int_{\Sigma_{\tac}} du \int_{\Sigma_{\tac}} dv\, e^{-\frac{4i}{3} \frac{1 - \sqrt{\lambda}}{1 + \sqrt{\lambda}} (u^3 - v^3) + \frac{4(\tau_1 u^2 - \tau_2 v^2)}{C^2} + \frac{2i(xu - yv)}{C}} \\
  \times \frac{1}{\pi C^3} \int^{\infty}_{\sigma} \exp \left( i\frac{1 - \sqrt{\lambda}}{1 + \sqrt{\lambda}} s (u - v) \right) \left[ f(u; s) g(v; s) + \frac{1}{\sqrt{\lambda}} g(u; s) f(v; s) \right] ds.
\end{multline}
By \eqref{tac:14}, we have
\begin{multline}\label{tac:15}
  \frac{\d}{\d s} \left[ \frac{i}{2} (1 + \la^{-1/2}) e^{\frac{is(1 - \sqrt{\la})}{1 + \sqrt{\la}}(u - v)} \frac{f(u; s)g(v; s)-f(v; s)g(u; s)}{u-v}\right] \\
  = \exp \left( i\frac{(1 - \sqrt{\la})}{1 + \sqrt{\la}} s(u - v) \right) \left[ f(u; s)g(v; s)+\frac{f(u; s)g(v; s)}{\sqrt{\la}} \right]. 
\end{multline} 
Using this identity, we can write \eqref{tac:12_new_new} as 
\begin{multline}\label{tac:15a}
  \tilde{\lcal}_{\tac}^{\la, \Sigma}(\tau_1, x; \tau_2, y) = \frac{1}{\pi C} \int_{\Sg_{\tac}}\,du\int_{\Sg_{\tac}}\,dv\, e^{-\frac{4i}{3}(\frac{1-\sqrt{\la}}{1+\sqrt{\la}})(u^3-v^3)+\frac{4}{C^2}(\tau_1u^2-\tau_2v^2)+\frac{2i}{C}(xu-yv)} \\
  \times \frac{i(1 + \la^{-1/2})}{2 \pi C^3} \int^\infty_{\sigma} \frac{\d}{\d s} \left[\frac{e^{\frac{is(1-\sqrt{\la})}{1+\sqrt{\la}}(u-v)}\left(f(u;s)g(v;s)-f(v;s)g(u;s)\right)}{u-v}\right] \,ds.
\end{multline}
Performing the integration in $s$ and using $C^3=1+\lambda^{-1/2}$, we derive
\begin{multline}
    \tilde{\lcal}_{\tac}^{\la, \Sigma}(\tau_1, x; \tau_2, y)=\frac{1}{C\pi} \int_{\Sg_{\tac}}\,du\int_{\Sg_{\tac}}\,dv\, e^{-\frac{4i}{3}(\frac{1-\sqrt{\la}}{1+\sqrt{\la}})(u^3-v^3)+\frac{4}{C^2}(\tau_1u^2-\tau_2v^2)+\frac{2i}{C}(xu-yv)+\frac{i\sg(1-\sqrt{\la})}{1+\sqrt{\la}}(u-v)} \\
    \times \left[\frac{\left(f(u; \sigma)g(v; \sigma)-f(v; \sigma)g(u; \sigma)\right)}{2\pi i(u-v)}\right],
\end{multline}
given that
\begin{equation} \label{eq:last_condition}
  f(u; \sigma) g(v; \sigma) - f(v; \sigma) g(u; \sigma) \to 0, \quad \text{as $\sigma \to +\infty$, for all $u, v \in \Sigma_{\tac}$.} 
\end{equation}
We thus have that \eqref{eq:last_condition} is implied by Lemma \ref{lem:tacnode_vanishing}. Thus we finish the proof of Theorem \ref{tacnode}.

\appendix

\section{Formulas of $V_1$, $V_2$, and $W$ in \eqref{int:7_b}} \label{sec:formulas_VVW}

In order to present the coefficient matrices of \eqref{int:7_b}, it is convenient to introduce several notations which were used in \cite{Delvaux13}. Below we define several parameters which depend on $r_{1,2}$, $s_{1,2}$, and $\tau$. These are the same notations given in \cite[Theorem 6.2]{Delvaux13} up to the rescaling $\tau \mapsto 2\tau/(r_1^2+r_2^2)$. The quantities $C$ and $\ga$ are the ones which were defined in \eqref{int:6}; $q$ and $q'$, are the PII solution and its derivative; $u$ is the PII Hamiltonian defined in \eqref{int:10}; and all Painlev\'{e} functions are evaluated at the point $\sg$ which is defined in \eqref{int:9}.  Other notations in this appenix are independent of the rest of the paper. In particular the parameters $b$ and $c$ given below are {\it not} the same ones as in \eqref{diffeq:2}. We use these symbols to match the notation of \cite{Delvaux13} and we trust it will not cause confusion. Define the parameters
\begin{multline}
b=\frac{1}{Cr_2 \sqrt{r_1 r_2} \ga}\left[\left(s_2^2+\frac{2r_2^2\tau}{r_1^2+r_2^2}\right)q-\frac{uq+q'}{C}\right], \quad \tilde{b} =\frac{\ga}{Cr_1 \sqrt{r_1 r_2}}\left[\left(s_1^2+\frac{2r_1^2\tau}{r_1^2+r_2^2}\right)q-\frac{uq+q'}{C}\right], \\
\be= \frac{\ga}{Cr_2 \sqrt{r_1 r_2}}\left[\left(s_2^2-\frac{2r_2^2\tau}{r_1^2+r_2^2}\right)q-\frac{uq+q'}{C}\right], \quad \tilde{\be}= \frac{1}{Cr_1 \sqrt{r_1 r_2}\ga}\left[\left(s_1^2-\frac{2r_1^2\tau}{r_1^2+r_2^2}\right)q-\frac{uq+q'}{C}\right], \\
\end{multline}
as well as
\begin{equation}
d= \frac{q}{C\sqrt{r_1r_2}\ga}, \qquad \tilde{d}=\frac{q\ga}{C\sqrt{r_1r_2}}, \qquad c=\frac{s_1^2}{r_1}-\frac{u}{r_1C}, \qquad \tilde{c}= \frac{s_2^2}{r_2}- \frac{u}{r_2C}.
\end{equation}
Also introduce the notations
\begin{multline}
f=\frac{4r_2C^{-1}\ga^{-1}}{(r_1^2+r_2^2)r_1\sqrt{r_1r_2}} \left[\frac{q' \tau}{C}+\frac{(r_1^2-r_2^2)\tau^2q}{r_1^2+r_2^2}-\frac{(s_1r_1-s_2r_2)q}{2}\right]+b\left(- c -\frac{r_2}{r_1} \tilde{c} +\frac{(r_1^2+r_2^2)\tau}{r_1}\right) \\
-d^2\tilde{d}+\frac{r_2}{r_1}\tilde{c}^2 d -\frac{2 s_2 d}{r_1},
\end{multline}
\begin{multline}
\tilde{f}=\frac{4r_1C^{-1}}{(r_1^2+r_2^2)r_2\sqrt{r_1r_2}} \left[\frac{q' \tau}{C}-\frac{(r_1^2-r_2^2)\tau^2q}{r_1^2+r_2^2}+\frac{(s_1r_1-s_2r_2)q}{2}\right]+\tilde{b}\left(- \frac{r_1}{r_2}c - \tilde{c} +\frac{(r_1^2+r_2^2)\tau}{r_2}\right) \\
-\tilde{d}^2 d+\frac{r_1}{r_2} c^2 \tilde{d} -\frac{2 s_1 \tilde{d}}{r_2}.
\end{multline}

Then the matrices $V_1$, $V_2$, and $W$ in \eqref{int:7_b} are given below. Note that besides the scaling of $\tau$, the solution to the Lax system in \cite{Delvaux13} differs from (the Hastings--McLeod case of) the solution to Lax system \eqref{int:7} by a scalar factor $\exp(-\tau z(r^2_1 - r^2_2)/(r^2_1 + r^2_2))$, so our Lax system is slightly different from \cite[Formulas (6.21)--(6.24)]{Delvaux13} in $U$ and $W$.
\begin{equation}
\begin{aligned}
  V_{1} &= 2 
  \begin{pmatrix}
     c & 0 & -i & 0 \\
    \tilde{d} & 0 & 0 &  0 \\
    i\left(-c^2+\frac{r_2}{r_1} d\tilde{d}+\frac{s_1}{r_1}-z\right)  & i  b &  - c & - d  \\
    i\be & 0 & 0 & 0
  \end{pmatrix}, \\
  V_{2} &= 2
  \begin{pmatrix}
     0& \ d & 0   & 0 \\
    0 & \tilde{c} & 0 &  i \\
    0  & -i\tilde{\be}&  0 & 0 \\
    -i \tilde{b} &  i \left(-\tilde{c}^2+\frac{r_1}{r_2} d\tilde{d}-\frac{s_2}{r_2}-z\right) & -\tilde{d} & -\tilde{c}
  \end{pmatrix}, \\
  W &= \diag(z, -z, z, -z)+2
  \begin{pmatrix}
    0 & -b & 0 & - i d\\
    - b &0 & i \tilde{d} & 0 \\
    0 & - if & 0 & -\tilde{\be} \\
    i\tilde{f} & 0 & -\be & 0 
  \end{pmatrix}.
  \end{aligned}
\end{equation}

\end{document}